%% file: spectrum-SPP-v7.tex
\documentclass[12pt]{article}
\usepackage[utf8]{inputenc}
\usepackage{graphicx}
\graphicspath{ {./figure/} }
\usepackage{amsmath, amssymb, amsthm}
\usepackage{hyperref}
\usepackage{setspace}
\usepackage{mathtools}
\usepackage[round]{natbib}
\usepackage{multirow} 
\usepackage{booktabs}
\usepackage{caption}
\usepackage{authblk}

\usepackage{subcaption}
\numberwithin{equation}{section}
\newtheorem{proposition}{Proposition}[section] 
\newtheorem{theorem}{Theorem}[section]
\newtheorem{lemma}{Lemma}[section] 
\newtheorem{assumption}{Assumption}[section] 
\newtheorem{definition}{Definition}[section] 
\newtheorem{corollary}{Corollary}[section] 

\newtheorem{remark}{Remark}[section]

\newcommand{\C}{\mathbb{C}}

\newcommand{\N}{\mathbb{N}}
\newcommand{\R}{\mathbb{R}}

\newcommand{\Z}{\mathbb{Z}}

\newcommand{\Dcon}{\stackrel{\mathcal{D}}{\rightarrow}}
\newcommand{\Pcon}{\stackrel{\mathcal{P}}{\rightarrow}}

\newcommand{\ubb}{\boldsymbol{u}}
\newcommand{\vbb}{\boldsymbol{v}}

\newcommand{\rb}{\boldsymbol{r}}
\newcommand{\ob}{\boldsymbol{\omega}}

\newcommand{\btheta}{\boldsymbol{\theta}}
\newcommand{\bbeta}{\boldsymbol{\beta}}
\newcommand{\balpha}{\boldsymbol{\alpha}}
\newcommand{\bOmega}{\boldsymbol{\Omega}}

\newcommand{\bxi}{\boldsymbol{\xi}}
\newcommand{\kb}{\boldsymbol{k}}

\newcommand{\xb}{\boldsymbol{x}}
\newcommand{\yb}{\boldsymbol{y}}
\newcommand{\zb}{\boldsymbol{z}}

\newcommand{\tb}{\boldsymbol{t}}

\newcommand{\bb}{\boldsymbol{b}}

\newcommand{\aB}{\boldsymbol{A}}

\newcommand{\cov}{\mathrm{cov}}
\newcommand{\var}{\mathrm{var}}
\newcommand{\cum}{\mathrm{cum}}
\newcommand{\Ex}{\mathbb{E}}
\newcommand{\diag}{\mathrm{diag}}

\title{Pseudo-spectra of multivariate inhomogeneous spatial point processes}

\author[1]{Qi-Wen Ding\thanks{The first two authors are ordered alphabetically.}
}
\author[1]{Junho Yang
}
\author[2]{Joonho Shin\thanks{\textit{Emails}: \texttt{\{qwding,junhoyang\}@stat.sinica.edu.tw, joonho.shin@sungshin.ac.kr}}
}

\affil[1]{Institute of Statistical Science, Academia Sinica}
\affil[2]{School of Mathematics, Statistics and Data Science, Sungshin Women's University}

\date{\today}

\begin{document}

\maketitle

\begin{abstract}

In this article, we propose a spectral method for a class of multivariate inhomogeneous spatial point processes, namely the second-order intensity reweighted stationary processes. A key ingredient of our approach is utilizing the asymptotic behavior of the periodogram. For second-order stationary point processes, the periodogram is an asymptotically unbiased estimator of the spectrum. By calculating the moment, we show that under inhomogeneity, the expectation of the periodogram converges to a matrix-valued function, which we refer to as the pseudo-spectrum. The pseudo-spectrum shares similar properties with the spectrum of stationary processes and admits interpretation in terms of local parameters. We derive a consistent nonparametric estimator of the pseudo-spectrum via kernel smoothing and propose two bandwidth selection methods. The performance and utility of the proposed methods are demonstrated through simulation studies and an application to rainforest point pattern data.

\vspace{0.5em}


\end{abstract}

\section{Introduction}

As an alternative to time and spatial domain approaches, frequency domain methods have been extensively studied for time series and random fields. For point pattern data, the spectral analysis of spatial point processes dates back to the seminal work of Bartlett in the 1960s (\cite{p:bar-64}), who defined the spectrum of two-dimensional point processes and proposed using the periodogram as its estimator. Nevertheless, analogous frequency domain tools for spatial point processes remain largely underdeveloped, due to the technical difficulties in handling point data that are scattered irregularly in space. In recent years, with advancements in techniques for analyzing irregularly spaced data (e.g., \cite{p:mat-09}), some theoretical progress on discrete Fourier transforms (DFTs) and periodograms for spatial point patterns have been established. For example, \cite{p:raj-23} calculated the moments of the DFTs and periodograms at fixed frequencies and \cite{p:yan-24} extended these results by considering frequencies that vary with $n$ and proved new $\alpha$-mixing CLTs for the integrated periodogram.

The aforementioned approaches assume that the data-generating point process is second-order stationary (SOS), where the spectrum is well-defined via the inverse Fourier transform of the complete covariance function with the reduced form. However, in applications, it is more realistic to assume that the underlying process is inhomogeneous over the domain. This violates the SOS condition, in turn, the spectrum is not well-defined. The same issue also arises in other types of spatial data. In geostatistics, a common remedy is to remove the spatial trends using appropriate de-trending techniques such as universal kriging (cf. \cite{b:cre-15}), then apply SOS models to the residuals. Such methods, however, are not applicable to point pattern data, necessitating the development of unique frequency domain tools for inhomogeneous spatial point processes.

Therefore, to keep pace with the increasing demand for analyzing inhomogeneous point pattern data, we aim to extend the concept of the spectrum for SOS point processes to a widely used class of multivariate inhomogeneous point processes, namely the second-order intensity reweighted stationary processes (SOIRS; \cite{p:bad-00}). Specifically, by utilizing the asymptotic behavior of the periodogram, we define the so-called pseudo-spectrum for SOIRS processes (Section \ref{sec:spect}). The pseudo-spectrum possesses properties similar to those of the spectrum for SOS processes, such as being conjugate symmetric and positive definite. More intriguingly, the pseudo-spectrum can be represented as a weighted integral of the spectra of certain SOS processes with a localization factor, enabling the extraction of frequency domain features for inhomogeneous point processes (Section\ref{sec:local}).

Moving on, we investigate the estimation of the pseudo-spectrum. In Section \ref{sec:feasibleDFT}, we propose a nonparametric estimator of the pseudo-spectrum via kernel smoothing of the periodogram. In Section \ref{sec:asymp}, we provide the asymptotic properties of our estimator. It is worth emphasizing that the asymptotic properties are proven even in the case where the DFT is demeaned under an incorrectly specified first-order intensity, which is the first result showing the validity of the second-order estimator for inhomogeneous point processes under a misspecified model. In Section \ref{sec:bandwidth}, we propose two kernel bandwidth selection methods, namely the optimal MSE criterion and the frequency domain cross-validation criterion. We illustrate our proposed methods through simulations (Section \ref{sec:sim}) and a real-data application of five-variate rainforest point pattern data (Section \ref{sec:data}).

Note that our frequency domain estimation procedures are developed under a triangular array formulation of spatial point processes (see Definition \ref{def:infill} below), which is in the same spirit as the widely adopted asymptotic setting for nonstationary time series (cf. \cite{p:dah-97}). This framework has also been successfully extended to construct nonstationary (one-dimensional) Hawkes processes (\cite{p:rou-16, p:mam-23}) and nonstationary random fields (\cite{p:kur-22, p:lah-24}), the latter of which can be viewed as marked inhomogeneous Poisson point processes. It should be emphasized that our new asymptotic framework differs from the two classical regimes for spatial point processes: increasing domain asymptotics (\cite{p:gau-07}) and infill asymptotics (\cite{p:waa-08}). Among the three, our framework appears to be the only promising approach that guarantees the well-definedness of the frequency domain estimand with solid theoretical support.  

Lastly, technical details, proofs of the results, and supplements for the simulation and data analysis can be found in the Appendix.

\section{Pseudo-spectrum of inhomogeneous point process} \label{sec:spect}

\subsection{Preliminaries} \label{sec:intensity12}

Throughout this article, for a vector $\vbb = (v_1, \dots, v_d)^\top \in \C^{d}$, let $|\vbb| = \sum_{j=1}^{d} |v_j|$, $\|\vbb\| = \{\sum_{j=1}^{d} |v_j|^2\}^{1/2}$, and $\|\vbb\|_{\infty} = \max_{1\leq j \leq d} |v_j|$. For a matrix $A \in \C^{m \times n}$, $|A| = \sum_{i,j} |A^{(i,j)}|$ denotes the entrywise sum. For matrices $A = (A^{(i,j)})$ and $B = (B^{(i,j)})$ in $\C^{m \times n}$ (which include vectors by setting $n=1$), we define $A \odot B = (A^{(i,j)} B^{(i,j)}) \in \C^{m \times n}$ and $A \oslash B = (A^{(i,j)} / B^{(i,j)}) \in \C^{m \times n}$ to be the entrywise product and division respectively, provided $B^{(i,j)} \neq 0$. For $p \in [1,\infty)$, $L_{p}^{m \times n}(\R^d)$ denotes the set of all $m \times n$ matrix-valued measurable function $G$ on $\R^d$ such that $\int_{\R^{d}} |G(\xb)|^p d\xb <\infty$. For $G$ belonging to either $L_1^{m \times n}(\R^d)$ or $L_2^{m \times n}(\R^d)$, the Fourier transform and inverse Fourier transform are respectively defined as $\mathcal{F}(G)(\cdot) = \int_{\R^d} G(\xb) \exp(i \xb^\top \cdot) d\xb$ 
and $\mathcal{F}^{-1}(G)(\cdot) = (2\pi)^{-d} \int_{\R^d} G(\xb) \exp(-i \xb^\top \cdot) d\xb$. Lastly, $\Pcon$ denotes convergence in probability.

Now, let $\underline{X} = (X_1, \dots, X_m)$ be an $m$-variate simple spatial point process defined on $\R^d$. 
For $A \in \mathcal{B}(\R^{d})$, the Borel set of $\R^{d}$, let $\underline{N}(A) = (N_1(A), \dots, N_m(A))^{\top}$ be the count random vector that counts the number of events of $\underline{X}$ in $A$. 
We assume the first- and second-order cumulant moment measures of $\underline{X}$ exist and are absolutely continuous with respect to their corresponding Lebesgue measures. In other words, there exist vector-valued first-order intensity function and the matrix-valued covariance intensity function of $\underline{X}$, denoted as $\underline{\lambda}_1(\cdot) = (\lambda_1^{(1)}(\cdot),\cdots, \lambda_1^{(m)}(\cdot))^\top:\R^d \rightarrow \R^m$ and $\Gamma_2(\cdot, \cdot) = (\gamma_2^{(i,j)}(\cdot, \cdot))_{1\leq i,j \leq m}: \R^{d}\times \R^d \rightarrow \R^{m\times m}$, that satisfy
\begin{equation} \label{eq:lambda-vector}
\Ex[\underline{N}(A)] = \int_{A} \underline{\lambda}_1(\xb) d\xb 
\quad \text{and} \quad
\cov(\underline{N}(A), \underline{N}(B)) =  \iint_{A\times B} \Gamma_{2} (\xb,\yb) d\xb d\yb
\end{equation}
for any disjoint $A, B \in \mathcal{B}(\R^d)$. 
By using this notation, we define the ``complete'' covariance intensity function of $\underline{X}$ in the sense of \cite{p:bar-64} by
\begin{equation} \label{eq:complete-cov-mat}
\widetilde{C}(\xb,\yb) = (\widetilde{C}^{(i,j)} (\xb, \yb))_{1\leq i,j \leq m}
=\diag(\underline{\lambda}_1(\xb)) \delta(\xb-\yb) +
\Gamma_{2}(\xb,\yb),
\end{equation} 
where $\diag$ denotes the diagonal matrix and $\delta(\cdot)$ denotes the Dirac delta function. Then, $\widetilde{C}(\xb,\yb)$ satisfies $\cov(\underline{N}(A), \underline{N}(B)) =  \iint_{A\times B} \widetilde{C}(\xb,\yb) d\xb d\yb$ for any $A, B \in \mathcal{B}(\R^d)$ that do not necessarily have to be disjoint.

\subsection{A working example: second-order stationary case}
Now, we first assume that $\underline{X}$ is SOS and review the frequency domain representation of (\ref{eq:complete-cov-mat}). Under SOS framework, there exists $\underline{\lambda}_1= (\lambda_1^{(1)}, \dots, \lambda_1^{(m)})^\top$ and $\Gamma_{2,\text{red}}(\cdot) = (\gamma_{2,\text{red}}^{(i,j)}(\cdot))_{1\leq i,j \leq m}$ such that
\begin{equation*} 
\underline{\lambda}_1(\xb)  = \underline{\lambda}_{1}
\quad \text{and} \quad
\Gamma_{2}(\xb,\yb)  = \Gamma_{2,\text{red}}(\xb-\yb).
\end{equation*} 
We refer to $\Gamma_{2,\text{red}}(\cdot)$ as the reduced covariance intensity function of $\underline{X}$. Then, the complete covariance intensity in (\ref{eq:complete-cov-mat}) also has the following reduced form:
\begin{equation} \label{eq:cov-mat}
\widetilde{C}(\xb, \yb) = \diag(\underline{\lambda}_1) \delta(\xb-\yb)
+ \Gamma_{2,\text{red}}(\xb-\yb)=: C(\xb-\yb).
\end{equation} We refer $C(\cdot)$ to as the reduced complete covariance intensity function.
Provided that $\Gamma_{2,\text{red}} \in L_1^{m \times m}(\R^d)$, the matrix-valued spectrum of $\underline{X}$ is defined as the inverse Fourier transform of $C$:
\begin{equation} \label{eq:spectrum-mat}
F(\ob) = \mathcal{F}^{-1}(C)(\ob) =  
 (2\pi)^{-d} \diag(\underline{\lambda}_1)
+ \mathcal{F}^{-1}(\Gamma_{2,\text{red}})(\ob), \quad \ob \in \R^d.
\end{equation} 
Then, $F(\cdot)$ satisfies the usual conditions for the spectrum of multivariate time series, such as $F(\cdot)$ being conjugate symmetric and positive definite (see \cite{b:dal-03}, Chapter 8.6).

To estimate the spectrum of $\underline{X}$, we consider the DFT of the observed point pattern. Let $\{D_n\}$ be a sequence of common windows of $X_1, \dots, X_m$ of the form:
\begin{equation} \label{eq:Dn}
D_n = [-A_1/2,A_1/2] \times \cdots \times  [-A_d/2,A_d/2], \quad n \in \N,
\end{equation}
where, for each $j$, $\{A_{j} = A_j(n)\}_{n \in \N}$ denotes an increasing sequence of positive numbers. When computing the DFT, data tapering is often applied to mitigate the bias of the periodogram (cf. \cite{p:dah-87}). Let $h(\xb)$, $\xb \in \R^d$, be a non-negative data taper with support $[-1/2, 1/2]^d$, and let $H_{h,k} = \int_{[-1/2,1/2]^d} h(\xb)^k \, d\xb > 0$ for all $k \in \N$. Using this notation, the (tapered) DFT of $\underline{X}$ is defined as 
$\underline{\mathcal{J}}_{h,n}(\ob) = (\mathcal{J}_{h,n}^{(1)}(\ob), \dots, \mathcal{J}_{h,n}^{(m)}(\ob))^\top$, where
\begin{equation} \label{eq:mathcalJn}
\mathcal{J}_{h,n}^{(j)}(\ob) 
= (2\pi)^{-d/2} H_{h,2}^{-1/2} |D_n|^{-1/2} 
\sum_{\xb \in X_j \cap D_n}
h(\xb \oslash \aB) \exp(-i \xb^\top \ob).
\end{equation} 
Here, $\aB = (A_1, \dots, A_d)^\top$ denotes the vector of side lengths. Lastly, the periodogram is defined as
\begin{equation} \label{eq:In}
I_{h,n}(\ob) = \underline{J}_{h,n}(\ob)\,\underline{J}_{h,n}(\ob)^*,
\end{equation}
where $A^*$ denotes the Hermitian operator and
$\underline{J}_{h,n}(\cdot) = \underline{\mathcal{J}}_{h,n}(\cdot) - \Ex[\underline{\mathcal{J}}_{h,n}(\cdot)]$
is the centered DFT (we will elaborate on the expression of $\Ex[\underline{\mathcal{J}}_{h,n}(\ob)]$ in Section \ref{sec:feasibleDFT}). 

For the statistical properties, \cite{p:yan-24} (hereafter YG24) showed that under the increasing domain framework, the periodogram is an asymptotically unbiased estimator of the spectrum:
\begin{equation} \label{eq:I-SOS}
\lim_{n \to \infty} \Ex[I_{h,n}(\ob)] = F(\ob), \quad \ob \in \R^d.
\end{equation}


\subsection{The pseudo-spectrum under a new asymptotic framework} \label{sec:pseudo}

In this section, we extend the frequency domain representation of SOS processes to inhomogeneous processes, where the first-order intensity $\underline{\lambda}_1(\cdot)$ may vary with location. In this scenario, the complete covariance intensity does not admit the reduced form described in (\ref{eq:cov-mat}); in turn, the spectrum is not well-defined. If, however, there is no structure imposed on the intensity functions, it remains unclear what the periodogram estimates in the inhomogeneous case. In this context, we confine our attention to a class of SOIRS processes proposed by \cite{p:bad-00}, which is a widely used inhomogeneous point process model in the literature. Specifically, we say $\underline{X}$ is an $m$-variate SOIRS process if there exists $L_2(\cdot) = (\ell_2^{(i,j)}(\cdot))_{1 \leq i,j \leq m}$ such that
\begin{equation} \label{eq:IRS-cov}
 \Gamma_{2}(\xb,\yb) = \diag(\underline{\lambda}_{1}(\xb)) \, L_2(\xb-\yb) \, \diag(\underline{\lambda}_{1}(\yb)). 
\end{equation} 
An equivalent definition is that there exists an $m$-variate SOS process $\underline{\widetilde{X}}$ whose first-order intensity and reduced covariance intensity are given by $(1, \dots, 1)^\top$ and $L_2(\cdot)$. Here, we refer to $\underline{\widetilde{X}}$ as the intensity reweighted process of $\underline{X}$. See Appendix \ref{sec:construction} for the detailed construction.

Now, we propose a counterpart concept of the spectrum. A key idea is to leverage the limiting behavior of the periodogram. To do so, we introduce a new asymptotic framework, which is built upon a class of triangular arrays of random fields by \cite{p:hal-94} (see also \cite{p:lah-03}). This triangular array formulation underpins the concept of locally stationary time series introduced by Dahlhaus (\cite{p:dah-97}), a widely used class of nonstationary time series that can be locally approximated by stationary ones. See also \cite{p:rou-16, p:mam-23} for an extension of the concept of local stationarity to a nonstationary Hawkes point process.

Below, we formally define a locally stationary SOIRS process.

\begin{definition}[Locally stationary SOIRS process] \label{def:infill}
Let $\{\underline{X}_{D_n}\}$ ($n\in \N$) be a sequence of $m$-variate spatial point processes, where for $n\in \N$, $\underline{X}_{D_n}$ denotes a spatial point process observed in $D_n \subset \R^{d}$ of the form (\ref{eq:Dn}). Let $\underline{\lambda}_{1,n}(\cdot)$ and $\Gamma_{2,n}(\cdot,\cdot)$ be the first-order intensity and covariance intensity function of $\underline{X}_{D_n}$. We say $\{\underline{X}_{D_n}\}_{}$ is locally stationary SOIRS process if it satisfies the following two structural assumptions:
\begin{itemize}
\item[(i)] There exists a reduced covariance intensity function $L_{2}(\cdot) = (\ell_2^{(i,j)}(\cdot))_{1 \leq i,j \leq m}$ that does not depend on $n \in \N$ such that
\begin{equation*}
\Gamma_{2,n}(\xb,\yb) = \diag(\underline{\lambda}_{1,n}(\xb)) L_2(\xb-\yb)  \diag(\underline{\lambda}_{1,n}(\yb)), \quad n \in \N,~~ \xb, \yb \in D_n.
\end{equation*}

\item[(ii)] There exists a non-negative function $\underline{\lambda}(\xb) = (\lambda^{(1)}(\xb), \dots, \lambda^{(m)}(\xb))^\top$, $\xb \in \R^d$, with support $[-1/2,1/2]^d$ such that
\begin{equation}  \label{eq:lambda-infill}
\underline{\lambda}_{1,n}(\xb) = \underline{\lambda}(\xb \oslash \aB), \quad n \in \N, \quad \xb \in D_n.
\end{equation}
\end{itemize}
\end{definition}
The first condition above yields that, as the domain increases, we gain information about the common reweighted (reduced) covariance intensity $L_2(\cdot)$.  
The second condition states that the first-order intensity function has a regular structure which ensures that we can also learn the structure of the first-order intensity as $n \rightarrow \infty$. 
As a special case, when $\underline{\lambda}(\xb)$ is constant over $\xb \in [-1/2,1/2]^d$, Definition \ref{def:infill} incorporates the usual increasing domain framework for SOS point processes. 
In general, both conditions above are essential for constructing the frequency domain estimand and the periodogram-based asymptotic theory for inhomogeneous processes.

Now, let 
$H_{h^2 \underline{\lambda}}  = (H_{h^2\lambda^{(1)},1}, \dots, H_{h^2\lambda^{(m)},1})^\top$ and 
$H_{h^2 \underline{\lambda}\cdot\underline{\lambda}^\top} = (H_{h^2 \lambda^{(i)} \lambda^{(j)},1})_{1 \leq i,j \leq m}$.
Then, in Theorem \ref{thm:asymDFT-IRS} in the Appendix, we show that under the above asymptotic framework, we have
\begin{equation*}
\lim_{n \rightarrow \infty} \Ex[I_{h,n}(\ob)] = 
(2\pi)^{-d} H_{h,2}^{-1} \diag (H_{h^2 \underline{\lambda}})
+ H_{h,2}^{-1} \left( H_{h^2 \underline{\lambda}\cdot\underline{\lambda}^\top} \odot
\mathcal{F}^{-1}(L_2)(\ob) \right).
\end{equation*}  
This indicates that even though the spectrum is not well-defined for an SOIRS process, the expectation of the periodogram still converges to a function with a closed-form expression. Therefore, using the above result in conjunction with (\ref{eq:I-SOS}), we now define the pseudo-spectrum.

\begin{definition}\label{defin:pseudo}
Let $\{\underline{X}_{D_n}\}$ be an $m$-variate locally stationary SOIRS process. Suppose further that $L_2(\cdot)$ in Definition \ref{def:infill}(i) belongs to $L^{m \times m}_1(\R^d)$. Then, the pseudo-spectrum of $\{\underline{X}_{D_n}\}$ corresponding to the data taper $h$ is defined as
\begin{equation} \label{eq:f-IRS-mat}
F_{h}(\ob) =  (2\pi)^{-d} H_{h,2}^{-1} \diag (H_{h^2 \underline{\lambda}})
+ H_{h,2}^{-1} \left( H_{h^2 \underline{\lambda}\cdot\underline{\lambda}^\top} \odot
\mathcal{F}^{-1}(L_2)(\ob) \right).
\end{equation} 
\end{definition}
We refer the reader to Proposition \ref{prop:Fh} in the Appendix for mathematical properties of $F_h$.

\begin{remark}[Comparison between the spatial domain and frequency domain estimands] \label{rmk:Comparision}
The $K$-function of the intensity reweighted process $\underline{\widetilde{X}}$ (often called the inhomogeneous $K$-function) is a widely used spatial domain estimand for SOIRS processes. A nonparametric estimator of the inhomogeneous $K$-function is given in Equation (5) of \cite{p:bad-00}, and under the increasing domain framework, the consistency of this estimator can be established. However, if we rely only on the increasing domain framework, a similar estimation approach does not work in the frequency domain. See Appendix \ref{sec:K-function} for detailed explanations.

The main distinction between the spatial domain and frequency domain estimands is that the former depends solely on the reweighted covariance intensity $L_2(\cdot)$, whereas the latter is a nonseparable function of both the first-order intensity $\underline{\lambda}_1(\cdot)$ and $L_2(\cdot)$. Therefore, to obtain a consistent estimator in the frequency domain, it is essential to learn the structure of both $\underline{\lambda}_{1}(\cdot)$ and $L_2(\cdot)$ as $n \rightarrow \infty$, a requirement not satisfied by either increasing domain or infill asymptotics.
\end{remark}

\section{Interpretation of the pseudo-spectrum} \label{sec:local}

\subsection{The local representation} \label{sec:Tri}
Note that our new asymptotic framework in Definition \ref{def:infill} is neither an infill asymptotics nor an increasing domain asymptotics. Specifically, as $n \rightarrow \infty$, the infill asymptotic framework assumes that the process is observed more densely on a fixed finite area, whereas under the increasing domain asymptotics, $\{D_n\}$ denotes observational windows that grow with $n$. However, it is still unwieldy to interpret the role of $\{D_n\}$ in our asymptotic regime and to understand why it is referred to as the ``locally stationary'' asymptotics.

To develop these ideas, we first focus on the univariate complete covariance intensity $\widetilde{C}_{n}(\xb,\yb)$ of $X_{D_n}$. 
By using the reparametrization $(\ubb,\rb)= (\xb \oslash \aB,\xb-\yb)$, we obtain
\begin{eqnarray*}
\widetilde{C}_n(\xb,\yb) &=& \lambda_{1,n}(\xb) \delta (\rb) + \gamma_{2,n}(\xb,\xb-\rb) \\ 
&=& \lambda(\ubb) \delta (\rb) + \lambda(\ubb) \lambda(\ubb-(\rb \oslash \aB)) \ell_2(\rb) =: \widetilde{C}_n^{\ubb} (\rb).
\end{eqnarray*} 
Suppose $\lambda(\cdot)$ is smooth on $[-1/2,1/2]^d$. Then, for a fixed $\ubb \in [-1/2,1/2]^d$, $\widetilde{C}_n^{\ubb}(\rb)$ can be closely approximated by $\lambda(\ubb) \delta(\rb) + \{\lambda(\ubb)\}^2 \ell_2(\rb)$ for spatial lag $\rb$ satisfying $\|\rb \oslash \aB\| = o(1)$.

A similar argument extends to the multivariate case. That is, for sufficiently large $n \in \N$,
\begin{equation} \label{eq:local-approx}
(\widetilde{C}_n^{\ubb})^{(i,j)}(\cdot) \approx 
\lambda^{(i)}(\ubb) \delta(\cdot) \delta_{i,j} + \lambda^{(i)}(\ubb) \lambda^{(j)}(\ubb)\ell_2^{(i,j)}(\cdot), 
\quad i,j \in \{1, \dots, m\},
\end{equation}
where $\delta_{i,j} = 1$ if $i=j$ and zero otherwise. Intriguingly, in Theorem \ref{thm:local-S} below, we show that the right-hand side of \eqref{eq:local-approx} is an $(i,j)$th component of a valid reduced complete covariance intensity function of an SOS process indexed by a localization factor $\ubb \in [-1/2,1/2]^d$. This motivates the introduction of a ``local'' version of the complete covariance intensity function and spectrum.

\begin{definition}
Let $\{\underline{X}_{D_n}\}$ be an $m$-variate locally stationary SOIRS process. Then, the local complete covariance intensity function at $\ubb \in [-1/2,1/2]^d$ is defined as
\begin{equation} \label{eq:cov-local}
C^{\ubb}(\xb) = \diag (\underline{\lambda}(\ubb)) \delta(\xb) 
+ \left( \underline{\lambda}(\ubb)\underline{\lambda}(\ubb)^\top \right) \odot L_2(\xb).
\end{equation}
Suppose further that $L_2 \in L^{m \times m}_1(\R^d)$. Then, the local spectrum is defined as
\begin{equation}\label{eq:F-local}
F^{\ubb}(\ob) = \mathcal{F}^{-1}(C^{\ubb})(\ob)
= (2\pi)^{-d} \diag (\underline{\lambda}(\ubb))  
+ \left( \underline{\lambda}(\ubb)\underline{\lambda}(\ubb)^\top \right) \odot \mathcal{F}^{-1}(L_2)(\ob). 
\end{equation}
\end{definition}
By this definition, a sequence of SOIRS processes in Definition \ref{def:infill} is locally stationary in the sense that the complete covariance function $\widetilde{C}_n(\xb, \yb)$ of $\underline{X}_{D_n}$ can be locally approximated by the  reduced complete covariance function $C^{\xb \oslash \aB} (\xb - \yb)$ of SOS process. In addition, if $\lambda^{(j)}(\cdot)$ is Lipschitz continuous on $[-1/2,1/2]^d$, then
\begin{equation*} 
\big| \widetilde{C}_n(\xb, \yb) - C^{\xb \oslash \aB} (\xb - \yb) \big| 
\leq C \sum_{j=1}^{m} \left| \lambda^{(j)} (\yb \oslash \aB) - \lambda^{(j)} (\xb \oslash \aB) \right| 
= O\!\left(\| (\xb - \yb)\oslash \aB \| \wedge 1 \right)
\end{equation*}
uniformly in $\xb, \yb \in D_n$, where $a \wedge b = \min(a, b)$. This shows that the (local) approximation error bound depends how fast $D_n$ expands. The same error rate arises in the autocovariance function of locally stationary random fields by \cite{p:lah-24}. Therefore, in our framework, the sequence $\{D_n\}$ not only servers as the observation windows, as in the traditional increasing domain setting, but also determines how closely the complete covariance function can be approximated by the corresponding complete covariance function of the local SOS process.

\subsection{Frequency domain features for inhomogeneous processes} \label{sec:feature}

In YG24, Section 2.3, we showed that the spectrum of an SOS process contains some global features of the point process, such as the overall clustering/repulsive behavior and the amount of clustering/repulsiveness. However, it remains unclear how the pseudo-spectrum characterizes the behavior of the corresponding SOIRS process. To address this, the theorem below represents $F_{h}(\cdot)$ in terms of the weighted integral of the local spectra.

\begin{theorem}  \label{thm:local-S}
Let $\{\underline{X}_{D_n}\}$ be an $m$-variate locally stationary SOIRS process. Suppose further that $L_2(\cdot) \in L^{m \times m}_1(\mathbb{R}^d)$.  
Then, for every $\ubb \in [-1/2, 1/2]^d$, there exists an $m$-variate SOS process $\underline{X}^{\ubb}$ such that the spectrum of $\underline{X}^{\ubb}$ is given by $F^{\ubb}(\cdot)$. Moreover, $F_h(\cdot)$ and $\{F^{\ubb}(\cdot)\}$ are related by the following identity:
\begin{equation*}
F_h(\cdot) = \frac{1}{H_{h,2}} \int_{[-1/2,1/2]^d} h(\ubb)^2 F^{\ubb}(\cdot) \, d\ubb.
\end{equation*}
\end{theorem}
\begin{proof} 
See Appendix \ref{sec:localS}.
\end{proof}

The above theorem yields that we can also extract some features of the inhomogeneous process from the pseudo-spectrum. Specifically, let $\text{Fea}(\ubb)$ be the frequency domain feature of $\underline{X}^{\ubb}$ (as in Theorem \ref{thm:local-S}) that can be extracted from its spectrum $F^{\ubb}$. Since $F^{\ubb}(\cdot)$ approximates the spectral behavior of the point process in the neighborhood of $\xb = \ubb \odot \aB$, Theorem \ref{thm:local-S} yields that the pseudo-spectrum $F_{h}$ has the same feature $\{\underline{X}_{D_n}\}$ that can be quantified as an averaged effect of all features $\text{Fea}(\ubb)$ on $\ubb \in [-1/2,1/2]^d$, with weights proportional to $h^2(\ubb)$. For example, by applying Theorem \ref{thm:local-S} with $h \equiv 1$ in $[-1/2,1/2]^d$ and using (\ref{eq:f-IRS-mat}), we obtain
\begin{equation*}
F_0^{(j,j)}(\ob) - (2\pi)^{-d} H_{\lambda^{(j)},1} = 
\int_{[-1/2,1/2]^d} \left\{ (F^{\ubb})^{(j,j)}(\ob) - (2\pi)^{-d} \lambda^{(j)}(\ubb) \right\} d\ubb.
\end{equation*}
Here, $F_0$ denotes the pseudo-spectrum for a unit taper function.
Following the interpretation in YG24, $(F^{\ubb})^{(j,j)}(\ob) - (2\pi)^{-d} \lambda^{(j)}(\ubb) > 0$ (resp., $< 0$) around the frequencies $\ob$ near the origin indicates that $X_j$ exhibits clustering (resp., repulsive) behavior in the local area near $\xb = \ubb \odot \aB$. Therefore, $F_0^{(j,j)}(\ob) - (2\pi)^{-d} H_{\lambda^{(j)},1} > 0$ (resp., $< 0$) at low frequencies implies that $X_j$ is ``overall'' clustered (resp., repelled) across the entire domain. We refer the reader to Section \ref{sec:data} for the additional frequency domain features through the real data example.


\section{Nonparametric estimates of the pseudo-spectrum} \label{sec:kernel}

In this section, we propose a nonparametric estimator of $F_{h}$ and investigate its large sample properties. Here, we assume the data taper $h$ is known, so we do not study the choice of the data taper function. However, the asymptotic results provided in this section hold for a very general class of data tapers and our simulations show that the choice of the data taper function does not significantly affect the performance of the estimator.

\subsection{Estimation of the bias under the possible misspecified first-order intensity model} \label{sec:feasibleDFT}
For the non-negative function $h$ with support $[-1/2,1/2]^d$ and for $k \in \N$, let
\begin{equation}  \label{eq:Hkn}
H_{h,k}^{(n)}(\ob) = \int_{D_n} h(\xb \oslash \aB)^k \exp(-i \xb^\top \ob) \, d\xb.
\end{equation}

Then, by using the celebrated Campbell's formula (e.g., \cite{p:zhu-25}, Equation (B.6)), the expectation of the DFT for a locally stationary SOIRS process is given by
\begin{equation} \label{eq:J-expectation}
\Ex[\underline{\mathcal{J}}_{h,n}(\ob)] = (2\pi)^{-d/2} H_{h,2}^{-1/2} |D_n|^{-1/2} H_{h \underline{\lambda}}^{(n)}(\ob),
\end{equation}
where $H_{h \underline{\lambda}}^{(n)}(\ob) = ( H_{h \lambda^{(1)},1}^{(n)}(\ob), \dots, H_{h \lambda^{(m)},1}^{(n)}(\ob))^\top$.
Therefore, the centered DFT and the theoretical periodogram are respectively defined as
\begin{equation}
\begin{aligned}
\underline{J}_{h,n}(\ob) &= \underline{\mathcal{J}}_{h,n}(\ob) - (2\pi)^{-d/2} H_{h,2}^{-1/2} |D_n|^{-1/2} H_{h \underline{\lambda}}^{(n)}(\ob) \quad \text{and} \\
I_{h,n}(\ob) &= \big(I_{h,n}^{(i,j)}(\ob)\big)_{1 \leq i,j \leq m}
= \underline{J}_{h,n}(\ob) \, \underline{J}_{h,n}(\ob)^* .
\end{aligned}
\label{eq:In-IRS}
\end{equation}

Since $H_{h\underline{\lambda}}^{(n)}(\cdot)$ depends on the unknown first-order intensities, the above quantities need to be estimated. To do so, we consider a parametric first-order intensity model of the form $\underline{\lambda}(\cdot;\bbeta) = (\lambda^{(1)}(\cdot;\bbeta), \dots, \lambda^{(m)}(\cdot;\bbeta))^\top$, $\bbeta \in \Theta \subset \mathbb{R}^p$.
Here, we do not necessarily assume that there exists $\bbeta_0 \in \Theta$ such that the true first-order intensity is $\underline{\lambda}(\cdot; \bbeta_0)$. Instead, we assume that we can compute $\widehat{\bbeta}_n$ based on the observed point pattern $\underline{X}_{D_n}$, and that $\widehat{\bbeta}_n$ converges in probability to some pseudo-true value $\bbeta_0 \in \Theta$. Under this assumption, our estimator of $H_{h\underline{\lambda}}^{(n)}(\ob)$ is 
$\widehat{H}_{h\underline{\lambda}}^{(n)}(\ob) = (H_{h \widehat{\lambda}^{(1)},1}^{(n)}(\ob), \dots, H_{h \widehat{\lambda}^{(m)},1}^{(n)}(\ob))^{\top}$,
where $\widehat{\lambda}^{(j)}(\cdot) = \lambda^{(j)}(\cdot; \widehat{\bbeta}_n)$. Finally, the feasible criteria of $\underline{J}_{h,n}(\ob)$ and $I_{h,n}(\ob)$, denoted as $\underline{\widehat{J}}_{h,n}(\ob)$ and $\underline{\widehat{I}}_{h,n}(\ob)$, are defined similarly to (\ref{eq:In-IRS}), but with $\widehat{H}_{h\underline{\lambda}}^{(n)}(\ob)$ replacing $H_{h\underline{\lambda}}^{(n)}(\cdot)$ in the equations.

Below, we discuss possible scenarios for the parametric models and their estimates.

\noindent \textit{Correctly specified model}. We assume there exists $\bbeta_0 \in \Theta$ such that the true first-order intensity function is $\underline{\lambda}(\xb;\bbeta_0)$. An estimator of $\bbeta_0$, denoted as $\widehat{\bbeta}_n$, and its large sample properties were studied in \cite{p:gau-07}.

\vspace{0.3em}

\noindent \textit{Homogeneous intensity model}. We consider the constant intensity model $\underline{\lambda}(\xb;\bbeta) = \bbeta$ for $\bbeta \in \Theta \subset [0,\infty)^{m}$. The two estimation schemes that fall within this framework are (1) $\widehat{\bbeta}_n \equiv \bbeta_0 \in \Theta$ ($n\in\N$) for some pre-determined constant vector $\bbeta_0 \in \Theta$, and (2) $\widehat{\bbeta}_n = |D_n|^{-1} \underline{N}(D_n)$. The former includes the case where $\bbeta_0 = (0,\dots, 0)^\top$, so the feasible periodogram in this case becomes $\underline{\mathcal{J}}_{h,n}(\ob) \underline{\mathcal{J}}_{h,n}(\ob)^{*}$. Moreover, under mild moment conditions, the latter estimator satisfies $\widehat{\bbeta}_n \Pcon (H_{\lambda^{(1)},1}, \dots, H_{\lambda^{(m)},1})^{\top} =: \bbeta_0$.

\vspace{0.3em}

\noindent \textit{General misspecified model}. Under a general misspecified model, \cite{p:cho-21} proposed a composite Bayesian information criterion and fit a regression model to the first-order intensity function. Under mild conditions, they showed that the estimator $\widehat{\bbeta}_n$ consistently estimates some $\bbeta_0$, where $\bbeta_0$ has an interpretation in terms of the information criterion.

Under appropriate assumptions to be stated in the next section, the feasible periodogram under correctly specified intensity satisfies the following assertions: for any sequence of frequencies $\{\ob_n\}$ with $\lim_{n\rightarrow \infty} \ob_n = \ob$ and $\lim_{n\rightarrow \infty} \|\aB \odot \ob_n\| = \infty$ (we say, $\{\ob_n\}$ is asymptotically distant from the origin), we have
\begin{equation} \label{eq:Per-thm}
\text{(a)} \lim_{n\rightarrow \infty} \Ex[\underline{\widehat{I}}_{h,n}(\ob_n)] = F_{h}(\ob) \quad \text{and} \quad
\text{(b)} \lim_{n\rightarrow \infty} \var(\underline{\widehat{I}}_{h,n}^{(i,j)}(\ob_n)) = |F_{h}^{(i,j)}(\ob)|^2.
\end{equation} 
See Theorem \ref{thm:asymDFT-IRS} in the Appendix for the precise statements and proofs. These results indicate that the periodogram is asymptotically unbiased but not consistent. Therefore, to consistently estimate the pseudo-spectrum, we introduce a kernel smoothing of the periodogram, which is a common approach in time series (cf. \cite{b:bro-dav-06}). 

To this end, let $K: \R^{d} \rightarrow \R$ be a non-negative symmetric kernel function that satisfies: (i) $K(\xb)$ has support on $\xb \in [-1,1]^d$ and is continuous on $\R^d$, (ii) $\int_{\R^d} K(\xb) \, d\xb = 1$, and (iii) $\int_{\R^d} |K(\xb)|^2 \, d\xb < \infty$. Let $\bb = (b_1, \dots, b_d)^\top \in (0,\infty)^d$ be the bandwidth, and let $K_{\bb}(\cdot) = (b_1 \cdots b_d)^{-1} K(\cdot \oslash \bb)$. Then, our feasible nonparametric estimator of the pseudo-spectrum is defined as
\begin{equation} \label{eq:KSDE}
\widehat{F}_{n,\bb}(\cdot) = \int_{\R^d} K_{\bb}(\cdot - \xb) \widehat{I}_{h,n}(\xb) \, d\xb.
\end{equation}

\subsection{The large sample properties} \label{sec:asymp}

To derive the large sample properties of $\widehat{F}_{n,\bb}(\ob)$, we require the following sets of assumptions.

The first assumption is on the asymptotic setting.
\begin{assumption} \label{assum:A}
$\{\underline{X}_{D_n}\}$ is an $m$-variate locally stationary SOIRS process. Moreover, for $j \in \{1, \dots, d\}$, it holds that $\lim_{n \rightarrow \infty}A_{j}(n) = \infty$ and $\inf_{\xb \in [-1/2,1/2]^d} \lambda^{(j)}(\xb) > 0$.
\end{assumption}

The next assumption is on the integrability of the joint cumulant intensity functions of $\underline{X}_{D_n}$. For $\boldsymbol{\alpha} = (\alpha_1, \dots, \alpha_m)^\top \in \{0,1, \dots\}^m$, let $\gamma^{\balpha}_n: D_n^{|\balpha|} \rightarrow \R$ be the $\balpha$th-order joint cumulant intensity function of $\underline{X}_{D_n}$ (cf. \cite{p:zhu-25}, pages 1892--1893).
\begin{assumption} \label{assum:B}
Let $k \in \{2, 3, \dots \}$ be fixed. Then, $\gamma^{\balpha}_n$ is well-defined for any $\balpha \in \{0,1,\dots\}^m$ with $|\balpha| \leq k$. Moreover, for any $2 \leq |\balpha| \leq k$, it holds that
\begin{equation*}
\sup_{n \in \N} \sup_{\xb_{1} \in D_n}
\int_{D_n^{|\boldsymbol{\alpha}|-1}} \left| \gamma^{\balpha}_n(\xb_1, \xb_2, \dots, \xb_{|\balpha|})\right| d\xb_2 \cdots d\xb_{|\balpha|} < \infty.
\end{equation*} 
\end{assumption}
Assumptions \ref{assum:A} and \ref{assum:B} for $k=2$ imply $L_2(\cdot) \in L_1^{m \times m}(\R^d)$ ensuring the well-definedness of the pseudo-spectrum.

The next assumption is on the data taper.
\begin{assumption} \label{assum:C}
The data taper $h(\cdot)$ is non-negative, has support in $[-1/2,1/2]^d$, and is continuous in $[-1/2,1/2]^d$. 
\end{assumption}

Next, we assume conditions on the parametric form of $\underline{\lambda}(\cdot) = \underline{\lambda}(\cdot; \bbeta)$, $\bbeta \in \Theta$.
\begin{assumption} \label{assum:beta}
The parameter space $\Theta \subset \mathbb{R}^p$ is compact. For any $j\in \{1, \dots, m\}$, $\lambda^{(j)}(\xb;\bbeta)$ is twice partially differentiable with respect to $\bbeta \in \Theta$ and the partial derivatives are bounded above.
\end{assumption}

\begin{assumption} \label{assum:beta2}
Let $\widehat{\bbeta}_n$ be the estimator of the pseudo-true value $\bbeta_0$ such that for any $w \in [0,1]$, $w \widehat{\bbeta}_n + (1-w) \bbeta_0 \in \Theta$. Moreover, the difference $\widehat{\bbeta}_n - \bbeta_0$ satisfy one of the following two conditions:
\begin{itemize}
\item[(i)] $|D_n|^{1/2}\big| \widehat{\bbeta}_n - \bbeta_0 \big| = O_p(1)$, $n \rightarrow \infty$.
\item[(ii)] There exists $r \in (1,\infty)$ such that $ \Ex \big \{ |D_n|^{1/2}\big| \widehat{\bbeta}_n - \bbeta_0 \big| \big\}^r = O(1)$, $n \rightarrow \infty$.
\end{itemize}
\end{assumption}
Under mild assumptions the composite likelihood estimator in \cite{p:gau-07} (for correctly specified model) and information-based estimator in \cite{p:cho-21} (for misspecified model) satisfy Assumption \ref{assum:beta2}(i) and (ii).

The last assumption is on the bandwidth $\bb$ of $\widehat{F}_{n,\bb}(\ob)$.

\begin{assumption} \label{assum:D}
The bandwidth $\bb = \bb(n) = (b_1(n), \dots, b_d(n))^\top$ depends only on $n \in \N$ and satisfies
\begin{equation*}
\lim_{n\rightarrow \infty}\big\{  \|\bb(n)\|_{\infty} + \| 1/ (\aB(n) \odot \bb(n)) \|_{\infty} \big\} = 0.
\end{equation*}
\end{assumption}


The theorem below states that, provided the first-order intensity is correctly specified, $\widehat{F}_{n,\bb}(\ob)$ converges to the pseudo-spectrum in two different modes of convergence.

\begin{theorem} \label{thm:KSDE1}
Suppose that Assumptions \ref{assum:A}, \ref{assum:B} (for $k = 4$), \ref{assum:C}, \ref{assum:beta}, \ref{assum:beta2}(i), and \ref{assum:D} hold. Moreover, we assume that the true first-order intensity function is $\underline{\lambda}(\cdot;\bbeta_0)$, where $\bbeta_0$ is as in Assumption \ref{assum:beta2}. Then, for any $\ob \in \R^d$,
\begin{equation} \label{eq:Fnb-conv1}
\widehat{F}_{n,\bb}(\ob) \Pcon F_h(\ob).
\end{equation}
If we further assume Assumption \ref{assum:beta2}(ii) with $r > 4$, then 
\begin{equation} \label{eq:Fnb-conv2}
\lim_{n \to \infty} \Ex \big| \widehat{F}_{n,\bb}(\ob) - F_h(\ob) \big|^2 = 0.
\end{equation}
\end{theorem}
\begin{proof} See Appendix \ref{appen:proofKSDE}. 
\end{proof}

Surprisingly, even in the case where the first-order intensity is incorrectly estimated, the convergence results in Theorem \ref{thm:KSDE1} still hold for frequencies away from the origin.

\begin{corollary} \label{coro:KSDE2}
Suppose the same set of assumptions in Theorem \ref{thm:KSDE1} hold, except we now allow for misspecified first-order intensity model. Furthermore, we assume the side lengths $\aB(n)$ satisfy (SL) below and the smoothness conditions on $h(\xb)$, $\underline{\lambda}(\xb)$, and $\underline{\lambda}(\xb;\bbeta_0)$ as in Assumption \ref{assum:smooth} in the Appendix hold. Then, the convergence results (\ref{eq:Fnb-conv1}) and (\ref{eq:Fnb-conv2}) also hold for $\ob \in \R^d$ such that the number of non-zero elements of $\ob$, denotes $|\ob|_0$, is greater than $d/4$.
\end{corollary}
\begin{proof} See Appendix \ref{appen:proofKSDE2}.
\end{proof}

\begin{remark}\label{rmk:coroKSDE}
An intuition for the convergence results under the misspecified model is that the quantity $|\Ex[\mathcal{J}_{h,n}^{(j)}(\ob)]|^2$ in (\ref{eq:J-expectation}) acts like a Fej\'{e}r kernel (a discrete version of the Dirac delta function). Therefore, the contribution of the demeaning in $\widehat{F}_{n,\bb}(\ob)$ is negligible for frequencies away from the origin.
\end{remark}

\begin{remark} \label{rmk:near-origin}
As shown in Theorem \ref{thm:KSDE1}, under a correctly specified first-order intensity, $\widehat{F}_{n,\bb}^{(j,j)}(\mathbf{0}) = O_p(1)$. On the other hand, following the results in Corollary \ref{coro:KSDE2} and their proofs, it can be seen that if the first-order intensity is incorrectly specified, then $\Ex[\widehat{F}_{n,\bb}^{(j,j)}(\mathbf{0})]$ grows at a rate of $O(|D_n|)$. Note that the same dichotomy is also true for the feasible periodogram, whether the model is correctly or incorrectly specified. Therefore, the behavior of $\widehat{F}_{n,\bb}^{(j,j)}(\ob)$ (or, $\widehat{I}_{n}^{(j,j)}(\ob)$) near the origin can be informally used to assess the goodness of fit of the first-order intensity model $\lambda^{(j)}(\cdot;\bbeta)$. 
\end{remark}

\begin{remark}
In practice, we compute $\widehat{F}_{n,\bb}(\cdot)$ using the Riemann sum approximation
\begin{equation} \label{eq:KSDE2}
\widehat{F}^{(R)}_{n,\bb}(\cdot)
= \frac{\sum_{\kb \in \Z^d} K_{\bb}(\cdot - \xb_{\kb, \bOmega}) \widehat{I}_{h,n}(\xb_{\kb, \bOmega})}{
\sum_{\kb \in \Z^d} K_{\bb}(\cdot - \xb_{\kb, \bOmega})
},
\end{equation}
where, $\xb_{\kb, \bOmega} = 2\pi \kb \oslash \bOmega$ for some grid vector $\bOmega = (\Omega_1, \dots, \Omega_d)^\top$. We note that since $K_{\bb}(\cdot)$ has finite support, the summations in (\ref{eq:KSDE2}) are finite sums. By using similar techniques to those in the proof of Theorem \ref{thm:KSDE1}, one can also show the consistency and $L_2$-convergence of $\widehat{F}^{(R)}_{n,\bb}(\ob)$, provided the grid $\bOmega$ increases with $n$ and has the form $c \aB(n)$ for some constant $c \in (0,\infty)$. See Appendix \ref{appen:proofR} for the detailed arguments and proofs.
\end{remark}

\section{The bandwidth selection methods} \label{sec:bandwidth}

Although our kernel spectral density estimator $\widehat{F}_{n,\bb}(\cdot)$ consistently estimates the pseudo-spectrum for any bandwidth $\bb$ satisfying Assumption \ref{assum:D}, a poorly chosen bandwidth may lead to underperforming estimation. In this section, we propose two different approaches for selecting the bandwidth of $\widehat{F}_{n,\bb}(\cdot)$. For the sake of parsimony, we set $b_1 = \cdots = b_d = b$ for some $b \in (0,\infty)$.

\subsection{Method I: the optimal MSE criterion} \label{sec:opt}

The first approach is based on the optimal mean squared error (MSE) criterion. Let
\begin{equation} \label{eq:MSEb}
\text{MSE}(b) = \sup_{\ob \in \R^d} \Ex |\widehat{F}_{n,b}(\ob) - F_h(\ob)|^2,
\end{equation}
where $\widehat{F}_{n,b}(\ob)$ is the kernel spectral density estimator using the common bandwidth $b$ in all coordinates. To obtain the rate of convergence of $\text{MSE}(b)$ to zero, we require the following two additional assumptions.

First, we assume that the side lengths $\aB(n)$ of $D_n$ grow regularly in all coordinates:
\begin{equation} \tag{SL}
A_{i}(n) \propto A_{j}(n), \quad i,j \in \{1, \dots, d\},
\end{equation}
where for two positive sequences $\{a_n\}$ and $\{b_n\}$, $a_n \propto b_n$ means $0 < \inf_{n} (a_n/b_n) \leq \sup_{n} (a_n/b_n) < \infty$. The second assumption is on the pseudo-spectrum. 
\begin{assumption} \label{assum:E}
The pseudo-spectrum $F_h(\ob)$ satisfies (i) $F_h(\ob) - (2\pi)^{-d} H_{h,2}^{-1} \diag (H_{h^2 \underline{\lambda}}) \in L_1^{m \times m}(\R^d)$
and (ii) $F_h(\ob)$ is twice partially differentiable with respect to $\ob$, with second-order derivatives that are continuous and bounded above.
\end{assumption}

Then, thanks to Theorem \ref{thm:MSEb-rate} in the Appendix, when the first-order intensity is correctly specified, we have
\begin{equation} \label{eq:MSE-SPP}
\text{MSE}(b) = O(b^{4} + |D_n|^{-2/d} + |D_n|^{-1}b^{-d}).
\end{equation} 
Therefore, we choose $b = b_{\text{opt}}$, which minimizes the asymptotic order of the right-hand side above. In the theorem below, we provide the rates of convergence of $b_{\text{opt}}$ and $\text{MSE}(b_{\text{opt}})$ to zero (the proof immediately follows from Theorem \ref{thm:MSEb-rate}, so we omit the details).

\begin{theorem} \label{thm:opt-MSE}
Suppose that Assumptions \ref{assum:A}, \ref{assum:B} (for $k = 4$), \ref{assum:C}, \ref{assum:beta}, \ref{assum:beta2}(ii) (for $r=8$), \ref{assum:D}, and \ref{assum:E} hold. Moreover, we assume that the true first-order intensity is $\underline{\lambda}(\xb;\bbeta_0)$, the data taper $h$ is Lipschitz continuous on $[-1/2,1/2]^d$, and the side lengths $\aB(n)$ satisfy (SL). Then, for a dimension $d \leq 4$,
\begin{equation*}
b_{\text{opt}} \propto |D_n|^{-1/(d+4)} \quad \text{and} \quad \text{MSE}(b_{\text{opt}}) \propto |D_n|^{-4/(d+4)}.
\end{equation*} 
\end{theorem}

\begin{remark} \label{rmk:sub-opt}
Now, we consider the case where the first-order intensity function is incorrectly specified. In this case, $\text{MSE}(b) = \infty$ due to the large spike of $\widehat{F}_{n,b}(\mathbf{0})$ (see Remark \ref{rmk:near-origin}). Instead, for a fixed $\delta > 0$, let 
\begin{equation*}
\text{MSE}(b;\delta) = \sup_{\ob : \|\ob\|_{\infty} > \delta}  \Ex |\widehat{F}_{n,b}(\ob) - F_h(\ob)|^2.
\end{equation*}
Then, for $d \in \{1,2\}$, $\text{MSE}(b;\delta)$ also achieves the same asymptotic bound as in (\ref{eq:MSE-SPP}); in turn, the optimal bandwidth rate in Theorem \ref{thm:opt-MSE} is also near-optimal even for the misspecified case. The precise arguments and proofs can be found in Corollary \ref{coro:MSEb-rate} in the Appendix.
\end{remark}

\subsection{Method II: the cross-validation} \label{sec:CV}

From our simulation results in Section \ref{sec:sim} below, the pseudo-spectrum estimator based on the optimal bandwidth still suffers from large MSE for small sampling windows. Therefore, our second bandwidth selection criterion is based on a data-driven cross-validation method, motivated by \cite{p:bel-87} in the time series literature. Let $\bOmega = (\Omega_1, \dots, \Omega_d)^\top \in (0,\infty)^d$ be the grid vector, and let
\begin{equation} \label{eq:KSDE-minus}
\widehat{F}_{n,b}^{(-1)}(\ob_{\tb, \bOmega})
= \frac{\sum_{\kb \in \Z^d \setminus \{\tb\}} K_{b}(\ob_{\tb, \bOmega} - \ob_{\kb, \bOmega}) \widehat{I}_{h,n}(\ob_{\kb, \bOmega})}{
\sum_{\kb \in \Z^d  \setminus \{\tb\}} K_{b}(\ob_{\tb, \bOmega} - \ob_{\kb, \bOmega})
}, \quad \ob_{\tb,\bOmega} = 2\pi \tb \oslash \bOmega.
\end{equation}
In the above, $K_{b}(\cdot) = b^{-d} K(b^{-1} \cdot)$ denotes the scaled kernel function using the univariate bandwidth $b$. Since $\widehat{F}_{n,b}^{(-1)}(\cdot)$ excludes the center in the calculation of the local average of the periodograms, it can be viewed as a leave-one-out estimator of the pseudo-spectrum $F_h(\cdot)$. 

To select an appropriate bandwidth $b \in (0,\infty)$, we consider the following cross-validated spectral divergence between $\widehat{F}_{n,b}^{(-1)}$ and the periodogram $\widehat{I}_{h,n}$:
\begin{equation} \label{eq:Whittle-b}
L(b) = \sum_{\ob_{\tb,\bOmega} \in W} \left[ \textrm{Tr} \big\{ \widehat{I}_{h,n}(\ob_{\tb,\bOmega}) \widehat{F}_{n,b}^{(-1)}(\ob_{\tb,\bOmega})^{-1} \big\} + \log \det \widehat{F}_{n,b}^{(-1)}(\ob_{\tb,\bOmega}) \right].
\end{equation}
Here, $W \subset \R^{d}$ denotes a prespecified compact domain in $\R^d$. Then, our second proposed bandwidth selection criterion is based on minimizing the cross-validated spectral divergence:
\begin{equation} \label{eq:b-CV}
b_{\textrm{CV}} = \arg\min_{b \in (0,\infty)} L(b).
\end{equation} 
Heuristically, the selected $b_{\textrm{CV}}$ minimizes the Kullback–Leibler divergence between $\widehat{F}_{n,b}^{(-1)}$ and $\widehat{I}_{h,n}$ based on the information captured within $W$, where both are estimators of the pseudo-spectrum using non-overlapping frequencies. In our simulation study, the kernel spectral density estimator based on $b_{\textrm{CV}}$ outperforms that based on $b_{\textrm{opt}}$ uniformly over all domains. However, we do not investigate the statistical properties of $b_{\textrm{CV}}$ and $L(b_{\textrm{CV}})$ in this article.
 

\input{sims}


\input{bci}

\section{Concluding remarks}

In this article, we propose spectral methods for multivariate inhomogeneous point processes. The main thrust of our approach is to define an analogue of the spectrum of a stationary process for inhomogeneous processes, which we term the pseudo-spectrum. Under a new asymptotic framework, we develop a consistent nonparametric estimator of the pseudo-spectrum. Empirical studies indicate that our estimator performs well in finite samples and, in some cases, captures the features of point patterns more effectively than spatial domain methods.

This study focuses on the case where the number $m$ of different process types is fixed. For large (but fixed) $m$, our approach appears to offer computational advantages over spatial domain methods, such as computing the PCF matrix. This is because our method requires only $m$ single summations to compute the DFT, whereas spatial domain methods require computing the PCF for each pair of processes, which involves double summations. Extending our results to high-dimensional settings will require certain regularity conditions on the pseudo-spectrum or its inverse, which we leave for future research.

\section*{Acknowledgments}
QWD and JY acknowledge the support of the Taiwan’s National Science and Technology Council (grant 113-2118-M-001-012). QWD's research assistantship was partly supported by Academia Sinica's Career Development Award Grant (AS-CDA-114-M03). The authors thank Yongtao Guan for fruitful discussions and suggestions and Abdollah Jalilian for kindly share the BCI data used in Section \ref{sec:data} and Appendix \ref{sec:bci-add}.

\bibliography{bib-PSD}
\bibliographystyle{plainnat}

\pagebreak

\appendix
\counterwithin{figure}{section}
\counterwithin{table}{section}

\input{appendix1}

\input{sim-appendix}

\end{document}

%% file: sims.tex
\section{Empirical study} \label{sec:sim}

To verify our theoretical findings, we conduct some simulations. For the data-generating process, we use a bi-variate product-shot-noise Cox process considered in \cite{p:jal-15}, but with a modification of their model. Further details on the simulation settings and additional simulation results can be found in Appendix \ref{sec:sim-add}.

\subsection{The data-generating process} \label{sec:dgp}

Let $\underline{X} = (X_1, X_2)$ be a bi-variate Cox process on $D = [-A/2,A/2]^2$, where the corresponding latent intensity field $(\Lambda_1(\xb), \Lambda_2(\xb))^\top$ of $\underline{X}$ is given by
\begin{equation*}
\Lambda_i(\xb) = \lambda^{(i)}(\xb/A) S_i(\xb) Y_i(\xb), \quad \xb \in D.
\end{equation*}
Here, $\lambda^{(i)}(\xb/A) = \Ex[\Lambda_i(\xb)]$ denotes the first-order intensity of $X_i$ in alignment with the asymptotic framework in Definition \ref{def:infill}(ii).

Let $\Phi_1, \Phi_2, \Phi_3$ be independent homogeneous Poisson point processes. In our bi-variate model, $\Phi_1$ influences $X_1$ through the shot-noise field $S_1(\cdot)$, with an isotropic Gaussian kernel as its dispersal function. The same $\Phi_1$ also contributes to the compound field $Y_2$, which in turn influences the second process, $X_2$. Similarly, $\Phi_2$ affects $X_2$ directly through the shot-noise field $S_2$ and influences $X_1$ indirectly through the compound field $Y_1$. 

What distinguishes our model from that of \cite{p:jal-15} is the inclusion of a third latent parent process, $\Phi_3$. In our model, $\Phi_3$ affects both $X_1$ and $X_2$ indirectly through their respective compound fields, which provides a more realistic scenario for real-life applications. See Figure \ref{fig:simdiag}(a) for the schematic diagram of our model.

\begin{figure}[h]
	\centering
\begin{subfigure}{0.25\textwidth}
    \centering
    \includegraphics[width=\textwidth]{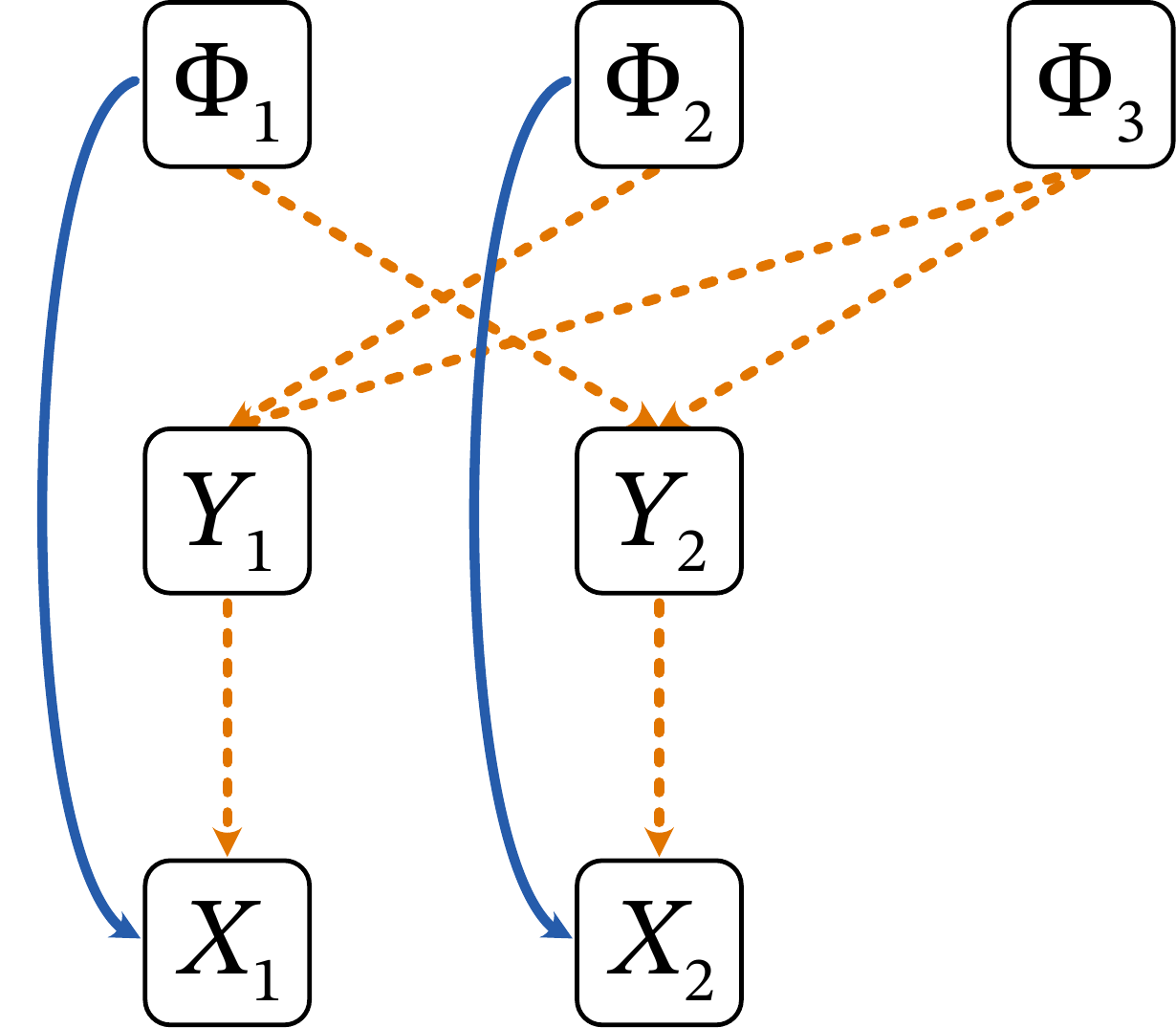}
    \caption{}
    \end{subfigure}
\begin{subfigure}{0.74\textwidth}
    \centering
    \includegraphics[width=0.95\textwidth]{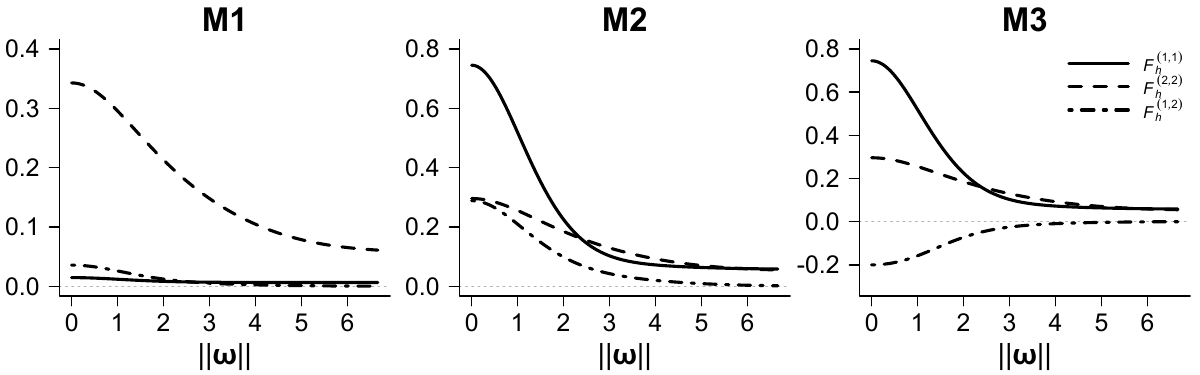}
    \caption{}
    \end{subfigure}
	\caption{
(a): A schematic diagram of our models. The solid lines refer to the direct effects through the shot-noise fields and the dashed lines refer to the indirect effects through the compound fields. 
(b): The marginal (solid and dashed lines) and cross (dot-dashed line) pseudo-spectrum plotted against $\|\ob\|$.
}
	\label{fig:simdiag}
\end{figure}

For the simulations, we consider the following three possible scenarios, namely M1--M3, based on the intensity function and the sign of interaction. Due to space constraints, we only depict their main features; detailed parameter settings can be found in Appendix \ref{sec:sim-detail}.
\begin{itemize}
    \item M1: Homogeneous isotropic process that exhibits inter-species clustering.
    \item M2: Inhomogeneous isotropic process that exhibits inter-species clustering.
    \item M3: Inhomogeneous isotropic process that exhibits inter-species repulsion.
\end{itemize}
Figure \ref{fig:simdiag}(b) plots the marginal and cross spectrum (for M1) and pseudo-spectrum (for M2 and M3). 
Note that since all three models above are isotropic, $F_h^{(i,j)}(\ob)$ are all real-valued and depend only on $\|\ob\|$. For convenience, we also refer to the spectrum of M1 as the pseudo-spectrum. 
Following the interpretation in Section \ref{sec:feature}, we see that $F^{(i,j)}$ lying above (resp. below) its asymptote line indicates that the underlying bi-variate process is marginally ($i=j$) or jointly ($i\neq j$) clustered (resp. repulsive).


\subsection{The pseudo-spectrum estimators} \label{sec:sim-est}

For each model above, we generate 500 replications of the bi-variate spatial point patterns within the observation window $D = [-A/2, A/2]^2$ for varying side lengths $A \in \{10, 20, 40\}$. To compare the different estimation approaches, for each replication, we evaluate three pseudo-spectrum estimators: (i) the raw periodogram ($\widehat{I}_{h,n}$), (ii) the kernel spectral density estimator using the optimal bandwidth ($\widehat{F}_{\text{opt}}$), and (iii) the kernel spectral density estimator using the cross-validation method ($\widehat{F}_{\text{CV}}$). Below, we discuss some practical issues that arise when calculating our estimators.

Firstly, when computing the DFT as in (\ref{eq:mathcalJn}), we use the separable data taper $h(\xb) = h_{a}(x_1) h_{a}(x_2)$ with $a=0.025$, where the form of $h_{a}(\cdot)$ is provided in Equation (\ref{eq:ha-25}) in the Appendix. Here, $a \in (0, 0.5)$ indicates the amount of taper applied at each edge. In our simulations, the choice of $a$ does not seem to affect the overall performance of the estimators. Next, when calculating $\widehat{I}_{h,n}$, the first-order intensities $\lambda^{(i)}$ ($i \in \{1, 2\}$) are estimated by fitting the correctly specified log-linear regression model. Simulation results for the case when the first-order intensity is incorrectly specified are provided in Appendices \ref{sec:sim-mis} and \ref{sec:sim-mis-res}.

Secondly, for both kernel estimators, we implement the Riemann sum approximation version as in (\ref{eq:KSDE2}) with the triangular kernel $K(\xb) = K_{\text{tri}}(x_1) K_{\text{tri}}(x_2)$, where $K_{\text{tri}}(x) = \max\{1 - |x|, 0\}$. To compute $\widehat{F}_{\text{opt}}$, we use the optimal bandwidth $b_{\text{opt}} = |D|^{-1/6} = A^{-1/3}$ for $A \in \{20, 40\}$, as described in Theorem \ref{thm:opt-MSE}. However, for the smallest side length $A = 10$, since the optimal bandwidth $b_{\text{opt}} = 10^{-1/3} \approx 0.46$ is smaller than the computational grid size (as explained below), we set $b_{\text{opt}} = 0.5$ for $A = 10$.

Lastly, to compute $\widehat{F}_{\text{CV}}$, we need to determine the prespecified domain $W$ and the computational grid $\bOmega$ as in (\ref{eq:Whittle-b}). Observing Figure \ref{fig:simdiag}(b), for all three models, when the frequency satisfies $\|\ob\|_{\infty} > 1.5\pi \approx 4.71$, $F_h^{(i,j)}(\ob)$ is close to its constant asymptotic value. This indicates that there is little additional information about the pseudo-spectrum for frequencies $\|\ob\|_{\infty} > 1.5\pi$. Therefore, we choose $W = [-1.5\pi, 1.5\pi]^2$ for all three models, ensuring robust results without additional computational cost. For the computational grid, we set $\bOmega = (\frac{4}{3}A, \frac{4}{3}A)^\top$, so that $L(b)$ in (\ref{eq:Whittle-b}) is evaluated on the grids $\ob_{\tb, \bOmega} = \left(\frac{1.5 \pi t_1}{A}, \frac{1.5 \pi t_2}{A}\right)^\top$ for $t_1, t_2 \in \{-A, -A + 1, \dots, A\}$.

\subsection{Simulation results} \label{sec:sim-results}

\subsubsection{The bandwidth selection methods} \label{sec:bandwidth-sim}

Table \ref{table:simband} below displays the three quartile values and the average of the cross-validated bandwidth $b_{\text{CV}}$ based on 500 replications of the point pattern. For reference, the first row also provides the optimal bandwidth $b_{\text{opt}}$, which depends only on the size of the observation window.

\begin{table}[ht!]
	\centering
	\begin{tabular}[t]{cccccccccc}
		\multicolumn{1}{c}{ } & \multicolumn{3}{c}{$D = [-5,5]^2$} & \multicolumn{3}{c}{$D = [-10,10]^2$} & \multicolumn{3}{c}{$D = [-20,20]^2$} \\
		\cmidrule(l{3pt}r{3pt}){2-4} \cmidrule(l{3pt}r{3pt}){5-7} \cmidrule(l{3pt}r{3pt}){8-10}
		& M1 & M2 & M3 & M1 & M2 & M3 & M1 & M2 & M3\\
		\hline \hline
		$b_\text{opt}$  & \multicolumn{3}{c}{0.5} & \multicolumn{3}{c}{0.37} & \multicolumn{3}{c}{0.29}\\
		Q1 & 1.13 & 1.06 & 1.08 & 0.59 & 0.54 & 0.55 & 0.37 & 0.28 & 0.28\\
		Q2 & 1.51 & 1.18 & 1.18 & 0.76 & 0.60 & 0.60 & 0.40 & 0.30 & 0.30\\
		Q3 & 1.88 & 1.55 & 1.53 & 0.84 & 0.77 & 0.75 & 0.40 & 0.32 & 0.30\\
		Mean & 1.61 & 1.26 & 1.29 & 0.80 & 0.66 & 0.65 & 0.39 & 0.31 & 0.30\\
		\hline
	\end{tabular}
	\caption{The three quartiles (second to fourth row) and the average (fifth row) of $b_\text{CV}$ from 500 replications for each window and model. The first row indicates the optimal bandwidth.}
	\label{table:simband}
\end{table}

We observe that both the value and the interquartile range of $b_\text{CV}$ decrease as the window size increases. Moreover, the distributions are slightly right-skewed, except for the case M1 with $D = [-20,20]^2$. Interestingly, within the same window size, the summary statistics of the selected bandwidths for M2 and M3 are close to each other, while the bandwidth for M1 is larger than those for M2 and M3. This suggests that the sign of the inter-species interaction may have a smaller effect on the choice of bandwidth.

Next, to compare $b_\text{opt}$ and $b_\text{CV}$, we note that the cross-validation method tends to select a larger bandwidth than the optimal MSE criterion. Specifically, for the smallest window, $b_\text{CV}$ is more than twice as large as $b_\text{opt}$. However, as the window size increases, the ratio $b_\text{CV} / b_\text{opt}$ approaches one, indicating that $b_\text{CV}$ may attain the same asymptotic convergence rate as $b_\text{opt}$.

\subsubsection{Estimation accuracy}

Moving on, we compare the accuracy of the three pseudo-spectrum estimators. To assess the effects of bias and variance separately, we compute two metrics that are widely used to evaluate the global performance of spectral estimators. The first metric is the integrated squared bias (IBIAS$^2$) over the prespecified domain:
\begin{equation} \label{eq:rBIAS}
\text{IBIAS$^2$} = 
\frac{1}{\# \{ \ob_{\tb,\bOmega} \in W_{o}\}} \sum_{\ob_{\tb,\bOmega} \in W_{o}}  \left( \frac{\frac{1}{500} \sum_{r=1}^{500} \Re \widehat{F}_{r}^{(i,j)} (\ob_{\tb,\bOmega})}{F_h^{(i,j)}(\ob_{\tb,\bOmega})} - 1 \right)^2,
\end{equation} 
where $\Re \widehat{F}^{(i,j)}_{r}$ denotes the real part of the $r$th replication of one of the three estimators of $F_h^{(i,j)}$. In (\ref{eq:rBIAS}), we choose the prespecified domain $W_{o} = \{\ob: 0.1\pi \leq \|\ob\|_{\infty} \leq 1.5\pi \}$ and the computational grid $\ob_{\tb, \bOmega} = \left( \frac{1.5\pi t_1}{A}, \frac{1.5 \pi t_2}{A} \right)^\top$, $t_1, t_2 \in \Z$. The upper bound of $W_{o}$ and the computational grid $\ob_{\tb, \bOmega}$ are chosen in the same spirit as determining the prespecified domain and grid for $L(b)$ in the previous section. According to (\ref{eq:Per-thm}), the feasible periodogram exhibits large relative bias near the origin. Therefore, we exclude low-frequency values (corresponding to $\|\ob\|_{\infty} < 0.1\pi$) when evaluating $\text{IBIAS}^2$. The second metric is the integrated MSE (IMSE), which is defined similarly to IBIAS$^2$, but replacing the summand in (\ref{eq:rBIAS}) with the MSE
$\frac{1}{500} \sum_{r=1}^{500}( \Re \widehat{F}_{r}^{(i,j)} (\ob_{\tb,\bOmega})/ F_h^{(i,j)}(\ob_{\tb,\bOmega}) - 1)^2$.

The estimation results are presented in Table \ref{table:simsummary}. First, we observe that the IBIAS$^2$ of all three estimators decreases to zero as $|D|$ increases. This confirms the asymptotic unbiasedness of the periodogram and the kernel spectral density estimator, as stated in Equation (\ref{eq:Per-thm}) and Theorem \ref{thm:KSDE1}. The magnitudes of IBIAS$^2$ are comparable across the three estimators. Within the same model and observation window, however, the IBIAS$^2$ for the cross pseudo-spectrum is larger than that of the marginal pseudo-spectrum. There are two explanations for the underperformance of the cross statistics. First, for large $\|\ob\|$ values, $F_h^{(1,2)}(\ob)$ is close to zero, whereas $F_h^{(1,1)}$ and $F_h^{(2,2)}$ remain strictly positive, which inflates the numerical error when calculating the relative bias as in (\ref{eq:rBIAS}). Second, while the marginal pseudo-spectrum is always positive-valued, the cross pseudo-spectrum is complex-valued that contains additional phase information between the two point processes, making it difficult to estimate the structure using cross statistics.

{\footnotesize
\begin{table}[ht!]
	\centering
	\begin{tabular}[t]{ccccccccc}
	\multicolumn{3}{c}{} & \multicolumn{2}{c}{$\widehat{I}_{h,n}$} & \multicolumn{2}{c}{$\widehat{F}_\text{opt}$} & \multicolumn{2}{c}{$\widehat{F}_\text{CV}$} \\
	\cmidrule(l{3pt}r{3pt}){4-5} \cmidrule(l{3pt}r{3pt}){6-7} \cmidrule(l{3pt}r{3pt}){8-9}
	\multirow{-2}{*}{Model} & \multirow{-2}{*}{\shortstack{Pseudo-\\spectrum}} & \multirow{-2}{*}{Window} & IBIAS$^2$ & IMSE & IBIAS$^2$ & IMSE & IBIAS$^2$ & IMSE\\
	\hline \hline
	&  & $[-5,5]^2$ & 0.00 & 1.14 & 0.00 & 0.83 & 0.00 & 0.24\\
	&  & $[-10,10]^2$ & 0.00 & 1.07 & 0.00 & 0.29 & 0.00 & 0.12\\
	& \multirow{-3}{*}{$F^{(1,1)}_h$} & $[-20,20]^2$ & 0.00 & 1.01 & 0.00 & 0.12 & 0.00 & 0.08\\
	\cmidrule{2-9}
	&  & $[-5,5]^2$ & 0.00 & 1.04 & 0.00 & 0.75 & 0.00 & 0.18\\
	&  & $[-10,10]^2$ & 0.00 & 1.02 & 0.00 & 0.26 & 0.00 & 0.10\\
	& \multirow{-3}{*}{$F^{(2,2)}_h$} & $[-20,20]^2$ & 0.00 & 1.00 & 0.00 & 0.11 & 0.00 & 0.07\\
	\cmidrule{2-9}
	&  & $[-5,5]^2$ & 0.39 & 145.85 & 0.27 & 101.75 & 0.04 & 16.72\\
	&  & $[-10,10]^2$ & 0.18 & 117.77 & 0.05 & 28.29 & 0.02 & 9.95\\
	\multirow{-9}{*}{M1} & \multirow{-3}{*}{$F^{(1,2)}_h$} & $[-20,20]^2$ & 0.21 & 102.52 & 0.02 & 10.97 & 0.01 & 7.09\\
	\cmidrule{1-9}
	&  & $[-5,5]^2$ & 0.01 & 1.19 & 0.01 & 0.90 & 0.01 & 0.38\\
	&  & $[-10,10]^2$ & 0.00 & 1.03 & 0.00 & 0.31 & 0.00 & 0.16\\
	& \multirow{-3}{*}{$F^{(1,1)}_h$} & $[-20,20]^2$ & 0.00 & 1.01 & 0.00 & 0.13 & 0.00 & 0.13\\
	\cmidrule{2-9}
	&  & $[-5,5]^2$ & 0.01 & 1.02 & 0.01 & 0.73 & 0.01 & 0.24\\
	&  & $[-10,10]^2$ & 0.00 & 1.00 & 0.00 & 0.25 & 0.00 & 0.12\\
	& \multirow{-3}{*}{$F^{(2,2)}_h$} & $[-20,20]^2$ & 0.00 & 0.98 & 0.00 & 0.11 & 0.00 & 0.11\\
	\cmidrule{2-9}
	&  & $[-5,5]^2$ & 0.06 & 27.80 & 0.04 & 19.62 & 0.02 & 5.48\\
	&  & $[-10,10]^2$ & 0.04 & 22.52 & 0.01 & 5.69 & 0.01 & 2.55\\
	\multirow{-9}{*}{M2} & \multirow{-3}{*}{$F^{(1,2)}_h$} & $[-20,20]^2$ & 0.04 & 20.16 & 0.01 & 2.29 & 0.00 & 2.14\\
	\cmidrule{1-9}
	&  & $[-5,5]^2$ & 0.01 & 1.16 & 0.01 & 0.87 & 0.01 & 0.31\\
	&  & $[-10,10]^2$ & 0.00 & 1.05 & 0.00 & 0.31 & 0.00 & 0.15\\
	& \multirow{-3}{*}{$F^{(1,1)}_h$} & $[-20,20]^2$ & 0.00 & 1.00 & 0.00 & 0.13 & 0.00 & 0.13\\
	\cmidrule{2-9}
	&  & $[-5,5]^2$ & 0.01 & 1.06 & 0.01 & 0.75 & 0.01 & 0.19\\
	&  & $[-10,10]^2$ & 0.00 & 1.05 & 0.00 & 0.25 & 0.00 & 0.12\\
	& \multirow{-3}{*}{$F^{(2,2)}_h$} & $[-20,20]^2$ & 0.00 & 0.99 & 0.00 & 0.11 & 0.00 & 0.11\\
	\cmidrule{2-9}
	&  & $[-5,5]^2$ & 0.40 & 288.41 & 0.28 & 199.14 & 0.10 & 34.83\\
	&  & $[-10,10]^2$ & 0.48 & 208.98 & 0.11 & 49.69 & 0.05 & 20.46\\
	\multirow{-9}{*}{M3} & \multirow{-3}{*}{$F^{(1,2)}_h$} & $[-20,20]^2$ & 0.34 & 174.90 & 0.03 & 19.10 & 0.03 & 18.50\\
	\bottomrule
	\end{tabular}
	\caption{IBIAS$^2$ and IMSE for the three pseudo-spectrum estimators.}
	\label{table:simsummary}
\end{table}
}

Second, we compare the IMSE. Unlike IBIAS$^2$, there is a stark distinction between the IMSE of the raw periodogram and those of the two kernel-smoothed estimators. The IMSE of the raw periodogram remains substantial even for the largest window, whereas the IMSE of both kernel spectral density estimators converges to zero as the window size increases. This observation corroborates the theoretical results in Equation (\ref{eq:Per-thm}) and Theorem \ref{thm:KSDE1}.
Within the same model and observation window, the IMSE of the cross pseudo-spectrum tends to be larger than that of the marginal pseudo-spectrum, which can be attributed to the numerical and structural issues discussed previously. Moreover, it is noteworthy that the IMSE of $\widehat{F}^{(1,2)}$ under the negative inter-species interaction model (M3) is larger than that under the positive inter-species interaction model (M2).

Last, we compare the performance of the two bandwidth selection methods. Within the same model and observation window, the IBIAS$^2$ of $\widehat{F}_{\text{CV}}$ tends to be smaller than that of $\widehat{F}_{\text{opt}}$. According to Theorem \ref{thm:KSDE-bound-b} in the Appendix, the bias arising from the estimation of the first-order intensity is of order $O(|D_n|^{-1}b^{-d})$. This leads to a larger bias for $\widehat{F}_{\text{opt}}$, provided the cross-validation method selects a larger bandwidth (see Table \ref{table:simband}). Furthermore, the IMSE of $\widehat{F}_{\text{CV}}$ is uniformly smaller than that of $\widehat{F}_{\text{opt}}$, with the difference being more pronounced for small and moderate window sizes ($A \in \{10,20\}$). This suggests that, for small to moderate windows, the benefits of the data-driven cross-validation method in selecting the bandwidth outweigh the additional computational costs required to compute $b_{\text{CV}}$. Finally, by defiintion, the IMSE of $\widehat{F}_{\text{opt}}$ achieves the fastest convergence rate to zero. Hence, the results in Table \ref{table:simsummary} for $D = [-20,20]^2$ also support the evidence that $b_{\text{CV}}$ is asymptotically optimal.

%% file: bci.tex
\section{Real data application} \label{sec:data}
In this section, we apply our methods to the tree species data from the tropical forest of Barro Colorado Island (BCI). The dataset contains spatial locations (within a $1000\times500~\text{m}^2$ rectangular window) of trees and shrubs with stem diameters larger than 10 mm (\cite{m:cond-19}). Point process modeling of these tree species has also been considered in \cite{p:jal-15, p:waa-16} (fitting parametric inhomogeneous models in the spatial domain) and \cite{p:gra-23} (analyzing frequency domain features under the stationarity assumption).

For the analysis, we focus on five specific tree species, namely \textit{Capparis frondosa} ($X_1$), \textit{Hirtella triandra} ($X_2$), \textit{Protium panamense} ($X_3$), \textit{Protium tenuifolium} ($X_4$), and \textit{Tetragastris panamensis} ($X_5$) from the seventh census, which are also studied in \cite{p:jal-15}. The spatial point patterns of these species are displayed in Figure \ref{fig:bci}(a). These patterns clearly reveal inhomogeneity in the first-order intensities, which is further supported by the frequency domain goodness of fit check (see Appendix \ref{sec:DA} for details). Accordingly, we model these point patterns using a five-variate SOIRS process.

To investigate frequency domain features, we compute the pseudo-spectrum estimator. Given the relatively large observation window measured in meter, the first-order intensities are on the order of $10^{-3}$. For better comparison with the results obtained using synthetic data in Section \ref{sec:sim}, we rescale the window to $50\times25~\text{unit}^2$, where one unit corresponds to 20 m. This ensures that the fitted first-order intensity functions have a scale between 1 and 10 (the same scale used in the simulation study) and corresponds to the moderate to large window size.

 Moreover, when calculating the feasible periodogram, we fit a log-linear first-order intensity model using ten environmental covariates as provided in \cite{p:jal-15}. For the kernel spectral estimator $\widehat{F}_{n,b}$, we employ the cross-validation bandwidth selection method on a computational grid $\{(\frac{1.5\pi t_1}{A_1}, \frac{1.5\pi t_2}{A_2})^\top\}$ with $(A_1, A_2) = (50, 25)$ and $W = [-1.5\pi, 1.5\pi]^2$. Consequently, we select the bandwidth  $b = 0.62$.

\begin{figure}[h]
	\centering
	\begin{subfigure}{0.17\textwidth}
		\centering
		\includegraphics[width=1.04\textwidth]{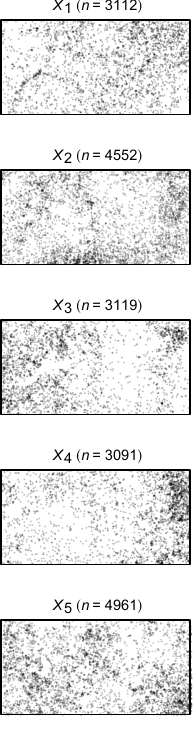}
		\caption{}
	\end{subfigure}
	\begin{subfigure}{0.82\textwidth}
		\centering
		\includegraphics[width=0.96\textwidth]{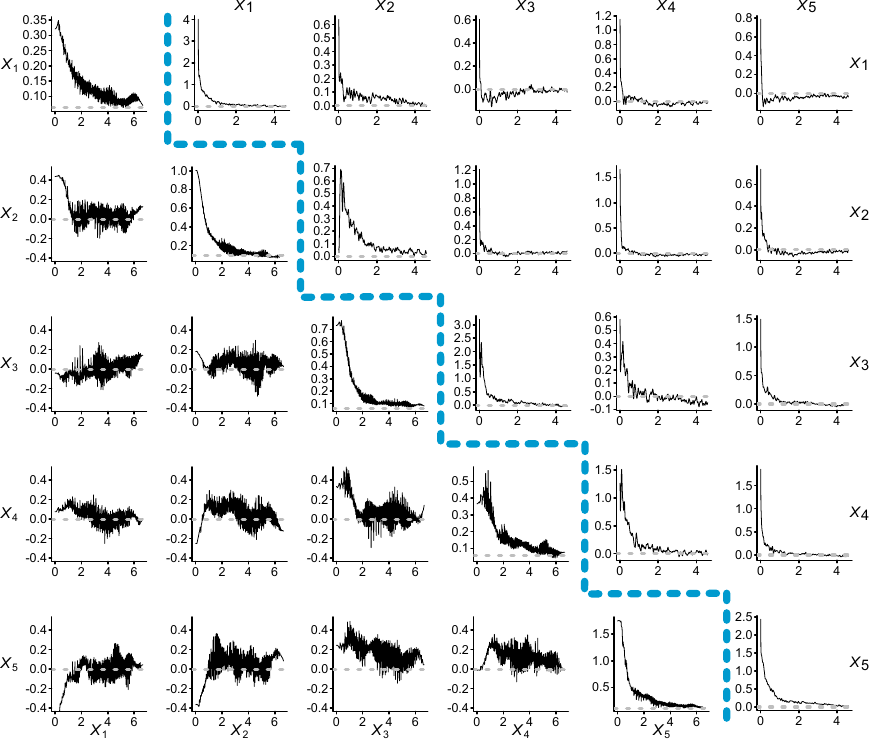}
		\caption{}
	\end{subfigure}
	\caption{(a): Point patterns of five tree species from the BCI dataset. The number in parentheses indicates the total number of observations. Lower (b): Radial average of $\widehat{R}^{(i,j)}(\ob)$ (for $i \neq j$) and $\widehat{F}^{(i,j)}(\ob)$ (for $i=j$). Upper (b): Nonparametric shifted PCF estimators $\widehat{g}^{(i,j)}(r)-1$.}
	\label{fig:bci}
\end{figure}

As one of the reviewers pointed out, to compare the cross pseudo-spectrum pairs on a uniform scale, in the lower diagonal part of Figure \ref{fig:bci}(b), we plot the radial average of the real part of the coherence estimator $\widehat{R}^{(i,j)}(\ob) = \Re \widehat{F}^{(i,j)}(\ob)/ \{\widehat{F}^{(i,i)}(\ob)\widehat{F}^{(j,j)}(\ob)\}^{1/2}$ for $i \neq j$ and $\widehat{F}^{(i,i)}(\ob)$ for $i=j$. The horizontal dotted lines indicate the asymptotic values. For comparison of our approach with spatial domain methods, the upper diagonal part of Figure \ref{fig:bci}(b) plots the nonparametric shifted pair correlation function (PCF) estimators, which correspond to the estimator of $\ell^{(i,j)}(\cdot)$ in our context.

We first focus on the marginal estimators. For all five species, the values $\widehat{F}^{(i,i)}(\ob)$ remain bounded near the origin and lie above their respective asymptotes. These observations indicate that the fitted first-order intensities using the log-linear regression model are reasonable approximations of the underlying true intensities (see, Remark \ref{rmk:near-origin}) and that all five species exhibit marginal clustering. Additionally, the pseudo-spectrum of $X_4$ shows a peak around $\|\ob\| \approx 1$, whereas the pseudo-spectra of the other four species are monotonically decreasing. Thanks to Corollary \ref{coro:KSDE2}, the nonzero peak observed in $\widehat{F}^{(4,4)}(\ob)$ is statistically valid even in the case we incorrectly specified the first-order intensity of $X_4$. Since a nonzero peak in the frequency domain corresponds to oscillatory behavior in the spatial domain, this suggests that $g^{(4,4)}(r) - 1$ may behave like a damping oscillatory function (e.g., $g^{(4,4)}(r) - 1 = \alpha e^{-\beta r} \cos(2\pi r)$ for some $\alpha, \beta > 0$). However, the oscillatory behavior of $X_4$ is less pronounced in the PCF plot.

Next, we examine the cross estimators. The cross pseudo-spectrum and the coherence provide insights into the second-order interactions, with the latter standardized to satisfy $|\widehat{R}^{(i,j)}(\ob)| \leq 1$. For example, $\widehat{R}^{(1,2)}$ is positive near the origin, indicating overall inter-species clustering between $X_1$ and $X_2$, which is corroborated by $\widehat{g}^{(1,2)}(r) - 1 > 0$ in the PCF plot. Similarly, negative values of $\widehat{R}^{(1,5)}$ near the origin indicate that $X_1$ and $X_5$ repel each other, although this repulsion is less evident in the PCF plots, as $\widehat{g}^{(1,5)}(r) - 1$ exhibits a large peak near the origin. Our results are also consistent with the fitted PCF results in \cite{p:jal-15}, where $g^{(i,j)}(r;\widehat{\btheta}) - 1 > 0$ (resp. $<0$) indicates inter-species clustering (resp. repulsion) between $X_i$ and $X_j$.

Finally, the coherence and its estimator also provide information about the asymptotic uncorrelatedness of two point processes. However, while coherence has a clear interpretation for SOS point processes, its use for inhomogeneous processes is not straightforward. In this regard, in Appendix \ref{sec:coh}, we propose new correlation structures for SOIRS processes via the coherence and partial coherence of the intensity reweighted process $\underline{\widetilde{X}}$ of $\underline{X}$. See Appendix \ref{sec:coh} for further details on the coherence and partial coherence analysis and its application to the BCI data.

%% file: appendix1.tex

\section{Background on the SOIRS process} \label{sec:joint}

\subsection{The intensity reweighted process} \label{sec:construction}
The authors of \cite{p:bad-00} were the first to define the SOIRS process through the marked point process framework. For simplicity, let $X$ be a univariate spatial point process with strictly positive first-order intensity $\lambda(\xb)$. Let
\begin{equation*}
\widetilde{N}(A) = \sum_{\xb \in X} \frac{1}{\lambda(\xb)} I(\xb \in A), \quad A \in \mathcal{B}(\R^d),
\end{equation*}
where $I(\cdot)$ denotes the indicator function. Then, \cite{p:bad-00} defined $X$ to be SOIRS if $\widetilde{N}(\cdot)$ is a second-order stationary measure. However, this construction appears to be not accurate. 
To see this, let $A$ and $B$ be bounded Borel sets and let $A+\tb$ and $B+\tb$ denote their translations by $\tb \in \R^d$. Under their construction, $X$ is SOIRS if $\Ex[\widetilde{N}(A+\tb) \widetilde{N}(B+\tb)] = \Ex[\widetilde{N}(A) \widetilde{N}(B)]$ for any $\tb \in \R^d$.
However, a simple calculation gives
\begin{equation*}
\begin{aligned}
& \Ex[\widetilde{N}(A+\tb) \widetilde{N}(B+\tb)] \\
&= \Ex\left[ \sum_{\xb \in X} \frac{1}{\lambda(\xb)^2} I( \xb \in (A+\tb) \cap (B+\tb)) \right] 
+ \Ex\left[ \sum_{\xb, \yb \in X}^{\neq} \frac{1}{\lambda(\xb) \lambda(\yb)}  I(\xb \in A+\tb)  I(\yb \in B+\tb) \right] \\
&= \int_{(A+\tb) \cap (B+\tb)} \frac{1}{\lambda(\xb)} d\xb +
\iint_{(A+\tb) \times (B+\tb)} g(\xb-\yb) d\xb d\yb \\
&= \int_{(A+\tb) \cap (B+\tb)} \frac{1}{\lambda(\xb)} d\xb +
\iint_{A \times B} g(\xb-\yb) d\xb d\yb,
\end{aligned}
\end{equation*}
where $g(\xb - \yb) = \lambda_2(\xb,\yb) / \{\lambda(\xb) \lambda(\yb) \}$ denotes the pair correlation function (PCF) of $X$. 
Therefore, the calculation of $\Ex [ \widetilde{N}(A) \widetilde{N}(B)]$ between the Equations (2) and (3) of \cite{p:bad-00} requires an additional term, namely $\int_{A \cap B} \frac{1}{\lambda(\xb)} d\xb$. Unfortunately, 
\begin{equation*}
\int_{(A+\tb) \cap (B+\tb)} \frac{1}{\lambda(\xb)} d\xb \neq \int_{A \cap B} \frac{1}{\lambda(\xb)} d\xb,
\end{equation*}
unless $\lambda(\cdot)$ is constant or $A \cap B = \emptyset$. Hence, in general, $\Ex[\widetilde{N}(A+\tb) \widetilde{N}(B+\tb)] \neq \Ex[\widetilde{N}(A) \widetilde{N}(B)]$, and $\widetilde{N}(\cdot)$ cannot be a stationary measure.

Alternatively, we adopt the concept of a weighted point process to construct the multivariate SOIRS process. 

\begin{definition}[Weighted spatial point process] \label{def:weight}
Let $X$ be a univariate point process on $D \subset \R^d$, and let $w(\cdot)$ be a bounded, non-negative function on $D$. Choose $k \in \mathbb{N}$ such that $k \geq \sup_{\xb \in D} w(\xb)$. Then, $p(\xb) = w(\xb)/k$ defines a probability (i.e., $p(\xb) \in [0,1]$). Let $X^{(1)}, \dots, X^{(k)}$ be independent copies of $X$. The weighted point process of $X$ with respect to $w(\cdot)$ is defined as the thinned process of the superposition of $X^{(1)}, \dots, X^{(k)}$, with independent thinning probability $p(\cdot)$.  

For a multivariate point process $\underline{X} = (X_1, \dots, X_m)$ and a bounded, non-negative vector-valued function $\underline{w}(\xb) = (w_1(\xb), \dots, w_m(\xb))$, the weighted point process of $\underline{X}$ with respect to $\underline{w}$ is an $m$-variate process where the $i$th component is the weighted process of $X_i$ with respect to weight $w_i$.
\end{definition}

Recall $\underline{\lambda}_1(\xb)$ in (\ref{eq:lambda-vector}). Using the above notion, we now define the multivariate SOIRS process.

\begin{definition} \label{def:k-IRS}
An $m$-variate point process $\underline{X}$ on $D \subset \R^d$ is called second-order intensity reweighted stationary (SOIRS) if $\underline{\widetilde{X}}$,
the weighted process of $\underline{X}$ with weights $(1/\lambda_1^{(1)}(\cdot), \dots, 1/\lambda_1^{(m)}(\cdot))$, is second-order stationary process on $D$. Here, we refer to $\underline{\widetilde{X}}$ as the intensity reweighted process of $\underline{X}$.
\end{definition}

Although this definition differs slightly from that of \cite{p:bad-00}, all other statistical properties presented therein remain unchanged.

\subsection{Comparison between the spatial and frequency domain estimands} \label{sec:K-function}

The inhomogeneous $K$-function (\cite{p:bad-00}) is widely used to assess the second-order structure of a SOIRS process. For simplicity, let $X$ be a univariate isotropic SOIRS process with first-order intensity $\lambda(\xb)$ and PCF $g(\xb)$. Then, the inhomogeneous $K$-function of $X$ is defined as
\begin{equation} \label{eq:Kinhom}
K_{\text{inhom}}(r) = \int_{B(\mathbf{0}, r)} g(\xb) \, d\xb, \quad r \in [0,\infty),
\end{equation}
where $B(\mathbf{0}, r)$ denotes the $d$-dimensional Euclidean ball of radius $r>0$ centered at the origin. The rationale behind $K_{\text{inhom}}(r)$ is that it coincides with Ripley's $K$-function of the corresponding intensity reweighted process $\widetilde{X}$ of $X$. Moreover, \cite{p:bad-00} proposed the nonparametric unbiased estimator of $K_{\text{inhom}}(r)$ as
\begin{equation} \label{eq:Khat}
\widehat{K}_n(r) = |D_n|^{-1} \sum_{\xb, \yb \in X \cap D_n}^{\neq} \frac{w(\xb,\yb) \, I( \|\xb-\yb\| \leq r)}{\lambda(\xb) \lambda(\yb)},
\end{equation}
where $w(\xb,\yb)$ is an edge correction factor and $\{D_n\}$ is a sequence of observation windows. Under an increasing domain framework, the distributional properties of $\widehat{K}_n(r)$ are investigated in \cite{p:waa-09}.

Now, we borrow the idea of the ``intensity reweighting'' to define the frequency domain estimand. Recall that $\widetilde{X}$ is SOS with unit first-order intensity and having $\ell_2(\cdot)$ as the reduced covariance intensity. Therefore, the spectrum of $\widetilde{X}$ is given by
\begin{equation*}
\widetilde{F}(\ob) = (2\pi)^{-d} + \mathcal{F}^{-1}(\ell_2)(\ob), \quad \ob \in \R^d.
\end{equation*}
To estimate $\widetilde{F}(\ob)$, we consider the intensity reweighted version of the (untapered) periodogram
$I_{n}^{\text{IR}}(\ob) = \big| \mathcal{J}_{n}^{\text{IR}}(\ob) - \Ex[\mathcal{J}_{n}^{\text{IR}}(\ob)] \big|^2$, where
\begin{equation} \label{eq:J-IR}
\mathcal{J}_{n}^{\text{IR}}(\ob) = (2\pi)^{-d/2} |D_n|^{-1/2} \sum_{\xb \in X \cap D_n} \frac{1}{\lambda(\xb)} \exp(-i \xb^\top \ob).
\end{equation}

However, this approach seems not so successful in at least three reasons. Firstly, by applying Campbell's formula, the expectation of $I_n^{\text{IR}}(\ob)$ is given by
\begin{equation} \label{eq:I-expectation}
\Ex[I_{n}^{\text{IR}}(\ob)] 
= (2\pi)^{-d} |D_n|^{-1} \int_{D_n} \frac{1}{\lambda(\xb)} \, d\xb
+ (2\pi)^{-d} |D_n|^{-1} \int_{D_n^2} \ell_2(\xb-\yb) \, e^{-i (\xb-\yb)^\top \ob} \, d\xb \, d\yb.
\end{equation}
Therefore, if we rely solely on the increasing domain framework as our asymptotics, the first term in (\ref{eq:I-expectation}) does not necessarily converge as $|D_n| \rightarrow \infty$. This implies that a stronger structural assumption on $\lambda(\cdot)$, such as that in Definition \ref{def:infill}(ii), must be imposed on top of the increasing domain framework to ensure the convergence of the periodogram-based estimator. Secondly, suppose that the limit in (\ref{eq:I-expectation}) exists. For example, consider the SOS case where $\lambda(\xb) \equiv \lambda$ is constant. Then, we have
\begin{equation*}
\lim_{n \rightarrow \infty} \Ex[I_{n}^{\text{IR}}(\ob)] 
= (2\pi)^{-d} \frac{1}{\lambda} + \mathcal{F}^{-1}(\ell_2)(\ob),
\end{equation*}
which is not equal to $\widetilde{F}(\ob)$ unless $\lambda = 1$. This indicates that an additional bias correction to the intensity reweighted periodogram is required to obtain an asymptotically unbiased estimator of $\widetilde{F}(\ob)$. Such a debiasing procedure complicates the analysis of periodogram-based estimation and inference for SOIRS processes. For example, positive-definiteness of the periodogram may be lost after debiasing.
Lastly, since the first-order intensity is typically unknown, $\mathcal{J}_{n}^{\text{IR}}(\ob)$ in (\ref{eq:J-IR}) must be estimated. However, it is unclear what the intensity reweighted periodogram estimates when the first-order intensity is misspecified, a situation considered in Section \ref{sec:feasibleDFT}.

Therefore, we claim that using the ordinary periodogram under the new asymptotic framework in Definition \ref{def:infill} provides the most natural approach to establishing spectral methods for SOIRS processes. This framework yields clear interpretations of the estimand with theoretical guarantees even under misspecified first-order intensity models.


\section{Properties of the pseudo-spectrum and its raw estimator} \label{sec:ps-prop}

Recall the pseudo-spectrum corresponding to the data taper $h$: 
\begin{equation*}
F_{h}(\ob) =  (2\pi)^{-d} H_{h,2}^{-1} \diag (H_{h^2 \underline{\lambda}})
+ H_{h,2}^{-1} \Big( H_{h^2 \underline{\lambda} \cdot \underline{\lambda}^\top} \odot
\mathcal{F}^{-1}(L_2)(\ob) \Big).
\end{equation*}

Below, we list three important mathematical properties of $F_h$.

\begin{proposition} \label{prop:Fh}
Let $\{\underline{X}_{D_n}\}$ be an $m$-variate locally stationary SOIRS process. Suppose further that $L_2(\cdot)$ in Definition \ref{def:infill}(i) belongs to $L^{m \times m}_1(\R^d)$. Then the following assertions hold.
\begin{itemize}
\item[(i)] If the first-order intensity function $\underline{\lambda}(\xb)$ is constant, then, regardless of the choice of $h$, $F_h$ coincides with the spectrum of the corresponding SOS process.

\item[(ii)] For any $\ob \in \R^d$, $F_h(\ob)$ is conjugate symmetric and positive definite.

\item[(iii)] If $\{\underline{X}_{D_n}\}$ corresponds to a sequence of $m$-variate inhomogeneous Poisson point processes, then the pseudo-spectrum simplifies to
\begin{equation*}
F_h(\ob) = (2\pi)^{-d} \, \diag \Bigg(
\frac{\int_{[-1/2,1/2]^d} (h^2 \lambda^{(1)})(\xb) \, d\xb}{\int_{[-1/2,1/2]^d} h^2(\xb) \, d\xb}, \dots,
\frac{\int_{[-1/2,1/2]^d} (h^2 \lambda^{(m)})(\xb) \, d\xb}{\int_{[-1/2,1/2]^d} h^2(\xb) \, d\xb}
\Bigg).
\end{equation*}
This shows that the pseudo-spectrum of an inhomogeneous Poisson point process is a constant diagonal matrix, analogous to the constant spectrum of a homogeneous Poisson process.
\end{itemize}
\end{proposition}
\begin{proof}
By substituting the constant first-order intensity $\underline{\lambda}$ into the expression of $F_h$, we obtain, for any $h$,
\begin{equation*}
F_h(\ob) = (2\pi)^{-d} \, \diag(\underline{\lambda}) + (\underline{\lambda} \underline{\lambda}^\top) \odot \mathcal{F}^{-1}(L_2)(\ob), \quad \ob \in \R^d.
\end{equation*}
Therefore, $F_h$ coincides with the spectrum of the SOS process, proving (i).

To show (ii), we use the integral representation of $F_h(\ob)$ given in Theorem \ref{thm:local-S}. Since the local spectrum $F^{\ubb}(\cdot)$ is conjugate symmetric and positive definite, $F_h$ also inherits these properties.

Lastly, (iii) follows directly from the fact that $L_2(\xb)$ is a zero matrix for an inhomogeneous Poisson process.
\end{proof}

Now, we study the statistical properties of the DFTs and periodograms, where the latter serves as a raw estimator of $F_h$. Recall the feasible DFT $\underline{\widehat{J}}_{h,n}(\ob)$ and periodogram $\underline{\widehat{I}}_{h,n}(\ob)$ defined in Section \ref{sec:feasibleDFT}. In this section, we focus only on the case where the parametric first-order intensity model is correctly specified; results for the misspecified case are provided in Appendix \ref{sec:KSDE-bound2}.

To establish asymptotic properties of the DFT and periodogram, we define two sequences of frequencies $\{\ob_{1,n}\}_{n=1}^\infty$ and $\{\ob_{2,n}\}_{n=1}^\infty$ on $\R^d$ to be asymptotically distant if 
\begin{equation}
\lim_{n \to \infty} \|\aB \odot (\ob_{1,n} - \ob_{2,n})\| = \infty,
\end{equation} 
where $\aB = \aB(n)$ denotes the vector of side lengths of $D_n$.

Below theorem states that the DFTs and periodograms are asymptotically uncorrelated for two asymptotically distant frequencies and that the periodogram is an asymptotically unbiased estimator (except at the origin) of the pseudo-spectrum.

\begin{theorem} \label{thm:asymDFT-IRS}
Suppose that Assumptions \ref{assum:A}, \ref{assum:B}(for $k = 2$), \ref{assum:C}, \ref{assum:beta}, and \ref{assum:beta2}(ii) (for $r>2$) hold. 
Moreover, we assume that the true first-order intensity function is $\underline{\lambda}(\cdot;\bbeta_0)$, where $\bbeta_0$ is in Assumption \ref{assum:beta2}. Let $\{\ob_{1,n}\}$ and $\{\ob_{2,n}\}$ be two sequences in $\R^d$ such that $\{\ob_{1,n}\}$, $\{\ob_{2,n}\}$, and $\{\mathbf{0}\}$ are pairwise asymptotically distant. Moreover, let $\{\ob_{n}\}$ be a sequence that is asymptotically distant from $\{\mathbf{0}\}$ and satisfies $\lim_{n\rightarrow \infty} \ob_{n} = \ob \in \R^d$. Then,
\begin{eqnarray}
&& \lim_{n\rightarrow \infty} \cov\big( \underline{\widehat{J}}_{h,n}(\ob_{1,n}), \underline{\widehat{J}}_{h,n}(\ob_{2,n}) \big)
= O_m \quad \text{and} \label{eq:lim-DFT-IRS1} \\
&& \lim_{n\rightarrow \infty} \var\big( \underline{\widehat{J}}_{h,n}(\ob_{n}) \big)
= \lim_{n\rightarrow \infty} \Ex\big[\widehat{I}_{h,n}(\ob_{n})\big] = F_h(\ob),
\label{eq:lim-DFT-IRS2}
\end{eqnarray} 
where $O_m$ denotes the $m \times m$ zero matrix. 

Suppose further that Assumptions \ref{assum:B}(for $k=4$) and \ref{assum:beta2}(ii) (for $r>4$) hold and that $\{\ob_{1,n}\}$ and $\{-\ob_{2,n}\}$ are asymptotically distant. Then, for $i_1, i_2, j_1, j_2 \in \{1, \dots, m\}$,
\begin{eqnarray} \label{eq:lim-Per-IRS1} 
&& \lim_{n\rightarrow \infty} \cov\big( \widehat{I}_{h,n}^{(i_1,j_1)}(\ob_{1,n}), \widehat{I}_{h,n}^{(i_2, j_2)}(\ob_{2,n}) \big)
= 0 \quad \text{and} \\
&& \lim_{n\rightarrow \infty} \cov\big( \widehat{I}_{h,n}^{(i_1,j_1)}(\ob_{n}), \widehat{I}_{h,n}^{(i_2,j_2)}(\ob_{n}) \big)
= F_h^{(i_1, i_2)}(\ob) F_h^{(j_1, j_2)}(-\ob).
 \label{eq:lim-Per-IRS2} 
\end{eqnarray}
\end{theorem}
\begin{proof} 
See Appendix \ref{sec:proof1}.
\end{proof}

Now, we derive the asymptotic joint distribution of the DFTs and periodograms. To do so, we introduce the $\alpha$-mixing coefficient of $\underline{X}$. For compact and convex subsets $E_i, E_j \subset \mathbb{R}^d$, let $d(E_{i}, E_{j}) = \inf_{\xb_i \in E_{i}, \xb_j \in E_{j}} \|\xb_i - \xb_j\|_{\infty}$. Then, the $\alpha$-mixing coefficient of $\underline{X}$ is defined as
\begin{equation*}
\begin{aligned}
\alpha_{p}^{}(k) &= \sup_{A_{i}, A_j, E_i, E_j} \Bigl\{ \left| P(A_{i} \cap A_{j}) - P(A_{i}) P(A_{j}) \right|: \\
& \quad A_{i} \in \mathcal{F}_{}(E_{i}), \, A_{j} \in \mathcal{F}_{}(E_{j}), \, |E_{i}|=|E_{j}| \leq p, \, d(E_{i}, E_{j}) \geq k \Bigr\}.
\end{aligned}
\end{equation*} 
Here, $\mathcal{F}_{}(E)$ denotes the $\sigma$-field generated by the superposition of $\underline{X}$ in $E \subset \mathbb{R}^d$, and the supremum is taken over all compact and convex subsets $E_i$ and $E_j$. 

Using the $\alpha$-mixing cofficients, we can now state the CLT for the DFTs and periodograms under the SOIRS framework.

\begin{theorem} \label{thm:asym-Normal}
Suppose that Assumptions \ref{assum:A}, \ref{assum:B}(for $k = 2$), \ref{assum:C}, \ref{assum:beta}, and \ref{assum:beta2}(i) hold. Moreover, 
we assume that the true first-order intensity function is $\underline{\lambda}(\cdot;\bbeta_0)$ and that there exists $\varepsilon >0$ such that
\begin{equation*}
\sup_{p \in (0,\infty)} \frac{\alpha_p(k)}{\max(p,1)} = O(k^{-d-\varepsilon}), \quad k \rightarrow \infty.
\end{equation*}
For fixed $r \in \mathbb{N}$, let $\{\ob_{1,n}\}, \dots, \{\ob_{r,n}\}$ be $r$ sequences of frequencies that satisfy 
(i) $\lim_{n\rightarrow \infty} \ob_{i,n} = \ob_{i} \in \mathbb{R}^d$; 
(ii) $\{\ob_{i,n}\}$ is asymptotically distant from $\{\textbf{0}\}$; and 
(iii) $\{\ob_{i,n} + \ob_{j,n}\}$ and $\{\ob_{i,n} - \ob_{j,n}\}$ are asymptotically distant from $\{\textbf{0}\}$ for $i \neq j$.
Then, we have
\begin{equation*}
\big( \underline{\widehat{J}}_{h,n}(\ob_{1,n}), \dots, \underline{\widehat{J}}_{h,n}(\ob_{r,n})\big) \Dcon (Z_1^c, \dots, Z_r^c),
\end{equation*} 
where $Z_1^c, \dots, Z_r^c$ are independent $m$-variate complex normal random variables with mean zero and variance matrices $F_h(\ob_1), \dots, F_h(\ob_r)$.

If we further assume Assumption \ref{assum:B} for $k=4$, then
\begin{equation*}
\big( \widehat{I}_{h,n}(\ob_{1,n}), \dots, \widehat{I}_{h,n}(\ob_{r,n}) \big) \Dcon (W_1^c, \dots, W_r^c),
\end{equation*}
where $W_1^c, \dots, W_r^c$ are independent $m$-variate complex Wishart random variables with one degree of freedom and variance matrices $F_h(\ob_1), \dots, F_h(\ob_r)$ (see \cite{b:bri-81}, Section 4.2, for details on the complex normal and complex Wishart distributions).
\end{theorem}
\begin{proof} 
Let $\underline{J}_{h,n}$ and $I_{h,n}(\ob)$ be the theoretical DFT and periodogram, respectively, as in (\ref{eq:J-expectation}). Then, Theorem \ref{thm:asymp-IRS-fea1} below states that under the same set of assumptions, $\underline{\widehat{J}}_{h,n}$ (resp., $\widehat{I}_{h,n}(\ob)$) shares the same asymptotic distribution as $\underline{J}_{h,n}$ (resp., $I_{h,n}(\ob)$). Therefore, it is enough to show
\begin{eqnarray*}
&& \left( \Re \underline{J}_{h,n}(\ob_{1,n})^\top,
\Im \underline{J}_{h,n}(\ob_{1,n})^\top, \cdots,
\Im \underline{J}_{h,n}(\ob_{r,n})^\top \right)^\top  \Dcon \bigg( \{(F_h(\ob_1)/2)^{1/2}Z_{1}\}^{\top}, 
 \\
&&  \{(F_h(\ob_1)/2)^{1/2}Z_{2}\}^{\top}, \dots, \{(F_h(\ob_r)/2)^{1/2}Z_{2r-1}\}^{\top}, \{(F_h(\ob_r)/2)^{1/2}Z_{2r}\}^{\top}
\bigg)^\top, \quad n\rightarrow \infty,
\end{eqnarray*} 
where $\Re \underline{J}_{h,n}(\ob)$ and $\Im \underline{J}_{h,n}(\ob)$ are the real and imaginary parts of $\underline{J}_{h,n}(\ob)$, 
$\{Z_{j}: j=1, \dots, 2r\}$ are i.i.d. standard normal random variables on $\R^{m}$, and for a positive definite matrix $A$, 
$A^{1/2}$ denotes the (unique) positive square root of $A$. Showing the above is almost identical to that of the proof of YG24, Theorem 3.2 (see Appendix B.2 of the same reference), so we omit the details.
\end{proof}

\subsection{Proof of Theorem \ref{thm:asymDFT-IRS}} \label{sec:proof1}

To prove the asymptotic uncorrelatedness of the DFTs, we require the following representation of the covariance of the DFT. Recall $H_{h,k}^{(n)}$ in (\ref{eq:Hkn}).
For two functions $h$ and $g$ with support $[-1/2,1/2]^d$, let
\begin{equation*}
R_{h,g}^{(n)} (\ubb, \ob)
= \int_{\R^d} h(\xb\oslash\aB) \bigg\{  g\left((\xb+\ubb)\oslash\aB\right) - g(\xb\oslash\aB)\bigg\} \exp(-i\xb^{\top} \ob) d\xb, \quad \ubb, \ob \in \R^d.
\end{equation*}
In the lemma below, we provide an expression for $\cov( \underline{J}_{h,n}(\ob_{1}) , \underline{J}_{h,n}(\ob_{2}))$ using $H_{h,k}^{(n)}(\cdot)$ and $R_{h,g}^{(n)} (\cdot, \cdot)$.

\begin{lemma} \label{lemma:DFT-expr}
Suppose that Assumptions \ref{assum:A} and \ref{assum:B} (for $k=2$) hold. Furthermore, we assume $h$ and $\lambda^{(j)}$ are bounded on $[-1/2,1/2]^d$. Then, for $i,j \in \{1, \dots, m\}$ and $\ob_1, \ob_2 \in \R^d$,
\begin{equation}
\begin{aligned}
\cov\big( J_{h,n}^{(i)}(\ob_{1}) , J_{h,n}^{(j)}(\ob_{2}) \big) &= (2\pi)^{-d}H_{h,2}^{-1}|D_n|^{-1} \bigg(
	\delta_{i,j} \cdot H_{h^2 \lambda^{(i)},1}^{(n)}(\ob_1-\ob_2) \\  
&~~ + H_{h^2 \lambda^{(i)} \lambda^{(j)},1}^{(n)}(\ob_{1} - \ob_{2}) \int_{\R^d} e^{-i \ubb^\top \ob_{1}} \ell_2^{(i,j)}(\ubb)  d\ubb \\ 
&~~ +  \int_{\R^d} e^{-i \ubb^\top \ob_{1}} 
\ell_2^{(i,j)}(\ubb)  R_{h\lambda^{(j)}, h\lambda^{(i)}}^{(n)}(\ubb, \ob_{1} - \ob_{2})  d\ubb  \bigg), 
\end{aligned}
\label{eq:DFT-exp}
\end{equation}
where $\delta_{i,j}=1$ if $i=j$ and zero otherwise.
\end{lemma}
\begin{proof} 
By using Campbell's formula and applying similar techniques as in YG24, Theorem D.1, we have
\begin{equation}
\begin{aligned}
& \cov\big( J_{h,n}^{(i)}(\ob_{1}) , J_{h,n}^{(j)}(\ob_{2}) \big) \\
&~~= (2\pi)^{-d}H_{h,2}^{-1}|D_n|^{-1} \bigg(
	\delta_{i,j} \int_{D_n} (h^2\lambda^{(i)})(\xb\oslash\aB) e^{-i\xb^\top (\ob_{1}-\ob_{2})} d\xb \\
&~~ \quad + \int_{D_n^2} (h\lambda^{(i)})(\xb\oslash\aB)  (h\lambda^{(j)})(\yb\oslash\aB) \ell_2^{(i,j)}(\xb-\yb) e^{-i(\xb^\top \ob_{1} - \yb^\top \ob_{2})} d\xb d\yb \bigg).
\end{aligned}
\label{eq:DFT-exp00}
\end{equation} 
The first term in the parentheses above is $\delta_{i,j} H_{h^2 \lambda^{(i)},1}^{(n)}(\ob_1-\ob_2)$. 
Observe that $(h\lambda^{(i)})(\xb\oslash\aB) = 0$ on $\xb \notin D_n$, and the second term can be written as an integral on $\R^{2d}$. Therefore, by using a change of variables $(\ubb,\vbb)=(\xb-\yb,\yb)$ and using similar techniques as in YG24, Lemma D.1, the second term above equals
\begin{equation*}
(2\pi)^{-d}H_{h,2}^{-1}|D_n|^{-1} \int_{\R^d} d\ubb e^{-i \ubb^\top \ob_{1}} \ell_2^{(i,j)}(\ubb)
\left( H_{h^2 \lambda^{(i)} \lambda^{(j)},1}^{(n)}(\ob_{1} - \ob_{2}) + R_{h\lambda^{(j)}, h\lambda^{(i)}}^{(n)}(\ubb, \ob_{1} - \ob_{2}) \right).
\end{equation*} 
Substitute the above into (\ref{eq:DFT-exp00}), we get the desired result.
\end{proof}

Now we are ready to prove Thoerem \ref{thm:asymDFT-IRS}.

\begin{proof}[Proof of Thoerem \ref{thm:asymDFT-IRS}]

Let $\underline{J}_{h,n}(\cdot)$ and $I_{h,n}(\cdot)$ be the theoretical DFT and periodogram, respectively. Then, thanks to Theorem \ref{thm:asymp-IRS-fea2} below, it is enough to prove the statements for $\underline{J}_{h,n}(\cdot)$ and $I_{h,n}(\cdot)$ replacing $\underline{\widehat{J}}_{h,n}(\cdot)$ and $\widehat{I}_{h,n}(\cdot)$, respectively.

We first show (\ref{eq:lim-DFT-IRS1}) and (\ref{eq:lim-DFT-IRS2}). First, by using YG24, Theorem D.2, we have 
\begin{equation} \label{eq:DFTlim0}
\lim_{n\rightarrow \infty} \sup_{\ob_1, \ob_2 \in \R^d} (2\pi)^{-d}H_{h,2}^{-1}|D_n|^{-1} \int_{\R^d} |\ell_2^{(i,j)}(\ubb)| 
R_{h\lambda^{(j)}, h\lambda^{(i)}}^{(n)}(\ubb, \ob_{1} - \ob_{2})| d\ubb = 0.
\end{equation} Therefore, substitute (\ref{eq:DFTlim0}) into (\ref{eq:DFT-exp}), we have
\begin{equation}
\begin{aligned}
& \cov\big( J_{h,n}^{(i)}(\ob_{1}) , J_{h,n}^{(j)}(\ob_{2}) \big) = (2\pi)^{-d}H_{h,2}^{-1}|D_n|^{-1} \bigg( \delta_{i,j} H_{h^2\lambda^{(i)},1}^{(n)}(\ob_{1}-\ob_{2})\\
& \qquad \qquad \qquad 
+ H_{h^2 \lambda^{(i)} \lambda^{(j)},1}^{(n)}(\ob_{1} - \ob_{2}) \int_{\R^d}e^{-i \ubb^\top \ob_{1}} \ell_2^{(i,j)}(\ubb) d\ubb 
\bigg) +o(1), \quad n\rightarrow \infty,
\end{aligned}
\label{eq:DFT-exp2}
\end{equation} 
where the $o(1)$ bound above is uniform over $\ob_{1}, \ob_{2} \in \R^d$.

Next, by using YG24, Lemma C.2, we have
\begin{equation} \label{eq:DFTlim1}
\lim_{n \rightarrow \infty}|D_n|^{-1} |H_{h^2\lambda^{(i)},1}^{(n)}(\ob_{1,n}-\ob_{2,n})| = 0.
\end{equation} 
The same argument holds when replacing $h^2\lambda^{(i)}$ with $h^2\lambda^{(i)}\lambda^{(j)}$ above. Therefore, since $\ell_2^{(i,j)} \in L_1(\R^d)$ due to Assumption \ref{assum:B} ($k=2$), substitute (\ref{eq:DFTlim1}) into (\ref{eq:DFT-exp2}), we show (\ref{eq:lim-DFT-IRS1}).

To show (\ref{eq:lim-DFT-IRS2}), by using (\ref{eq:DFT-exp2}), we have
\begin{equation}
\begin{aligned}
& \cov\big( J_{h,n}^{(i)}(\ob_{n}) , J_{h,n}^{(j)}(\ob_{n}) \big) \\
&~~= (2\pi)^{-d}H_{h,2}^{-1}|D_n|^{-1} \left( \delta_{i,j}H_{h^2\lambda^{(i)},1}^{(n)} (\textbf{0}) 
+  H_{h^2 \lambda^{(i)} \lambda^{(j)},1}^{(n)}(\textbf{0}) \int_{\R^d}  e^{-i \ubb^\top \ob_{n}} \ell_2^{(i,j)}(\ubb) d\ubb \right) + o(1) \\
&~~= (2\pi)^{-d}H_{h,2}^{-1} H_{h^2\lambda^{(i)},1} \delta_{i,j} + H_{h,2}^{-1} \mathcal{F}^{-1} (\ell^{(i,j)})(\ob_{n}) + o(1) \\
&~~= (2\pi)^{-d}H_{h,2}^{-1} H_{h^2\lambda^{(i)},1} \delta_{i,j} + H_{h,2}^{-1} \mathcal{F}^{-1} (\ell^{(i,j)})(\ob_{}) + o(1), \quad n \rightarrow \infty,
\end{aligned}
\label{eq:covJ00}
\end{equation} 
where the $o(1)$ bound above is uniform over $\ob_{n}$.
Here, the last inequality is due to the (uniform) continuity of $\mathcal{F}^{-1} (\ell^{(i,j)})(\cdot)$, provided Assumption \ref{assum:B}(for $k=2$).
Therefore, in a matrix form, we have
\begin{equation}
\begin{aligned}
& \lim_{n\rightarrow \infty}
\var\left( \underline{J}_{h,n}(\ob_{n}) \right) \\
&~~ = (2\pi)^{-d} H_{h,2}^{-1} \diag (H_{h^2 \underline{\lambda}}) +
  H_{h,2}^{-1} \left( H_{h^2 \underline{\lambda}\cdot\underline{\lambda}^\top} \odot
\mathcal{F}^{-1}(L_2)(\ob) \right) = F_{h}(\ob).
\end{aligned}
\label{eq:Jconv-to-F}
\end{equation} 
This proves (\ref{eq:lim-DFT-IRS2}).

Next, we show (\ref{eq:lim-Per-IRS1}) and (\ref{eq:lim-Per-IRS2}). By using cumulant decomposition, for $i_1, j_1, i_2, j_2 \in \{1,\dots, m\}$, we have
\begin{equation} \label{eq:cum-Per}
\begin{aligned} 
&  \cov\big(I_{h,n}^{(i_1,j_1)}(\ob_{1}), I_{h,n}^{(i_2, j_2)}(\ob_{2})\big)  \\
&\quad = \cum\big( J_{h,n}^{(i_1)}(\ob_1), J_{h,n}^{(i_2)}(-\ob_2) \big) \cum\big(J_{h,n}^{(j_1)}(-\ob_1), J_{h,n}^{(j_2)}(\ob_2) \big) \\
&\quad ~~+ \cum\big( J_{h,n}^{(i_1)}(\ob_1), J_{h,n}^{(j_2)}(\ob_2) \big) \cum\big(J_{h,n}^{(j_1)}(-\ob_1), J_{h,n}^{(i_2)}(-\ob_2) \big) \\
&\quad ~~+ 
\cum\big( J_{h,n}^{(i_1)}(\ob_1),J_{h,n}^{(j_1)}(-\ob_1), J_{h,n}^{(i_2)}(-\ob_2) , J_{h,n}^{(j_2)}(\ob_2) \big) \\
&
\quad = \cum\big( J_{h,n}^{(i_1)}(\ob_1), J_{h,n}^{(i_2)}(-\ob_2) \big) \cum\big(J_{h,n}^{(j_1)}(-\ob_1), J_{h,n}^{(j_2)}(\ob_2) \big) \\
&\quad ~~+ \cum\big( J_{h,n}^{(i_1)}(\ob_1), J_{h,n}^{(j_2)}(\ob_2) \big) \cum\big(J_{h,n}^{(j_1)}(-\ob_1), J_{h,n}^{(i_2)}(-\ob_2) \big) + O(|D_n|^{-1}) \\
& \quad =: S_1 + S_2 +O(|D_n|^{-1}), \quad n\rightarrow \infty,
\end{aligned}
\end{equation} 
where the $O(|D_n|^{-1})$ bound above is uniform over $\ob_1, \ob_2 \in \R^d$. Here, the last identity is due to YG24, Lemma D.5.

To show (\ref{eq:lim-Per-IRS1}), since $\{\ob_{1,n}\}$ is asymptotically distant from $\{\ob_{2,n}\}$ and $\{-\ob_{2,n}\}$, by applying (\ref{eq:lim-DFT-IRS1}), we have
\begin{equation*}
\cov\big(I_{h,n}^{(i_1,j_1)}(\ob_{1}), I_{h,n}^{(i_2, j_2)}(\ob_{2})\big)
= S_1(\ob_{1,n}, \ob_{2,n}) + S_1(\ob_{1,n}, \ob_{2,n}) + O(|D_n|^{-1}) = o(1), \quad n\rightarrow \infty.
\end{equation*}
This shows (\ref{eq:lim-Per-IRS1}). 

To show (\ref{eq:lim-Per-IRS2}), since $\{\ob_{1,n}\}$ is asymptotically distant from $\{\textbf{0}\}$, 
by using (\ref{eq:lim-DFT-IRS1}) and (\ref{eq:lim-DFT-IRS2}), we have 
\begin{eqnarray*}
&&  \cov\big(I_{h,n}^{(i_1,j_1)}(\ob_{n}), I_{h,n}^{(i_2, j_2)}(\ob_{n})\big) \\
&&\quad = \cum\big( J_{h,n}^{(i_1)}(\ob_{n}), J_{h,n}^{(i_2)}(-\ob_{n}) \big) \cum\big(J_{h,n}^{(j_1)}(-\ob_{n}), J_{h,n}^{(j_2)}(\ob_{n}) \big) \\
&&\quad ~~+ \cum\big( J_{h,n}^{(i_1)}(\ob_{n}), J_{h,n}^{(j_2)}(\ob_{n}) \big) \cum\big(J_{h,n}^{(j_1)}(-\ob_{n}), J_{h,n}^{(i_2)}(-\ob_{n}) \big) + o(1) \\
&&\quad = F_h^{(i_1, i_2)}(\ob) F_h^{(j_1, j_2)}(-\ob) + o(1), \quad n \rightarrow \infty.
\end{eqnarray*} 
This shows (\ref{eq:lim-Per-IRS2}). All together, we get the desired results.
\end{proof}


\section{Proof of the main results}

\subsection{Proof of Theorem \ref{thm:local-S}} \label{sec:localS}

Let $\underline{\widetilde{X}}_{D_n}$ denote the intensity reweighted process of $\underline{X}_{D_n}$ as in Definition \ref{def:k-IRS}. Then, the first-order intensity function and covariance intensity function of $\underline{\widetilde{X}}_{D_n}$ are $(1, \dots, 1)^\top$ and $L_2(\cdot)$, respectively. For a fixed $\ubb \in [-1/2, 1/2]^d$, let $\underline{X}^{\ubb}$ be the weighted process (see Definition \ref{def:weight}) of $\underline{\widetilde{X}}_{D_n}$ with respect to the weight $\underline{\lambda}(\ubb) = (\lambda^{(1)}(\ubb), \dots, \lambda^{(m)}(\ubb))^\top$. Then, it is straightforward that $\underline{X}^{\ubb}$ is also an SOS process with first-order intensity $\underline{\lambda}(\ubb)$ and reduced covariance function $\left( \underline{\lambda}(\ubb)\underline{\lambda}(\ubb)^\top \right) \odot L_2(\cdot)$. Therefore, the complete covariance intensity function of $\underline{X}^{\ubb}$ is $C^{\ubb}$, and in turn, the corresponding spectrum is $F^{\ubb}(\ob)$ as defined in (\ref{eq:F-local}).  

Next, to show the identity, we first substitute (\ref{eq:F-local}) into the right-hand side. The first integral term is
\begin{equation*}
\frac{(2\pi)^{-d}}{H_{h,2}} \int_{[-1/2,1/2]^d} \diag \left(h(\ubb)^2 \lambda^{(1)}(\ubb), \dots, h(\ubb)^2 \lambda^{(m)}(\ubb) \right) d\ubb = \frac{(2\pi)^{-d}}{H_{h,2}} \diag (H_{h^2 \underline{\lambda}}).
\end{equation*}
The $(i,j)$th ($i,j \in \{1, \dots, m\}$) element of the second integral term is
\begin{equation*}
\frac{1}{H_{h,2}} \mathcal{F}^{-1}(\ell_2^{(i,j)})(\ob) \int_{[-1/2,1/2]^d} h(\ubb)^2 \lambda^{(i)}(\ubb)  \lambda^{(j)}(\ubb) d \ubb 
= \frac{1}{H_{h,2}} \mathcal{F}^{-1}(\ell_2^{(i,j)})(\ob)  H_{h^2 \lambda^{(i)} \lambda^{(j)},1}.
\end{equation*} 
Therefore, the second integral term can be written as $H_{h,2}^{-1} ( H_{h^2\underline{\lambda}\cdot\underline{\lambda}^\top} \odot
\mathcal{F}^{-1}(L_2)(\ob))$. 
By comparing these expressions with (\ref{eq:f-IRS-mat}), we obtain the desired result.
\hfill $\Box$


\subsection{Proof of Theorem \ref{thm:KSDE1}} \label{appen:proofKSDE}

Let
\begin{equation} \label{eq:KSDE-th}
F_{n,\bb}(\ob) = \int_{\R^d} K_{\bb}(\ob - \xb) I_{h,n}(\xb) d\xb, \quad \ob \in \R^d,
\end{equation}
be the theoretical counterpart of $\widehat{F}_{n,\bb}$, where $I_{h,n}(\cdot)$ denotes the theoretical periodogram defined in (\ref{eq:In-IRS}).

The following theorem, which is key in the proof of Theorem \ref{thm:KSDE1}, addresses the MSE convergence of $F_{n,\bb}(\ob)$.

\begin{theorem} \label{thm:KSDE-th}
Suppose that Assumptions \ref{assum:A}, \ref{assum:B} (for $k = 4$), \ref{assum:C}, \ref{assum:beta}, \ref{assum:beta2}(i), and \ref{assum:D} hold. Then,
\begin{equation*}
\lim_{n\rightarrow \infty} \Ex \big| F_{n,\bb}(\ob) - F_{h}(\ob) \big|^2 = 0, \quad \ob \in \R^d.
\end{equation*}
\end{theorem}
\begin{proof}
To show the result, we calculate the bias and variance of $F_{n,\bb}(\ob)$ seperately. First, by using the triangle inequality together with $\int_{\R^{d}} K_{\bb}(\ob) d\ob = 1$, the bias can be bounded by
\begin{equation} \label{eq:Fnb-bias0}
\begin{aligned}
&| \Ex[ F_{n,\bb}(\ob)] - F_h(\ob) | \\
&~~\leq
\left|\int_{\R^d} K_{\bb} (\ob - \xb) F_h(\xb) d\xb - F_h(\ob) \right|  + 
\left| \Ex[ F_{n,\bb}(\ob)] - \int_{\R^d} K_{\bb} (\ob - \xb) F_h(\xb) d\xb\right| \\
&~~= \left| \int_{\R^d} K_{\bb} (\xb) \left\{ F_h(\ob - \xb) - F_h(\ob) \right\} d\xb \right| +  \int_{\R^d} K_{\bb}(\ob - \xb) \big| \Ex[I_{h,n}(\xb)] - F_h(\xb)\big| d\xb \\
& = A_1 + A_2.
\end{aligned}
\end{equation}
We first bound $A_1$. Since $F_h(\cdot)$ is uniformly continuous due to Assumption \ref{assum:B} (for $k=2$), given $\varepsilon > 0$, there exists $\delta > 0$ such that 
\begin{equation*}
\sup_{\|\xb - \yb\| < \delta} |F_h(\xb) - F_h(\yb)| < \varepsilon.
\end{equation*}
Next, we note that the support of $K_{\bb}$ is $[-b_1, b_1] \times \cdots \times [-b_d, b_d]$, which shrinks in all coordinates as $n \to \infty$. Therefore, for sufficiently large $n \in \mathbb{N}$, we have $\operatorname{supp}(K_{\bb}) \subset B(\mathbf{0}; \delta)$, where $B(\mathbf{0}; \delta)$ denotes the centered ball in $\mathbb{R}^d$ with radius $\delta > 0$. Combining these two facts, for large enough $n \in \mathbb{N}$, we obtain
\begin{equation*}
A_1 \leq \int_{B(\mathbf{0}; \delta)} K_{\bb}(\xb)\, |F_h(\ob - \xb) - F_h(\ob)|\, d\xb \leq \varepsilon.
\end{equation*}
Since $\varepsilon > 0$ is arbitrary, we conclude that
\begin{equation} \label{eq:Fnb-bias1}
\lim_{n \to \infty} A_1 = 0.
\end{equation}
To bound $A_2$, we observe that $\lim_{n \to \infty} \sup_{\ob \in \mathbb{R}^d} \big| \mathbb{E}[I_{h,n}(\ob)] - F_h(\ob) \big| = 0$ due to (\ref{eq:covJ00}). Therefore, we have 
\begin{equation} \label{eq:Fnb-bias2}
A_2 \leq \sup_{\ob \in \mathbb{R}^d} \big| \mathbb{E}[I_{h,n}(\ob)] - F_h(\ob) \big| = o(1), \quad n \to \infty.
\end{equation}
Substituting (\ref{eq:Fnb-bias1}) and (\ref{eq:Fnb-bias2}) into (\ref{eq:Fnb-bias0}), we obtain
\begin{equation} \label{eq:Fnb-bias-all}
|\Ex[ F_{n,\bb}(\ob)] - F_h(\ob)| = o(1), \quad n \to \infty.
\end{equation}
Next, we bound the variance. Fix $i, j \in \{1, \dots, m\}$. Then,  
\begin{equation} \label{eq:varF-expression}
\var(F_{n,\bb}^{(i,j)}(\ob))
= \int_{\R^d} d\yb \, K_{\bb}(\ob - \yb)
  \int_{\R^d} d\xb \, K_{\bb}(\ob - \xb)
  \cov(I_{h,n}^{(i,j)}(\xb), I_{h,n}^{(i,j)}(\yb)).
\end{equation}  
First, under Assumption \ref{assum:D}, we can choose $r = r(n) \in (0,\infty)$ such that, as $n \to \infty$,  
\begin{equation*}
r \rightarrow 0,~~ |D_n| r^{d} \rightarrow \infty, ~~\text{and} \quad \max_j (r/b_j) \rightarrow 0.
\end{equation*}  
For $\xb \in \R^d$ and $r \in (0,\infty)$, let $B_{\infty}(\xb; r) = \{\yb: \|\xb - \yb\|_\infty \leq r\}$. 

For $\yb \neq \mathbf{0}$ and sufficiently large $n \in \N$, we can partition $\R^d$ into the disjoint sets $B_{\infty}(\yb; r)$, $B_{\infty}(-\yb; r)$, and the remaining set, denoted by $C(\yb; r)$.  
When $\yb = \mathbf{0}$, $\R^d$ can be partitioned into $B_{\infty}(\mathbf{0}; r)$ and the remaining set $C(\mathbf{0}; r)$.

Now, we bound the integral involving $C(\yb; r)$. If $\xb \in C(\yb; r)$, then  
\begin{equation*}
|D_n|^{1/d} \|\xb - \yb\|_{\infty} > |D_n|^{1/d} r \rightarrow \infty.
\end{equation*}  
Therefore, by (\ref{eq:DFT-exp2}), (\ref{eq:DFTlim1}), and (\ref{eq:cum-Per}), it follows that  
\begin{equation*}
\sup_{\yb \in \R^d} \sup_{\xb \in C(\yb; r)} 
\big| \cov(I_{h,n}^{(i,j)}(\xb), I_{h,n}^{(i,j)}(\yb)) \big| 
= o(1), \quad n \to \infty.  
\end{equation*}  
Hence,  
\begin{equation} 
\begin{aligned}
& \int_{\R^d} d\yb \, K_{\bb}(\ob - \yb)
   \int_{C(\yb; r)} d\xb \, K_{\bb}(\ob - \xb)
   \big| \cov(I_{h,n}^{(i,j)}(\xb), I_{h,n}^{(i,j)}(\yb)) \big| \\  
&\leq o(1) 
   \left( \int_{\R^d} d\yb \, K_{\bb}(\ob - \yb) \right)^2 
   = o(1), \quad n \to \infty.
\end{aligned}
\label{eq:varI-bound1}
\end{equation}  
Next, to bound the integral involving $B_{\infty}(\yb;r)$, we observe that if $\xb \in B_\infty(\yb;r)$, then
\begin{equation*}
\|(\xb-\yb)\oslash\bb\|_{\infty} \leq \max_j r/b_j \rightarrow 0.
\end{equation*}
Therefore, since $K(\cdot)$ is Lipschitz continuous on $\R^d$, for an arbitrary $\varepsilon \in (0,\infty)$, there exists $N = N(\varepsilon) \in \N$ such that if $n > N$, then
\begin{equation*}
\sup_{\ob, \yb \in \R^d} \sup_{\xb \in B_{\infty}(\yb;r)}
| K_{\bb}(\ob- \xb) - K_{\bb}(\ob - \yb)| \leq \varepsilon (b_1 \cdots b_d)^{-1}.
\end{equation*}
Moreover, by using the Cauchy-Schwarz inequality and YG24, Lemma D.5, we have
\begin{equation*}
\sup_{\xb, \yb\in \R^d} |\cov(I_{h,n}^{(i,j)}(\xb), I_{h,n}^{(i,j)}(\yb))|
\leq \sup_{\xb \in \R^d} \var (I_{h,n}^{(i,j)}(\xb)) = O(1), \quad n\rightarrow \infty.
\end{equation*}
Therefore, combining the above two inequalities, we obtain
\begin{equation} 
\begin{aligned}
& \int_{\R^d} d\yb K_{\bb}(\ob - \yb) \int_{B_{\infty}(\yb;r)} d\xb K_{\bb}(\ob - \xb) 
|\cov(I_{h,n}(\xb), I_{h,n}(\yb))| \\  
&\leq O(1) \int_{\R^d} d\yb K_{\bb}(\ob - \yb) \int_{B_{\infty}(\yb;r)} d\xb K_{\bb}(\ob - \xb) \\  
&\leq O(1) \int_{\R^d} d\yb K_{\bb}(\ob - \yb)^2 \int_{B_{\infty}(\yb;r)} d\xb  
+ O(1) \varepsilon (b_1 \cdots b_d)^{-1} \int_{\R^d} d\yb K_{\bb}(\ob - \yb) \int_{B_{\infty}(\yb;r)} d\xb \\  
&\leq O(1) r^{d} \int_{\R^d} d\yb K_{\bb}(\ob - \yb)^2  
+ O(1) \varepsilon (b_1 \cdots b_d)^{-1} r^d \int_{\R^d} d\yb K_{\bb}(\ob - \yb)  \\  
& = O( r^d (b_1 \cdots b_d)^{-1}) +  O(\varepsilon  r^d  (b_1 \cdots b_d)^{-1}) = o(1), \quad n\rightarrow \infty.
\end{aligned}
\label{eq:varI-bound2}
\end{equation}  
Here, we use $\int_{B_{\infty}(\yb;r)} d\xb = |B_{\infty}(\yb;r)| = O(r^d)$ in the third inequality.

Lastly, using similar techniques, we can show that
\begin{equation}
\int_{\R^d} d\yb \, K_{\bb}(\ob - \yb) 
\int_{B_{\infty}(-\yb; r)} d\xb \, K_{\bb}(\ob - \xb) 
\big| \cov(I_{h,n}(\xb), I_{h,n}(\yb)) \big| = o(1), \quad n \to \infty.
\label{eq:varI-bound3}
\end{equation}
Substituting (\ref{eq:varI-bound1})--(\ref{eq:varI-bound3}) into (\ref{eq:varF-expression}) and applying the triangle inequality, we conclude that, for $i, j \in \{1, \dots, m\}$,
\begin{equation*}
\lim_{n \to \infty} \var(F_{n,\bb}^{(i,j)}(\ob)) = 0.
\end{equation*}
Summing the above over $i$ and $j$, and combining with (\ref{eq:Fnb-bias-all}), we obtain the desired result.
\end{proof}

Now, we are ready to prove Theorem \ref{thm:KSDE1}.

\begin{proof}[Proof of Theorem \ref{thm:KSDE1}]
We first show (\ref{eq:Fnb-conv1}). By applying Markov's inequality to Theorem \ref{thm:KSDE-th}, we obtain 
$F_{n,\bb}(\ob) - F_{h}(\ob) = o_p(1)$ as $n \to \infty$. Furthermore, Theorem \ref{thm:KSDE-bound-a} below states that 
$\widehat{F}_{n,\bb}(\ob) - F_{n,\bb}(\ob) = o_p(1)$. 
Combining these two results with the triangle inequality yields (\ref{eq:Fnb-conv1}).  

Next, the $L_2$-convergence result in (\ref{eq:Fnb-conv2}) follows immediately from 
Theorems \ref{thm:KSDE-th} and \ref{thm:KSDE-bound-b}. 

Altogether, we obtain the desired results.
\end{proof}

\subsection{Proof of Corollary \ref{coro:KSDE2}} \label{appen:proofKSDE2}
\begin{proof}
Let $\ob \in \R^{d}$ be such that the number of its nonzero elements, denoted by $|\ob|_0$, is greater than $d/4$, where $d \in \N$ is the spatial dimension of the point process. Then, by applying Theorems \ref{theorem:Fij-L2} (Equation (\ref{eq:E3-1})) and \ref{thm:KSDE-bound-c} (Equation (\ref{eq:E4-1})), together with Markov's inequality and the triangle inequality, we obtain 
$\widehat{F}_{n,\bb}(\ob) - F_{n,\bb}(\ob) = o_p(1)$. 
Combining this result with Theorem \ref{thm:KSDE-th}, we obtain the analogous consistency results under a misspecified intensity model.

Next, the $L_2$-convergence results follow from Theorems \ref{theorem:Fij-L2} (Equation (\ref{eq:E3-1})), 
\ref{thm:KSDE-bound-c} (Equation (\ref{eq:E4-2})), and \ref{thm:KSDE-th}.

Altogether, we obtain the desired results.
\end{proof}

\subsection{Proof of convergence results for Reimann sum version} \label{appen:proofR}

Recall the Riemann sum version of the kernel spectral density estimator:
\begin{equation} \label{eq:FRob}
\widehat{F}^{(R)}_{n,\bb}(\ob)
= \frac{\sum_{\kb \in \Z^d} K_{\bb}(\ob - \xb_{\kb, \bOmega}) \widehat{I}_{h,n}(\xb_{\kb, \bOmega})}
       {\sum_{\kb \in \Z^d} K_{\bb}(\ob - \xb_{\kb, \bOmega})},
\end{equation}
where $\xb_{\kb, \bOmega} = (2\pi k_1 / \Omega_1, \dots, 2\pi k_d / \Omega_d)^\top$ for a given grid vector $\bOmega = (\Omega_1, \dots, \Omega_d)^\top$.

In the following corollary, we show that $\widehat{F}^{(R)}_{n,\bb}(\ob)$ also converges to the pseudo-spectrum.

\begin{corollary}\label{thm:KSDE2}
Suppose that the grid vector $\bOmega$ increases with $n \in \N$ and takes the form $\bOmega = c \, \aB(n)$ for some constant $c \in (0,\infty)$. Then, under the same assumptions as in Theorem \ref{thm:KSDE1} and Corollary \ref{coro:KSDE2}, the convergence results stated therein also hold for $\widehat{F}^{(R)}_{n,\bb}(\ob)$ replacing $\widehat{F}_{n,\bb}(\ob)$.
\end{corollary}
\begin{proof}
For brevity, we only show the results for the correctly specified bias case (as in Theorem \ref{thm:KSDE1}); the misspecified bias case (as in Corollary \ref{coro:KSDE2}) can be treated similarly (details are omitted).

To show the convergence results in (\ref{eq:Fnb-conv1}) and (\ref{eq:Fnb-conv2}) for $\widehat{F}^{(R)}_{n,\bb}$, we define $F^{(R)}_{n,\bb}$ to be the ideal (Riemann sum) version of $\widehat{F}^{(R)}_{n,\bb}(\ob)$ by replacing $\widehat{I}_{h,n} (\xb_{\kb, \bOmega})$ with the theoretical periodogram $I_{h,n} (\xb_{\kb, \bOmega})$ in (\ref{eq:FRob}). We will first show
\begin{equation} \label{eq:thmC1}
\lim_{n \rightarrow \infty} \Ex\big| F^{(R)}_{n,\bb}(\ob) - F_h(\ob) \big|^2 = 0, \quad \ob \in \R^{d},
\end{equation}
which is the discrete analogue of Theorem \ref{thm:KSDE-th}. The proof of this is almost identical to that of Theorem \ref{thm:KSDE-th}, except that all integrals of the form $\int_{\R^d} K_{\bb}( \ob - \xb ) g(\xb) d\xb$ 
for various functions $g$ should be replaced by
\begin{equation*}
\int_{\R^d} K_{\bb}(\ob - \xb) g(\xb) d\mu(\xb) = \frac{ \sum_{ \kb\in\Z^d } K_{\bb}(\ob - \xb_{\kb, \bOmega}) g(\xb_{\kb, \bOmega}) }{ \sum_{ \kb\in\Z^d } K_{\bb}(\ob - \xb_{\kb, \bOmega}) },
\end{equation*}
where $\mu$ is the Dirac measure on $\{\xb_{\kb, \bOmega}: \kb \in \Z^d\}$ with point mass
$\mu( \{\xb_{\kb, \bOmega}\} ) =1/ \{ \sum_{ \kb\in\Z^d } K_{\bb}(\ob-\xb_{\kb, \bOmega}) \}$.

Then, all arguments in the proof of Theorem \ref{thm:KSDE-th} remain valid by replacing $d\xb$ with $d\mu(\xb)$, with a help from the following identity:
\begin{equation*}
\frac{1}{\text{vol}(\bOmega)} \sum_{\kb\in\Z^d} K_{\bb}(\ob - \xb_{\kb, \bOmega}) = \int_{\R^d} K_{\bb}(\ob - \xb) d\xb + o(1) = 1 + o(1), \quad n \rightarrow \infty,
\end{equation*}
where $\text{vol}(\bOmega) = \prod_{j=1}^{d} \Omega_j(n)$.

Next, we bound the difference $\big| \widehat{F}^{(R)}_{n,\bb}(\ob) - F^{(R)}_{n,\bb}(\ob) \big|$. In particular, under appropriate conditions, we have, as $n \rightarrow \infty$,
\begin{eqnarray}
\sup_{\ob \in \R^d} \big| \widehat{F}^{(R)}_{n,\bb}(\ob) - F^{(R)}_{n,\bb}(\ob) \big| &=& o_p(1), \label{eq:thmD3} \\
\sup_{\ob \in \R^d} \Ex \big| \widehat{F}^{(R)}_{n,\bb}(\ob) - F^{(R)}_{n,\bb}(\ob) \big|^2 &=& O\big( |D_n|^{-1} (b_1 \cdots b_d)^{-1} \big) + O\big( |D_n|^{1-r/4} \big), \label{eq:thmD4}
\end{eqnarray}
which are the discrete analogues of Theorems \ref{thm:KSDE-bound-a} and \ref{thm:KSDE-bound-b} below.

The proofs of (\ref{eq:thmD3}) and (\ref{eq:thmD4}) are almost identical to those of Theorems \ref{thm:KSDE-bound-a} and \ref{thm:KSDE-bound-b}. Here, again, the main modification is to replace the Lebesgue integral $d\xb$ by the Dirac measure $d\mu(\xb)$ introduced above. Another key step is to show the following:
\begin{equation} \label{eq:D3-R}
\sup_{\ob \in \R^{d}} \int_{\R^d} K_{\bb}(\ob-\xb) \left\{ |D_n|^{-1} \big| H_{\underline{g}}^{(n)}(\xb) \big| \right\}^2 d\mu(\xb) 
= O\big(|D_n|^{-1} (b_1 \cdots b_d)^{-1}\big), \quad n\rightarrow \infty,
\end{equation} where $\underline{g} = h\partial\underline{\lambda}$, which is a discrete setting of Equation (\ref{eq:KHbound}) below. For simplicity, assume $\underline{\lambda}$ and $\bbeta$ are univariate, so that $\underline{g} = g$ is univariate. To show (\ref{eq:D3-R}), observe that
\begin{equation*}
|D_n|^{-1} H_g^{(n)}(\xb) = \frac{1}{|D_n|} \int_{D_n} g(\yb\oslash\aB) \exp(-i\yb^\top \xb) d\yb
= \int_{[-1/2,1/2]^d} g(\zb) \exp(-i\zb^\top (\aB\odot\xb)) d\zb,
\end{equation*}
where the second identity follows from the change of variables $\yb\oslash\aB = \zb$. Combining this with $\sup_{\ob} K_{\bb}(\ob) < C(b_1 \cdots b_d)^{-1}$ and 
$\sum_{\kb\in\Z^d} K_{\bb}(\ob - \xb_{\kb, \bOmega}) = \text{vol}(\bOmega) (1 + o(1)) = O(|D_n|)$, we obtain
\begin{equation} \label{eq:D3-R2}
\begin{aligned}
& \sup_{\ob \in \R^{d}} \int_{\R^d} K_{\bb}(\ob-\xb) \left\{ |D_n|^{-1} \big| H_g^{(n)}(\xb) \big| \right\}^2 d\mu(\xb) \\
& \quad \leq O\big(|D_n|^{-1} (b_1 \cdots b_d)^{-1}\big) \sum_{\kb\in\Z^d} \left| \int_{[-1/2,1/2]^d} g(\zb) \exp\big(-i (2\pi/c) \zb^\top \kb\big) d\zb \right|^2.
\end{aligned}
\end{equation}
Here, we use the condition $\aB\odot\xb_{\kb, \bOmega} = 2\pi c^{-1}\kb$ in the inequality. To bound the summand in (\ref{eq:D3-R2}), we extend $g(\xb) = 0$ for $\xb \notin [-1/2,1/2]^d$ and, using the change of variables $\xb = (2\pi/c)\zb$, we have
\begin{equation*}
\begin{aligned}
&\sum_{\kb\in\Z^d} \left| \int_{[-1/2,1/2]^d} g(\zb) \exp(-i (2\pi/c) \zb^\top \kb ) d\zb \right|^2 \\
&\quad = (2\pi c)^d \sum_{\kb\in\Z^d} \left| \frac{1}{(2\pi)^d} \int_{\R^d} g\left(\frac{c\xb}{2\pi}\right) \exp(-i \xb^\top \kb) d\xb \right|^2 \\
&\quad = (2\pi c)^d \int_{\R^d} \left| g\left(\frac{c\xb}{2\pi}\right) \right|^2 d\xb 
= (2\pi)^{2d} \int_{[-1/2,1/2]^d} |g(\yb)|^2 d\yb = O(1).
\end{aligned}
\end{equation*}
Here, the second identity follows from Parseval's identity for the compactly supported function $g_c(\xb) = g(c\xb/(2\pi))$ and the third identity is due to the change of variables $\yb = c\xb/(2\pi)$. Substituting this bound into (\ref{eq:D3-R2}) shows (\ref{eq:D3-R}). Once (\ref{eq:D3-R}) is proved, the proofs of (\ref{eq:thmD3}) and (\ref{eq:thmD4}) follow almost identically to those of Theorems \ref{thm:KSDE-bound-a} and \ref{thm:KSDE-bound-b}.

Finally, (\ref{eq:thmC1}) and (\ref{eq:thmD3}) imply $\widehat{F}_{n,\bb}^{(R)}(\ob) \Pcon F_h(\ob)$, and (\ref{eq:thmC1}) together with (\ref{eq:thmD4}) yields \\ $\lim_{n\rightarrow\infty} \Ex \big| \widehat{F}_{n,\bb}^{(R)}(\ob) - F_h(\ob) \big|^2 = 0$. Thus, we get the desired results.
\end{proof}


\section{Convergence rate of the MSE of $\widehat{F}_{n,\bb}(\ob)$}

In some situations, it is necessary to quantify the rate at which $\widehat{F}_{n,\bb}(\ob)$ converges to $F_h(\ob)$. For instance, the MSE convergence rate of $\widehat{F}_{n,\bb}(\ob)$ is used to determine the optimal bandwidth in Section \ref{sec:opt}. However, because $\widehat{F}_{n,\bb}$ involves estimating the parametric form of the first-order intensity, directly computing its moments is infeasible. Instead, we consider the moments of its theoretical counterpart, denoted by $F_{n,\bb}(\ob)$ in (\ref{eq:KSDE-th}). 

The following theorem addresses the asymptotic order of the bias and variance of $F_{n,\bb}(\ob)$.
\begin{theorem} \label{thm:KSDE-bias-var}
Suppose that Assumptions \ref{assum:A}, \ref{assum:B} (for $k = 4$), \ref{assum:C}, \ref{assum:D}, and \ref{assum:E} hold. Moreover, assume that the data taper $h$ is Lipschitz continuous on $[-1/2,1/2]^d$ and the side lengths $\aB(n)$ satisfy condition \textrm{(SL)}. Then, the following two assertions hold:
\begin{itemize}
    \item[(i)] $\sup_{\ob \in \R^d}|\Ex [F_{n,\bb}(\ob)]  - F_h(\ob)| = O(\|\bb\|^2 + |D_n|^{-1/d})$, $n\rightarrow \infty$.
    \item[(ii)] $\sup_{\ob \in \R^d} \var |F_{n,\bb}(\ob)| = O(|D_n|^{-1} (b_1 \cdots b_d)^{-1})$, $n\rightarrow \infty$.
\end{itemize}
\end{theorem}
\begin{proof}
See Appendix \ref{appen:proofKSDE-MSE}.
\end{proof}

\begin{remark} \label{rmk:bias}
We observe from Theorem \ref{thm:KSDE-bias-var}(i) that $F_{n,\bb}(\ob)$ has two sources of bias. The $O(\|\bb\|^2)$ bias arises from the classical nonparametric kernel smoothing, whereas the $O(|D_n|^{-1/d})$ bias originates from the bias of the periodogram. Specifically, in (\ref{eq:Ibias-uniform}) below, we show that
\begin{equation*}
\sup_{\ob \in \R^d} |\Ex[I_{h,n}(\ob)] - F_h(\ob)| = O(|D_n|^{-1/d}), \quad n \rightarrow \infty.
\end{equation*}
This $O(|D_n|^{-1/d})$ bias transfers to the bias of the kernel estimator.
\end{remark}

Now, we assume a common bandwidth for all coordinates, i.e., $b_1 = \cdots = b_d = b \in (0,\infty)$. Recall
\begin{equation*}
\text{MSE}(b) = \sup_{\ob \in \R^d} \Ex \big|\widehat{F}_{n,b}(\ob) - F_h(\ob)\big|^2.
\end{equation*}
The following theorem addresses the convergence rate of $\text{MSE}(b)$ under the correctly specified first-order intensity model.

\begin{theorem} \label{thm:MSEb-rate}
Suppose that Assumptions \ref{assum:A}, \ref{assum:B} (for $k = 4$), \ref{assum:C}, \ref{assum:beta}, \ref{assum:beta2}(ii) (for $r = 8$), \ref{assum:D}, and \ref{assum:E} hold. Moreover, assume that the true first-order intensity is $\underline{\lambda}(\xb;\bbeta_0)$, the data taper $h$ is Lipschitz continuous on $[-1/2,1/2]^d$, and the side lengths $\aB(n)$ satisfy condition (SL). Then, 
\begin{equation*}
\text{MSE}(b) = O(b^{4} + |D_n|^{-2/d} + |D_n|^{-1} b^{-d}), \quad n \rightarrow \infty.
\end{equation*}
\end{theorem}
\begin{proof}
Using the inequality $\Ex |a+b|^2 \leq 2\{\Ex |a|^2 + \Ex |b|^2\}$, we have
\begin{equation*}
\text{MSE}(b) \leq 2 \sup_{\ob \in \R^d} \Ex |F_{n,b}(\ob) - F_h(\ob)|^2 + 2 \sup_{\ob \in \R^d} \Ex |\widehat{F}_{n,b}(\ob) - F_{n,b}(\ob)|^2.
\end{equation*}
By Theorem \ref{thm:KSDE-bias-var}, the first term above is $O(b^{4} + |D_n|^{-2/d} + |D_n|^{-1} b^{-d})$.
Moreover, Theorem \ref{thm:KSDE-bound-b} below shows that the second term in $\text{MSE}(b)$ is bounded by 
$O(|D_n|^{-1} b^{-d}) + O(|D_n|^{-1}) = O(|D_n|^{-1} b^{-d})$. Combining these two bounds, we get the desired result.
\end{proof}

Now, we consider the case when the first-order intensity is incorrectly specified. In this situation, as noted in Remark \ref{rmk:near-origin}, $\text{MSE}(b) = \infty$ due to the large peak at the origin. Therefore, we focus on the MSE at frequencies away from the origin: for $\delta > 0$, define
\begin{equation*}
\text{MSE}(b;\delta) = \sup_{\ob: \|\ob\|_{\infty} > \delta} \Ex |\widehat{F}_{n,b}(\ob) - F_h(\ob)|^2.
\end{equation*}
The following corollary addresses the convergence rate of $\text{MSE}(b;\delta)$ to zero.

\begin{corollary} \label{coro:MSEb-rate}
Suppose that Assumptions \ref{assum:A}, \ref{assum:B} (for $k = 4$), \ref{assum:C}, \ref{assum:beta}, \ref{assum:beta2}(ii) (for $r =8$), \ref{assum:D}, \ref{assum:E}, and \ref{assum:smooth} hold. Moreover, assume that side lengths $\aB(n)$ satisfy condition (SL).
Let $d$ equals to one or two. Then, for any fixed $\delta>0$,
\begin{equation*}
\text{MSE}(b;\delta) = O(b^{4} + |D_n|^{-2/d} + |D_n|^{-1}b^{-d}), \quad n\rightarrow \infty.
\end{equation*}
\end{corollary}
\begin{proof}
By using arguments similar to those in the proof of Theorem \ref{thm:MSEb-rate}, we have
\begin{equation*}
\text{MSE}(b;\delta) \leq O(b^{4} + |D_n|^{-2/d} + |D_n|^{-1} b^{-d}) + 
2 \sup_{\|\ob\|_{\infty} > \delta} \Ex |\widehat{F}_{n,b}(\ob) - F_{n,b}(\ob)|^2.
\end{equation*}
Next, by applying Theorems \ref{theorem:Fij-L2} (Equation (\ref{eq:E3-2})) and \ref{thm:KSDE-bound-c} (Equation (\ref{eq:E4-3})) below, the second term above is $O(|D_n|^{-1} b^{-d} + |D_n|^{1-r/4} + |D_n|^{1-4/d})$ as $n \rightarrow \infty$. 

Therefore, for $r = 8$ and $d \in \{1,2\}$, the second term is $O(|D_n|^{-1} b^{-d})$. Substituting this into the bound for $\text{MSE}(b;\delta)$ yields the desired result.
\end{proof}

\subsection{Proof of Theorem \ref{thm:KSDE-bias-var}} \label{appen:proofKSDE-MSE}
\begin{proof}
We first show the order of the bias. Recall from (\ref{eq:Fnb-bias0}) that the bias term can be decomposed into two components, $A_1$ and $A_2$. To calculate the rate of convergence of $A_1$, we apply a Taylor expansion of $F_h$ to obtain
\begin{equation*}
\begin{aligned}
\bigg| F_h(\vbb) - F_h(\ob) + \sum_{\ell=1}^{d} \frac{\partial F_h(\ob)}{\partial \omega_\ell} (v_\ell - \omega_\ell) \bigg|
&\leq \sup_{\ob \in \R^d, 1 \leq \ell, k \leq d} \left| \frac{\partial^2 F_h(\ob)}{\partial \omega_\ell \partial \omega_k} \right| \|\vbb - \ob\|^2 \
&\leq C \|\vbb - \ob\|^2,
\end{aligned}
\end{equation*}
for any $\vbb = (v_1, \dots, v_d)^\top$ and $\ob = (\omega_1, \dots, \omega_d)^\top$ in $\R^d$.
Here, the second inequality follows from the boundedness of the second-order partial derivatives of $F_h$.
Using this bound, we have
\begin{eqnarray*}
A_1 &=& \left| \int_{\R^d} K_{\bb}(\xb) \left\{ F_h(\ob - \xb) - F_h(\ob) \right\} d\xb \right| \\
&=& \left| \int_{\R^d} K_{\bb}(\xb) \left\{ F_h(\ob - \xb) - F_h(\ob) - \sum_{\ell=1}^{d} \frac{\partial F_h(\ob)}{\partial \omega_\ell} x_\ell \right\} d\xb \right| \\
&\leq& C \int_{\R^d} \|\xb\|^2 K_{\bb}(\xb) d\xb,
\end{eqnarray*}
where the second equality is due to the symmetry of $K_{\bb}$.
Since the support of $K_{\bb}$ is $[-b_1, b_1] \times \cdots \times [-b_d, b_d]$, we have
\begin{equation} \label{eq:A1-order}
A_1 \leq C \int_{\R^d} \|\xb\|^2 K_{\bb}(\xb) d\xb \leq C \|\bb\|^2 \int_{\R^d} K_{\bb}(\xb) d\xb = O(\|\bb\|^2), \quad n \rightarrow \infty.
\end{equation}

Next, we bound $A_2$. To do so, we need a rate of convergence for $|\Ex[I_{h,n}(\ob)] - F_h(\ob)|$. By using Lemma \ref{lemma:DFT-expr}, we have
\begin{equation} \label{eq:Ex-I-1}
\Ex[I_{h,n}^{(i,j)}(\ob)] - F_h^{(i,j)}(\ob) =
(2\pi)^{-d} H_{h,2}^{-1} |D_n|^{-1} \int_{\R^d} e^{-i \ubb^\top \ob} 
\ell_2^{(i,j)}(\ubb)  R_{h\lambda^{(j)}, h\lambda^{(i)}}^{(n)}(\ubb, \textbf{0})  d\ubb.
\end{equation}
Next, Theorem C.1 of YG24 states that there exists $C \in (0,\infty)$ such that
\begin{equation*}
|D_n|^{-1} \big| R_{h\lambda^{(j)}, h\lambda^{(i)}}^{(n)}(\ubb, \textbf{0}) \big| 
\leq C \min\{ \|\ubb \oslash \aB\|, 1\}, \quad \ubb \in \R^d.
\end{equation*}
Therefore, summing (\ref{eq:Ex-I-1}) over $i,j \in \{1, \dots, m\}$, we have
\begin{equation*}
\begin{aligned}
|\Ex[I_{h,n}(\ob)] - F_h(\ob)| 
&\leq C \int_{\R^d} |L_2(\ubb)| \min\{ \|\ubb \oslash \aB\|, 1\} d\ubb \\
&\leq C \int_{\ubb: \|\ubb \oslash \aB\| \geq 1} |L_2(\ubb)| d\ubb 
+ C \int_{\ubb: \|\ubb \oslash \aB\| \leq 1} \|\ubb \oslash \aB\|\, |L_2(\ubb)| d\ubb.
\end{aligned}
\end{equation*}
Now, from the condition (SL) on $\aB = (A_1, \dots, A_d)$, $\|\ubb \oslash \aB\| \geq 1$ implies that there exists $C_1 \in (0,\infty)$ such that $\|\ubb\| \geq  C_1 |D_n|^{1/d}$. Therefore, the first integral term is bounded by
\begin{equation} \label{eq:Ibias-1}
\begin{aligned}
& C \int_{\ubb:  \|\ubb \oslash \aB\| \geq 1} 
\left\{ \frac{\|\ubb\|}{C_1 |D_n|^{1/d}} \right\} |L_2(\ubb)| d\ubb \\
&\leq C |D_n|^{-1/d} \int_{\R^d} \|\ubb\| |L_2(\ubb)| d\ubb 
= O(|D_n|^{-1/d}), \quad n \to \infty.
\end{aligned}
\end{equation}
Here, we use $\|\ubb\| |L_2(\ubb)| \in L_1(\R^d)$, provided that Assumption \ref{assum:E}(ii) holds.

To bound the second term, (SL) also implies that there exists $C_2 \in (0,\infty)$ such that $\|\ubb \oslash \aB\| \leq C_2 |D_n|^{-1/d} \|\ubb\|$. Therefore, the second term is bounded by
\begin{equation} \label{eq:Ibias-2}
C_2 |D_n|^{-1/d} \int_{\ubb: \|\ubb \oslash \aB\| \leq 1} \|\ubb\| |L_2(\ubb)| d\ubb = O(|D_n|^{-1/d}), \quad n \to \infty.
\end{equation}
Combining (\ref{eq:Ibias-1}) and (\ref{eq:Ibias-2}), we get
\begin{equation} \label{eq:Ibias-uniform}
\sup_{\ob \in \R^d} |\Ex[I_{h,n}(\ob)] - F_h(\ob)| = O(|D_n|^{-1/d}), \quad n \to \infty.
\end{equation}
Therefore, $A_2$ is bounded by
\begin{equation} \label{eq:A2-order}
\begin{aligned}
A_2 &= C \int_{\R^d} K_{\bb} (\ob-\xb) |\Ex[I_{h,n}(\xb)] - F_h(\xb)| d\xb \\
&= O(|D_n|^{-1/d}) \int_{\R^d} K_{\bb} (\ob-\xb) d\xb = O(|D_n|^{-1/d}), \quad n \to \infty.
\end{aligned}
\end{equation}
Finally, combining (\ref{eq:A1-order}) and (\ref{eq:A2-order}), we have
\begin{equation*}
\sup_{\ob \in \R^d} |\Ex [F_{n,\bb}(\ob)] - F_h(\ob)| = O(\|\bb\|^2 + |D_n|^{-1/d}), \quad n \to \infty.
\end{equation*}
This proves the bias bound.

\vspace{0.5em}

Now, we prove the order of the variance. By using the cumulant decomposition as in (\ref{eq:cum-Per}), for $i,j \in \{1, \dots, m\}$, we have
\begin{equation*}
\begin{aligned}
\var(F_{n,\bb}^{(i,j)}(\ob)) 
&= \int_{\R^{2d}} K_{\bb}(\ob-\ob_1) K_{\bb}(\ob-\ob_2) \bigg\{ 
\cum( J_{h,n}^{(i)}(\ob_1), J_{h,n}^{(i)}(-\ob_2)) \cum(J_{h,n}^{(j)}(-\ob_1), J_{h,n}^{(j)}(\ob_2)) \\
&\quad + \cum( J_{h,n}^{(i)}(\ob_1), J_{h,n}^{(j)}(\ob_2)) \cum(J_{h,n}^{(j)}(-\ob_1), J_{h,n}^{(i)}(-\ob_2)) + O(|D_n|^{-1}) \bigg\}  d\ob_1 d\ob_2 \\
&= B_1(\ob) + B_2(\ob) + O(|D_n|^{-1}), \quad n \rightarrow \infty.
\end{aligned}
\end{equation*}
We calculate the bound for $B_1(\ob)$, and the bound for $B_2(\ob)$ can be treated similarly. Let 
\begin{equation*}
\widetilde{F}(\ob) =  (\widetilde{f}^{(i,j)}(\ob))_{1\leq i,j \leq m} = \mathcal{F}^{-1}(L_2)(\ob).
\end{equation*}
Then, Assumption \ref{assum:E}(i) implies that  $\widetilde{F} \in L_1^{m \times m}(\R^d)$, and in turn,
\begin{equation*}
L_2(\xb) = \mathcal{F}(\widetilde{F})(\xb) = \int_{\R^d} \widetilde{F}(\ob) e^{i\xb^\top\ob} d\ob.
\end{equation*}
Substituting the above into the cumulant expression in (\ref{eq:DFT-exp00}), we get
\begin{equation*}
\begin{aligned}
(2\pi)^{d}H_{h,2}|D_n| \cum(J_{h,n}^{(i)}(\ob_1) , J_{h,n}^{(j)}(\ob_2)) 
&= \delta_{i,j} H_{h^2\lambda^{(i)}}^{(n)}(\ob_1 + \ob_2) \\
&\quad + \int_{\R^d} \widetilde{f}^{(i,j)}(\bxi) H_{h\lambda^{(i)},1}^{(n)}(\ob_1 - \bxi) H_{h\lambda^{(j)},1}^{(n)}(\ob_2 + \bxi) d\bxi.
\end{aligned}
\end{equation*}
For the notational convenience, we write $H_{h^2\lambda^{(i)},1}^{(n)} = H_{h^2\lambda^{(i)}}$, $H_{h\lambda^{(i)},1}^{(n)} = H_{h \lambda^{(i)}}$, $\cdots$. Then, we have $B_1(\ob) = 
(2\pi)^{-2d} H_{h,2}^{-2} \{
B_{11}(\ob) + B_{12}(\ob) + B_{13}(\ob) + B_{14}(\ob)\}$, where
\begin{equation*}
\begin{aligned}
B_{11}(\ob) &= |D_n|^{-2} \iint_{\R^{2d}} K_{\bb}(\ob - \ob_1) K_{\bb}(\ob - \ob_2) H_{h^2\lambda^{(i)}}(\ob_1 - \ob_2) H_{h^2\lambda^{(j)}}(-\ob_1 + \ob_2) d\ob_1 d\ob_2, \\
B_{12}(\ob) &= |D_n|^{-2} \iint_{\R^{2d}} d\ob_1 d\ob_2 K_{\bb}(\ob - \ob_1) K_{\bb}(\ob - \ob_2) H_{h^2\lambda^{(i)}}(\ob_1 - \ob_2), \\
& \quad \times \int_{\R^d} \widetilde{f}^{(j,j)}(\bxi) H_{h\lambda^{(j)}}(-\ob_1 - \bxi) H_{h\lambda^{(j)}}(\ob_2 + \bxi) d\bxi, \\
B_{13}(\ob) &= |D_n|^{-2} \iint_{\R^{2d}} d\ob_1 d\ob_2 K_{\bb}(\ob - \ob_1) K_{\bb}(\ob - \ob_2) H_{h^2\lambda^{(j)}}(-\ob_1 + \ob_2), \\
& \quad \times \int_{\R^d} \widetilde{f}^{(i,i)}(\bxi) H_{h\lambda^{(i)}}(\ob_1 - \bxi) H_{h\lambda^{(i)}}(-\ob_2 + \bxi) d\bxi, \\
B_{14}(\ob) &= |D_n|^{-2} \iint_{\R^{2d}} K_{\bb}(\ob - \ob_1) K_{\bb}(\ob - \ob_2) \int_{\R^d} \widetilde{f}^{(i,i)}(\bxi_1) H_{h\lambda^{(i)}}(\ob_1 - \bxi_1) H_{h\lambda^{(i)}}(-\ob_2 + \bxi_1) d\bxi_1 \\
& \quad \times \int_{\R^d} \widetilde{f}^{(j,j)}(\bxi_2) H_{h\lambda^{(j)}}(-\ob_1 - \bxi_2) H_{h\lambda^{(j)}}(\ob_2 + \bxi_2) d\bxi_2 d\ob_1 d\ob_2.
\end{aligned}
\end{equation*}
Since $K_{\bb}$ is bounded and has finite support, we can define $\widehat{K}_{\bb}(\ob) = \mathcal{F}^{-1}(K_{\bb})(\ob)$. Then, by substituting the expressions for $H_{h^2\lambda^{(i)}}(\cdot)$ and $H_{h^2\lambda^{(j)}}(\cdot)$ into $B_{11}(\ob)$ and applying Fubini's theorem, we have
\begin{equation} \label{eq:B11}
\begin{aligned}
B_{11}(\ob) 
&= (2\pi)^{2d} |D_n|^{-2} \iint_{D_n^2} |\widehat{K}_{\bb}(\xb - \yb)|^2 \left( h^2\lambda^{(i)}(\xb\oslash\aB) \right) \left( h^2\lambda^{(j)}(\yb\oslash\aB) \right) d\xb d\yb \\
&\leq C |D_n|^{-2} \int_{D_n} \int_{\R^d} |\widehat{K}_{\bb}(\xb - \yb)|^2 d\xb d\yb \\
&\leq C |D_n|^{-1} \int_{\R^d} |\widehat{K}_{\bb}(\xb)|^2 d\xb \\
&= C |D_n|^{-1} \int_{\R^d} |K_{\bb}(\ob)|^2 d\ob = O(|D_n|^{-1} (b_1 \cdots b_d)^{-1}), \quad n \rightarrow \infty.
\end{aligned}
\end{equation}
Here, we use Parseval's theorem in the second-to-last identity.

Next, we bound $B_{12}(\ob)$. By substituting the expression of $H_{h^2\lambda^{(i)}}$, ..., and applying Fubini's theorem, we have
\begin{equation*}
\begin{aligned}
B_{12}(\ob) &= (2\pi)^{2d} |D_n|^{-2} \int_{\R^d} d\bxi \widetilde{f}^{(j,j)}(\bxi) \iiint_{D_n^3} \widehat{K}_{\bb}(\xb-\yb) \widehat{K}_{\bb}(\zb-\xb) (h^2\lambda^{(i)})(\xb\oslash\aB) (h\lambda^{(j)})(\yb\oslash\aB) \\
& \quad \times (h\lambda^{(j)})(\zb\oslash\aB)e^{i(\yb-\zb)^\top\bxi} d\xb d\yb d\zb \\
&= (2\pi)^{2d} |D_n|^{-2} \int_{\R^d} d\bxi \widetilde{f}^{(j,j)}(\bxi) \int_{D_n} d\xb (h^2\lambda^{(i)})(\xb\oslash\aB) \left\{ \int_{D_n} \widehat{K}_{\bb}(\xb-\yb) (h\lambda^{(j)})(\yb\oslash\aB) e^{i\yb^\top\bxi} d\yb \right\} \\
&\times \left\{ \int_{D_n} \widehat{K}_{\bb}(\zb-\xb) (h\lambda^{(j)})(\zb\oslash\aB) e^{-i\zb^\top\bxi} d\zb \right\} d\bxi \\
&= (2\pi)^{2d} |D_n|^{-2} \int_{\R^d} d\bxi \widetilde{f}^{(j,j)}(\bxi) \int_{D_n} (h^2\lambda^{(i)})(\xb\oslash\aB) |g_j(\xb,\bxi)|^2 d\xb,
\end{aligned}
\end{equation*}
where
\begin{equation*}
g_j(\xb,\bxi) := \int_{\R^d} \widehat{K}_{\bb}(\xb-\yb) (h\lambda^{(j)})(\yb\oslash\aB) e^{i\yb^\top\bxi} d\yb.
\end{equation*}
Since $\sup_{\bxi} |f^{(j,j)}(\bxi)| \leq (2\pi)^{-d}\int_{\R^d} |\ell^{(j,j)}_2(\xb)| d\xb < \infty$ and $|h^2 \lambda^{(i)}|$ is bounded, $B_{12}$ is bounded by
\begin{equation} \label{eq:B12}
\begin{aligned}
B_{12}(\ob) &\leq 
O(|D_n|^{-2}) \int_{D_n} d\xb \int_{\R^d} |g_j(\xb,\bxi)|^2 d\bxi \\
&= O(|D_n|^{-2}) \int_{D_n} d\xb \int_{\R^d} \left| \widehat{K}_{\bb}(\xb-\yb) (h\lambda^{(j)}) (\yb\oslash\aB) \right|^2 d\yb \\
&\leq O(|D_n|^{-1}) \int_{\R^d} | \widehat{K}(\ubb)|^2 d\ubb = O(|D_n|^{-1} (b_1 \cdots b_d)^{-1}), \quad n\rightarrow \infty.
\end{aligned}
\end{equation}
Here, the first identity follows from Parseval's theorem for $g_j(\xb,\cdot)$ and the second identity follows from Parseval's theorem for $\widehat{K}_{\bb}$. Using similar techniques, we obtain
\begin{equation} \label{eq:B13}
B_{13}(\ob), B_{14}(\ob) = O(|D_n|^{-1} (b_1 \cdots b_d)^{-1}), \quad n \rightarrow \infty.
\end{equation}
Combining (\ref{eq:B11})--(\ref{eq:B13}), we conclude that $B_{1}(\ob) = O(|D_n|^{-1} (b_1 \cdots b_d)^{-1})$ as $n \rightarrow \infty$. Similarly, $B_{2}(\ob) = O(|D_n|^{-1} (b_1 \cdots b_d)^{-1})$ as $n \rightarrow \infty$. This implies that
\begin{equation*}
\var(F_{n,\bb}^{(i,j)}(\ob)) = B_1(\ob) + B_2(\ob) + O(|D_n|^{-1})
= O(|D_n|^{-1} (b_1 \cdots b_d)^{-1}), \quad n \rightarrow \infty.
\end{equation*}
Summing the above over $i,j \in \{1,\dots,m\}$ shows the variance bound.

Altogether, we obtain the desired result.
\end{proof}

\section{Asymptotic equivalence between the feasible and theoretical quantitites} \label{sec:feasible}

In this section, we show the asymptotic equivalence between the feasible and theoretical quantities. To do so, we first fix terms.

Let $\underline{J}_{h,n}(\ob)$ be the theoretical centered DFT as in (\ref{eq:In-IRS}). To obtain the feasible criterion of $\underline{J}_{h,n}(\ob)$, it is required to estimate the first-order intensity $\underline{\lambda}(\xb) = (\lambda^{(1)}(\xb), \dots, \lambda^{(m)}(\xb))^\top$ in (\ref{eq:J-expectation}). In the framework of Section \ref{sec:feasibleDFT}, the feasible criterion of $\lambda^{(j)}(\xb)$ is $\widehat{\lambda}^{(j)}(\xb) = \lambda^{(j)}(\xb; \widehat{\bbeta}_n)$, thus the corresponding feasible DFT can be defined similarly which we will denote it as $\underline{\widehat{J}}_{h,n}(\ob)$. Since $\widehat{\bbeta}_n$ estimates the pseudo-true value $\bbeta_0$, we also can define the $j$th ``ideal'' first-order intenisty as $\lambda_0^{(j)}(\xb) = \lambda^{(j)}(\xb;\bbeta_0)$ and the corresponding ideal DFT as $\underline{\widetilde{J}}_{h,n}(\ob)$. In case the parametric first-order intensity model is correctly specified, we have $\underline{\widetilde{J}}_{h,n}(\ob) = \underline{J}_{h,n}(\ob)$. However, $\underline{\widetilde{J}}_{h,n}(\ob)$ is, in general, not necessarily equal to the theoretical centered DFT.

Lastly, let $I_{h,n}(\ob) = \underline{J}_{h,n}(\ob) \underline{J}_{h,n}(\ob)^*$, $\widehat{I}_{h,n}(\ob) = \underline{\widehat{J}}_{h,n}(\ob) \underline{\widehat{J}}_{h,n}(\ob)^*$, and $\widetilde{I}_{h,n}(\ob) = \underline{\widetilde{J}}_{h,n}(\ob) \underline{\widetilde{J}}_{h,n}(\ob)^*$ be the corresponding theoretical, feasible, and ideal periodogram, respectively.

\subsection{Bounds for the DFTs and periodograms}

In this section, we show that $|\underline{\widehat{J}}_{h,n}(\ob) - \underline{\widetilde{J}}_{h,n}(\ob)|$ and $|\widehat{I}_{h,n}(\ob) - \widetilde{I}_{h,n}(\ob)|$ converges to zero in two different modes of convergence.

Recall the parametric first-order intensity model $\underline{\lambda}(\cdot;\bbeta) = (\lambda^{(1)}(\cdot;\bbeta), \dots, \lambda^{(m)}(\cdot;\bbeta))^\top$ for $\bbeta \in \Theta \subset \R^{p}$. For $\balpha = (\alpha_1, \dots, \alpha_p) \in \{0,1,\dots\}^p$, let $\partial^{\balpha} \lambda^{(j)}$ denote the $\balpha$th-order partial derivative of the $j$th component of $\underline{\lambda}(\cdot;\bbeta)$ with respect to $\bbeta$. For $t \in \N$, define
\begin{equation*}
\big| H_{h \partial^{t} \underline{\lambda}}^{(n)} (\ob) \big| = \sum_{j=1}^{m} \sum_{|\balpha|=t} \big| H_{h \partial^{\balpha} \lambda^{(j)},1}^{(n)} (\ob) \big|.
\end{equation*}

The following lemma provides an upper bound for $|\underline{\widehat{J}}_{h,n}(\ob) - \underline{\widetilde{J}}_{h,n}(\ob)|$.

\begin{lemma} \label{lemma:Jbound}
Suppose that Assumptions \ref{assum:A} and \ref{assum:beta} hold. Suppose further $\lambda^{(j)}(\xb;\bbeta)$ has a bounded $k$th-order partial deriviatives with respect to $\bbeta \in \Theta$. Then, there exists a constant $C \in (0,\infty)$ such that for any $s \in [1,\infty)$,
\begin{equation}
\begin{aligned}
& | \underline{\widehat{J}}_{h,n}(\ob)- \underline{\widetilde{J}}_{h,n}(\ob)|^s \\
& \leq C
\left\{ |D_n|^{1/2}|\widehat{\bbeta}_n - \bbeta_0|\right\}^s \sum_{t=1}^{k-1} \left\{ |D_n|^{-1} \big|H_{h\partial^t \underline{\lambda}}^{(n)}(\ob)\big| \right\}^s + C \left\{ |D_n|^{1/2} |\widehat{\bbeta}_n - \bbeta_0|^{k}\right\}^s.
\end{aligned}
  \label{eq:H-diff1}
\end{equation}
\end{lemma}
\begin{proof}
For simplicity, we will only show the lemma for $m = p = 1$ and $k = 3$. The general case can be treated similarly. Here, $C \in (0,\infty)$ denotes a generic constant that may vary from line to line. The difference between the feasible and ideal DFT is
\begin{equation} \label{eq:Jdiff00}
|\underline{\widehat{J}}_{h,n}(\ob) - \underline{\widetilde{J}}_{h,n}(\ob)| 
= (2\pi)^{-d/2} H_{h,2}^{-1/2} |D_n|^{-1/2} |\widehat{H}_{h\lambda,1}^{(n)}(\ob) - H_{h\lambda_0,1}^{(n)}(\ob)|.
\end{equation}
By using the Taylor expansion of $\lambda(\xb;\bbeta)$ with respect to $\bbeta$, for each $\xb \in [-1/2,1/2]^d$, there exists $\widetilde{\bbeta}_n = \widetilde{\bbeta}_n(\xb)$, a convex combination of $\widehat{\bbeta}_n$ and $\bbeta_0$, such that 
\begin{equation*}
\widehat{\lambda}(\xb) - \lambda_0(\xb) 
= (\widehat{\bbeta}_n - \bbeta_0) \partial \lambda(\xb;\bbeta_0)
+ \frac{1}{2} (\widehat{\bbeta}_n - \bbeta_0)^2 \partial^{2} \lambda(\xb;\bbeta_0) 
+ \frac{1}{6} (\widehat{\bbeta}_n - \bbeta_0)^3 \partial^{3} \lambda(\xb;\widetilde{\bbeta}_n).
\end{equation*}
Substituting the above into $\widehat{H}_{h\lambda,1}^{(n)}(\ob) - H_{h\lambda_0,1}^{(n)}(\ob)$ and applying the triangle inequality, we have
\begin{equation*}
\begin{aligned}
| \widehat{H}_{h\lambda,1}^{(n)}(\ob) - H_{h\lambda_0,1}^{(n)}(\ob) | 
& \leq |\widehat{\bbeta}_n - \bbeta_0| \, \big| H_{h\partial\lambda,1}^{(n)}(\ob)\big| + \frac{1}{2}|\widehat{\bbeta}_n - \bbeta_0|^2 \, \big| H_{h\partial^{2}\lambda,1}^{(n)}(\ob)\big| \\
& \quad + \frac{1}{6}|\widehat{\bbeta}_n - \bbeta_0|^3 \int_{D_n} h(\xb\oslash\aB) \, \big|\partial^{3} \lambda(\xb\oslash\aB;\widetilde{\bbeta}_n)\big| d\xb.
\end{aligned}
\end{equation*}
Since $\Theta$ is compact, and hence bounded, we have $|\widehat{\bbeta}_n - \bbeta_0| \leq \operatorname{diam}(\Theta) < \infty$. 
Moreover, $|\partial^{3} \lambda(\cdot;\bbeta)|$ is bounded. Therefore, $| \widehat{H}_{h\lambda,1}^{(n)}(\ob) - H_{h\lambda_0,1}^{(n)}(\ob) |$ is bounded by
\begin{equation*}
C |\widehat{\bbeta}_n - \bbeta_0| \Big\{ \big| H_{h\partial\lambda,1}^{(n)}(\ob)\big| + \big| H_{h\partial^{2}\lambda,1}^{(n)}(\ob)\big| \Big\}
+ C |D_n| \, |\widehat{\bbeta}_n - \bbeta_0|^3.
\end{equation*}

Substituting the above into (\ref{eq:Jdiff00}) and using the inequality $(a+b+c)^s \leq 3^{s-1} (a^s + b^s + c^s)$ for all $a,b,c>0$, we obtain the desired result for $m = p = 1$ and $k = 3$.
\end{proof}

Using the above lemma, we show the convergence of the $s$th-moment of $|\underline{\widehat{J}}_{h,n}(\ob) - \underline{\widetilde{J}}_{h,n}(\ob)|$.

\begin{corollary} \label{coro:Jbound}
Suppose that Assumptions \ref{assum:A}, \ref{assum:C}, \ref{assum:beta}, and \ref{assum:beta2}(ii) (for some $r \in (1, \infty)$) hold. 
Then, for a sequence $\{\ob_{n}\}$ in $\R^d$ that is asymptotically distant from $\{\textbf{0}\}$, and for any $s \in (r/2, r)$, we have
\begin{equation*}
\Ex \big| \underline{\widehat{J}}_{h,n}(\ob_{n}) - \underline{\widetilde{J}}_{h,n}(\ob_{n}) \big|^s = o(1) \quad \text{as } n \to \infty.
\end{equation*}
\end{corollary}
\begin{proof}
By applying Lemma \ref{lemma:Jbound} with $k=2$, we have
\begin{equation} \label{eq:EJ-bound}
\begin{aligned}
\Ex \big| \underline{\widehat{J}}_{h,n}(\ob_{n}) - \underline{\widetilde{J}}_{h,n}(\ob_{n}) \big|^s 
&\leq C \, \Ex \Big\{ |D_n|^{1/2} |\widehat{\bbeta}_n - \bbeta_0| \Big\}^s 
\Big\{ |D_n|^{-1} \big| H_{h \partial \underline{\lambda}}^{(n)}(\ob_{n}) \big| \Big\}^s \\
&\quad + C \, |D_n|^{s/2} \, \Ex |\widehat{\bbeta}_n - \bbeta_0|^{2s}.
\end{aligned}
\end{equation}
Since $|D_n|^{-1} \big| H_{h \partial \underline{\lambda}}^{(n)}(\ob_{n}) \big| = o(1)$ by YG24, Lemma C.2, and $\Ex \big\{ |D_n|^{1/2} |\widehat{\bbeta}_n - \bbeta_0| \big\}^s = O(1)$ by Assumption \ref{assum:beta2}(ii) (for $r > s$), the first term in \eqref{eq:EJ-bound} is $o(1)$. Moreover, since $2s > r$, the second term is bounded by
$C |D_n|^{s/2} \, \Ex |\widehat{\bbeta}_n - \bbeta_0|^r \leq O(1) \, |D_n|^{s/2 - r/2} = o(1)$.

Altogether, we obtain the desired result.
\end{proof}

Using the above results, in next two theoresm, we show that $|\underline{\widehat{J}}_{h,n}(\ob) - \underline{\widetilde{J}}_{h,n}(\ob)|$ and $|\widehat{I}_{h,n}(\ob) - \widetilde{I}_{h,n}(\ob)|$ converges to zero in probability and in $L_2$.

\begin{theorem} \label{thm:asymp-IRS-fea1}
Suppose that Assumptions \ref{assum:A}, \ref{assum:C}, \ref{assum:beta}, and \ref{assum:beta2}(i) hold. Then, for a sequence $\{\ob_{n}\}$ on $\R^d$ that is asymptotically distant from $\mathbf{0}$, we have
\begin{equation*}
\big| \underline{\widehat{J}}_{h,n}(\ob_n) - \underline{\widetilde{J}}_{h,n}(\ob_n) \big| \Pcon 0 \quad \text{and} \quad
\big| \widehat{I}_{h,n}(\ob_n) - \widetilde{I}_{h,n}(\ob_n) \big| \Pcon 0.
\end{equation*}
\end{theorem}
\begin{proof}
By applying Lemma \ref{lemma:Jbound} with $s=1$ and $k=2$, together with Assumption \ref{assum:beta2}(i), we have
\begin{equation*}
\big| \underline{\widehat{J}}_{h,n}(\ob_n) - \underline{\widetilde{J}}_{h,n}(\ob_n) \big| 
\leq O_p\Big( |D_n|^{-1} \big| H_{h \partial \underline{\lambda}}^{(n)}(\ob_n) \big| \Big) 
+ O_p(|D_n|^{-1/2}).
\end{equation*}
To bound the first term, using YG24, Lemma C.2, we have $|D_n|^{-1} \big| H_{h \partial \underline{\lambda}}^{(n)}(\ob_n) \big| = o(1)$ as $n \rightarrow \infty$. Thus, the first term is $o_p(1)$. Combining this result, we obtain
$\big| \underline{\widehat{J}}_{h,n}(\ob_n) - \underline{\widetilde{J}}_{h,n}(\ob_n) \big| \Pcon 0$.
In turn, $\big| \widehat{I}_{h,n}(\ob_n) - \widetilde{I}_{h,n}(\ob_n) \big| \Pcon 0$ due to the continuous mapping theorem. Thus, we obtain the desired results.
\end{proof}

\begin{theorem} \label{thm:asymp-IRS-fea2}
Suppose that Assumptions \ref{assum:A}, \ref{assum:C}, \ref{assum:beta}, and \ref{assum:beta2}(ii) (for $r > 2$) hold. Then, for a sequence $\{\ob_{n}\}$ on $\R^d$ that is asymptotically distant from $\{\textbf{0}\}$, we have
\begin{equation} \label{eq:Jdiff-MSE}
\lim_{n \to \infty} \Ex| \underline{\widehat{J}}_{h,n}(\ob_n) - \underline{\widetilde{J}}_{h,n}(\ob_n) \big|^2 = 0.
\end{equation}
If we further assume Assumptions \ref{assum:B} (for $k = 4$) and \ref{assum:beta2}(ii) (for $r > 4$) hold and assume the true first-order intensity is $\underline{\lambda}(\xb;\bbeta_0$), then
\begin{equation} \label{eq:Perdiff-MSE}
\lim_{n \to \infty} \Ex \big| \widehat{I}_{h,n}(\ob_n) - \widetilde{I}_{h,n}(\ob_n) \big|^2 = 0.
\end{equation} 
\end{theorem}
\begin{proof}
(\ref{eq:Jdiff-MSE}) follows immediately from Corollary \ref{coro:Jbound} with $s = 2$ and $r > 2$. To show (\ref{eq:Perdiff-MSE}), we first note that
\begin{equation}
\begin{aligned}
\big| \widehat{I}_{h,n}(\ob_n) - \widetilde{I}_{h,n}(\ob_n) \big| &= \big|  \underline{\widehat{J}}_{h,n}(\ob_n)\underline{\widehat{J}}_{h,n}^{*}(\ob_n) - 
\underline{\widetilde{J}}_{h,n}(\ob_n)\underline{\widetilde{J}}_{h,n}^{*}(\ob_n) \big|\\
&\leq \big|  \underline{\widehat{J}}_{h,n}(\ob_n) \{ \underline{\widehat{J}}_{h,n}^{*}(\ob_n) - \underline{\widetilde{J}}_{h,n}^{*}(\ob_n) \} \big| +
\big|  
\{ \underline{\widehat{J}}_{h,n}(\ob_n) - \underline{\widetilde{J}}_{h,n}(\ob_n) \} \underline{\widetilde{J}}_{h,n}^{*}(\ob_n) \big| \\
&\leq  \big| \underline{\widehat{J}}_{h,n}(\ob_n) - \underline{\widetilde{J}}_{h,n}(\ob_n)  \big| \big\{ \big|  \underline{\widehat{J}}_{h,n}(\ob_n) \big| 
+\big|  \underline{\widetilde{J}}_{h,n}(\ob_n) \big| 
\big\} \\
&\leq  \big| \underline{\widehat{J}}_{h,n}(\ob_n) - \underline{\widetilde{J}}_{h,n}(\ob_n)  \big| \big\{ \big|  \underline{\widehat{J}}_{h,n}(\ob_n) - \underline{\widetilde{J}}_{h,n}(\ob_n)   \big| 
+2 \big|  \underline{\widetilde{J}}_{h,n}(\ob_n) \big| 
\big\}
\end{aligned}
\label{eq:Idiff-00}
\end{equation} 
Here, we use the triangle inequality in the first and third inequalities and $|\xb \yb^*| = |\xb|\, |\yb^*|$ and $|\xb^*| = |\xb|$ in the second inequality. 
Taking the square on both sides above and then iteratively applying the Cauchy-Schwarz inequality, we obtain
\begin{equation} \label{eq:Idiff-bound00}
\begin{aligned}
& \Ex \big| \widehat{I}_{h,n}(\ob_n) - \widetilde{I}_{h,n}(\ob_n) \big|^2 \\
&~~\leq C \left\{ \Ex \big| \underline{\widehat{J}}_{h,n}(\ob_n) - \underline{\widetilde{J}}_{h,n}(\ob_n) \big|^4 \right\}^{1/2}
\left\{ \Ex \big| \underline{\widehat{J}}_{h,n}(\ob_n) - \underline{\widetilde{J}}_{h,n}(\ob_n) \big|^4 + \Ex \big| \underline{\widetilde{J}}_{h,n}(\ob_n) \big|^4 \right\}^{1/2}.
\end{aligned}
\end{equation}
The first term above is $o(1)$ due to Corollary \ref{coro:Jbound}. To bound the second term, applying the triangle inequality iteratively yields
\begin{equation} \label{eq:J4}
\Ex \big| \underline{\widetilde{J}}_{h,n}(\ob_n) \big|^4 \leq C \sum_{j=1}^{m} \Ex |\widetilde{J}^{(j)}_{h,n}(\ob_n)|^4 
\leq C \sum_{j=1}^{m} \Big\{ \Ex \big| \widetilde{J}^{(j)}_{h,n}(\ob_n) - {J}^{(j)}_{h,n}(\ob_n) \big|^4 + \Ex \big| {J}^{(j)}_{h,n}(\ob_n) \big|^4 \Big\}.
\end{equation}
Under the correctly specified first-order intensity model, the first term on the right-hand side above is zero. 
Moreover, $\Ex \big| {J}^{(j)}_{h,n}(\ob_n) \big|^4 = \Ex \big| I^{(j,j)}_{h,n}(\ob_n) \big|^2 = O(1)$, provided Assumption \ref{assum:B} holds for $k=4$. 
Therefore, we conclude that $\Ex \big| \underline{J}_{h,n}(\ob_n) \big|^4 = O(1)$ as $n \to \infty$, and in turn, the second term in (\ref{eq:Idiff-bound00}) is $O(1)$. 

Altogether, we have $\Ex \big| \widehat{I}_{h,n}(\ob_n) - \widetilde{I}_{h,n}(\ob_n) \big|^2 = o(1)$ as $n \to \infty$, thus obtain the desired result.
\end{proof}

\subsection{Bounds for the kernel spectral density estimators} \label{sec:KSDE-bound}
Recall $\widehat{F}_{n,\bb}(\ob)$ in (\ref{eq:KSDE}) and $F_{n,\bb}(\ob)$ in (\ref{eq:KSDE-th}).
In this section, we derive bounds between the feasible and theoretical kernel spectral density estimators when the true first-order intensity is $\underline{\lambda}(\xb;\bbeta_0)$. 
In this case, $F_{n,\bb}(\ob)$ coincides with the ``ideal'' kernel spectral density estimator
\begin{equation*}
\widetilde{F}_{n,\bb}(\ob) = \int_{\R^{d}} K_{\bb}(\ob - \xb) \widetilde{I}_{h,n}(\xb) \, d\xb.
\end{equation*}
The case of misspecified first-order intensity will be investigated in the following section. 

In the two theorems below, we first show that $\widehat{F}_{n,\bb}(\ob) - F_{n,\bb}(\ob)$ is asymptotically negligible, and then calculate the $L_2$ convergence rate.

\begin{theorem} \label{thm:KSDE-bound-a}
Suppose that Assumptions \ref{assum:A}, \ref{assum:B} (for $k=4$), \ref{assum:C}, \ref{assum:beta}, \ref{assum:beta2}(i), and \ref{assum:D} hold. Moreover, assume that the true first-order intensity is $\underline{\lambda}(\xb;\bbeta_0)$. Then,
\begin{equation*}
\sup_{\ob \in \R^d}|\widehat{F}_{n,\bb}(\ob) - F_{n,\bb}(\ob)| = o_p(1), \quad n\rightarrow \infty.
\end{equation*}
\end{theorem}
\begin{proof}
By using (\ref{eq:Idiff-00}) together with triangular inequality and Cauchy-Schwarz inequality, we have
\begin{equation}
\begin{aligned}
& |\widehat{F}_{n,\bb}(\ob) -F_{n,\bb}(\ob)| \leq \int_{\R^d} K_{\bb}(\ob - \xb) \big| \widehat{I}_{h,n}(\xb) - I_{h,n}(\xb) \big| \, d\xb \\
& \quad \leq \int_{\R^d} K_{\bb}(\ob - \xb) \big| \underline{\widehat{J}}_{h,n}(\xb) - \underline{J}_{h,n}(\xb) \big| \big\{ 
 \big|  \underline{\widehat{J}}_{h,n}(\xb) - \underline{J}_{h,n}(\xb)   \big| 
+2 \big|  \underline{J}_{h,n}(\xb) \big| 
 \big\} \, d\xb \\
& \quad \leq C \left\{ \int_{\R^d} K_{\bb}(\ob - \xb) \big| \underline{\widehat{J}}_{h,n}(\xb) - \underline{J}_{h,n}(\xb) \big|^2 \, d\xb \right\}^{1/2} \\
& \quad \quad \times \left\{ \int_{\R^d} K_{\bb}(\ob - \xb) \big\{  \big|  \underline{\widehat{J}}_{h,n}(\xb) - \underline{J}_{h,n}(\xb)   \big|^2 
+\big|  \underline{J}_{h,n}(\xb) \big|^2  \big\} \, d\xb \right\}^{1/2}.
\end{aligned}
\label{eq:Fbound11}
\end{equation}
To bound the first term above, since $\underline{J}_{h,n}(\xb) = \underline{\widetilde{J}}_{h,n}(\xb)$, applying Lemma \ref{lemma:Jbound} for $(s,k) = (2,2)$ and obtain
\begin{eqnarray*}
&& \int_{\R^d} K_{\bb}(\ob - \xb) \big| \underline{\widehat{J}}_{h,n}(\xb) - \underline{J}_{h,n}(\xb) \big|^2 \, d\xb \\
&& \leq C \big\{ |D_n|^{1/2}|\widehat{\bbeta}_n - \bbeta_0|\big\}^2 \int_{\R^d} K_{\bb}(\ob - \xb) \left\{ |D_n|^{-1} \big| H_{h\partial \underline{\lambda}}^{(n)}(\xb) \big| \right\}^2 \, d\xb \\
&& \quad + C \left\{ |D_n|^{1/2} \cdot |\widehat{\bbeta}_n - \bbeta_0|^2 \right\}^2 \int_{\R^d} K_{\bb}(\ob - \xb) \, d\xb \\
&& \leq O_p(1) \int_{\R^d} K_{\bb}(\ob - \xb) \left\{ |D_n|^{-1} \big| H_{h\partial \underline{\lambda}}^{(n)}(\xb) \big| \right\}^2 \, d\xb + O_p(|D_n|^{-1}).
\end{eqnarray*}
Here, the second inequality follows from Assumption \ref{assum:beta2}(i) and the identity $\int_{\R^d} K_{\bb}(\ob - \xb) \, d\xb = 1$.

Moreover, since $\sup_{\ob} K_{\bb}(\ob) < C (b_1 \cdots b_d)^{-1}$, this, together with YG24, Lemma C.3(a), yields
\begin{equation}
\begin{aligned}
& \int_{\R^d} K_{\bb}(\ob - \xb) \left\{ |D_n|^{-1} \big| H_{h\partial \underline{\lambda}}^{(n)}(\xb) \big| \right\}^2 \, d\xb \\
& \leq C |D_n|^{-1} (b_1 \cdots b_d)^{-1} \int_{\R^d} \left\{ |D_n|^{-1/2} \big| H_{h\partial \underline{\lambda}}^{(n)}(\xb) \big| \right\}^2 \, d\xb = O(|D_n|^{-1} (b_1 \cdots b_d)^{-1}),
\end{aligned}
\label{eq:KHbound}
\end{equation} where the bound above is uniform over $\ob \in \R^d$.
Therefore, the first term in (\ref{eq:Fbound11}) is \\ $O_p(|D_n|^{-1/2} (b_1 \cdots b_d)^{-1/2}) + O_p(|D_n|^{-1/2})$, which is $o_p(1)$ due to Assumption \ref{assum:D} on the bandwidth.

Next, by using the above result together with the triangle inequality, the square of the second term is bounded by
\begin{equation}
o_p(1) + C \sum_{j=1}^{m} \int_{\R^{d}} K_{\bb}(\ob - \xb) I_{h,n}^{(j,j)}(\xb) d\xb.
\label{eq:KKbound}
\end{equation}
Since $I_{h,n}^{(j,j)}(\xb) = O_p(1)$ uniformly over $\xb \in \R^d$ and $\int_{\R^{d}} K_{\bb}(\ob-\xb) d\xb = 1$, the expression in (\ref{eq:KKbound}) is $O_p(1)$ uniformly over $\ob \in \R^d$. 

Altogether, we conclude that $|\widehat{F}_{n,\bb}(\ob) - F_{n,\bb}(\ob)| = o_p(1)$, and the $o_p(1)$ bound is uniform over $\ob \in \R^d$. Thus we get the desired result.
\end{proof}

\begin{theorem} \label{thm:KSDE-bound-b}
Suppose that Assumptions \ref{assum:A}, \ref{assum:B}(for $k=4$), \ref{assum:C}, \ref{assum:beta}, \ref{assum:beta2}(ii) (for $r>4$), and \ref{assum:D} hold. Moreover, assume that the true first-order intensity is $\underline{\lambda}(\xb;\bbeta_0)$. Then,
\begin{equation*}
\sup_{\ob \in \R^d} \Ex |\widehat{F}_{n,\bb}(\ob) - F_{n,\bb}(\ob)|^2 = O(|D_n|^{-1}(b_1 \cdots b_d)^{-1})+ O(|D_n|^{1-r/4}), \quad n\rightarrow \infty.
\end{equation*}
\end{theorem}
\begin{proof}
By using triangular inequality and Jensen's inequality, we have
\begin{eqnarray*}
|\widehat{F}_{n,\bb}(\ob) - F_{n,\bb}(\ob)|^2 &\leq&
\left( \int_{\R^d} K_{\bb}(\ob - \xb) \big| \widehat{I}_{h,n}(\xb)- I_{h,n}(\xb) \big| d\xb \right)^2 \\
&\leq& \int_{\R^d} K_{\bb}(\ob - \xb) \big| \widehat{I}_{h,n}(\xb)- I_{h,n}(\xb) \big|^2 d\xb.
\end{eqnarray*}
Next, by using the last inequality of (\ref{eq:Idiff-00}) together with Cauchy-Scwarz inequlity, the expectation of the above is bounded by
\begin{equation*}
C \int_{\R^{d}} K_{\bb}(\ob - \xb) \Ex \big| \underline{\widehat{J}}_{h,n}(\xb) - \underline{J}_{h,n}(\xb)  \big|^4 d\xb   \\
+ C
\int_{\R^{d}} K_{\bb}(\ob - \xb) \Ex \big| \underline{\widehat{J}}_{h,n}(\xb) - \underline{J}_{h,n}(\xb)  \big|^2 \big|\underline{J}_{h,n}(\xb)  \big|^2 d\xb.
\end{equation*}
To bound the first term above, we use (\ref{eq:EJ-bound}) for $s=4$ and get
\begin{eqnarray*}
&& \int_{\R^{d}} K_{\bb}(\ob - \xb) \Ex \big| \underline{\widehat{J}}_{h,n}(\xb) - \underline{J}_{h,n}(\xb)  \big|^4 d\xb \\
&& ~\leq C  \int_{\R^{d}} K_{\bb}(\ob - \xb) \big\{ |D_n|^{-1} \big|H_{h\partial \underline{\lambda}}^{(n)}(\xb)\big| \big\}^4 d\xb
+ O(|D_n|^{2-r/2}) \int_{\R^{d}}  K_{\bb}(\ob - \xb) d\xb \\
&&~\leq C \int_{\R^{d}} K_{\bb}(\ob - \xb) \big\{ |D_n|^{-1} \big|H_{h\partial \underline{\lambda}}^{(n)}(\xb)\big| \big\}^2 d\xb
+ O(|D_n|^{2-r/2}) \\
&& ~ = O(|D_n|^{-1}(b_1 \cdots b_d)^{-1}) + O(|D_n|^{2-r/2}).
\end{eqnarray*} 
Here, the second inequality is due to $\sup_{\xb} |D_n|^{-1} \big|H_{h\partial \underline{\lambda}}^{(n)}(\xb)| < \infty$ and the last inequality is due to (\ref{eq:KHbound}).

To bound the second term, using the Cauchy-Schwarz inequality together with $\sup_{\xb} \mathbb{E}\big|  \underline{J}_{h,n}(\xb) \big|^4 = O(1)$, the second term is bounded by
\begin{equation*}
C \int_{\R^d} K_{\bb}(\ob - \xb)  \left\{ \mathbb{E} \big| \underline{\widehat{J}}_{h,n}(\xb) - \underline{J}_{h,n}(\xb)  \big|^4 \right\}^{1/2}  d\xb.
\end{equation*}
Using similar techniques, the above term is bounded by
\begin{eqnarray*}
&& C \int_{\R^d} K_{\bb}(\ob - \xb) \left\{ |D_n|^{-1} \big|H_{h\partial \underline{\lambda}}^{(n)}(\xb)\big| \right\}^2 d\xb 
+ C |D_n| \left\{ \mathbb{E} |\widehat{\bbeta}_n - \bbeta_0|^{r} \right\}^{1/2} \int_{\R^d} K_{\bb}(\ob - \xb) d\xb \\
&& = O(|D_n|^{-1} (b_1 \cdots b_d)^{-1}) + O(|D_n|^{1-r/4}).
\end{eqnarray*}
Combining the above two bounds, and noting that these bounds are uniform over $\ob \in \R^d$, we obtain
\begin{equation*}
\sup_{\ob \in \R^d} \mathbb{E} \big| \widehat{F}_{n,\bb}(\ob) - F_{n,\bb}(\ob) \big|^2
= O(|D_n|^{-1} (b_1 \cdots b_d)^{-1}) + O(|D_n|^{1-r/4}), \quad n \to \infty.
\end{equation*}
Thus, we get the desired results.
\end{proof}

\subsection{Bounds for the quantities under the misspecified intensity model} \label{sec:KSDE-bound2}

In this section, we prove the analogous convergence results in Theorems \ref{thm:KSDE-bound-a} and \ref{thm:KSDE-bound-b} for the case when the first-order intensity function is misspecified. In this setting, the ideal kernel spectral density estimator $\widetilde{F}_{n,\bb}(\ob)$ is no longer equal to its theoretical counterpart $F_{n,\bb}(\ob)$. To obtain the bound, we require additional smoothness assumptions on the data taper $h$ and the first-order intensities $\underline{\lambda}$ and $\underline{\lambda}_0$.

\begin{assumption} \label{assum:smooth}
For any $\balpha = (\alpha_1, \dots, \alpha_d)^\top \in \{0,1,2\}^{d}$, the data taper $h$ has continuous $\balpha$th-order partial derivatives in $\R^{d}$, and for $j \in \{1, \dots, m\}$, both $\lambda^{(j)}(\cdot)$ and $\lambda^{(j)}(\cdot;\bbeta_0)$ have continuous $\balpha$th-order partial derivatives in $\xb \in [-1/2,1/2]^d$.
\end{assumption}
We note that the data taper function $h_{a}(\xb)$ in (\ref{eq:ha-25}) below satisfies Assumption \ref{assum:smooth}. 

The following lemma is useful for obtaining bounds for $|\widehat{F}_{n,\bb}(\ob) - F_{n,\bb}(\ob)|$. For a nonnegative function $g$ with support in $[-1/2,1/2]^d$, recall $H_{g,1}^{(n)}(\ob)$ as defined in (\ref{eq:Hkn}).

\begin{lemma} \label{lemma:Cn}
Let $g(\xb)$ be a nonnegative function on $\R^{d}$ with support $[-1/2,1/2]^d$, and let $\aB(n)$ be the vector of side lengths of $D_n$. Suppose that, for any $\balpha = (\alpha_1, \dots, \alpha_d)^\top \in \{0,1,2\}^{d}$, $g$ has continuous $\balpha$th-order partial derivatives on $\R^{d}$, and that $\aB(n)$ satisfies (SL). Let $i(\ob) = \{j: \omega_j \neq 0\}$ and let $|\ob|_0$ denote the number of nonzero elements of $\ob \in \R^{d}$. Then, there exists a constant $C>0$ such that for any $\ob = (\omega_1, \dots, \omega_d)^\top \in \R^{d} \setminus \{\textbf{0}\}$,
\begin{equation*}
|H_{g,1}^{(n)}(\ob)| \leq C |D_n|^{1-2|\ob|_0/d} \left(\prod_{j \in i(\ob)} \omega_{j}\right)^{-2}.
\end{equation*}
\end{lemma}
\begin{proof}
By a change of variables, we can write $H_{g,1}^{(n)}(\ob) = |D_n| \mathcal{F}(g)(-\aB \odot \ob)$. Since $g$ has continuous $\balpha$th-order partial derivatives on $\R^{d}$ and $\partial^{\balpha} g$ also has support in $[-1/2,1/2]^d$, we can apply \cite{b:fol-99}, Theorem 8.22(e). That is, there exists a constant $C>0$ such that for any $\ob = (\omega_1, \dots, \omega_d)^\top \in \R^{d} \setminus \{\textbf{0}\}$,
\begin{equation*}
|H_{g,1}^{(n)}(\ob)| = |D_n| |\mathcal{F}(g)(-\aB \odot \ob)| \leq C |D_n| \bigg(\prod_{j \in i(\ob)} A_j \omega_{j}\bigg)^{-2}
= C |D_n|^{1-2|\ob|_0/d} \bigg(\prod_{j \in i(\ob)} \omega_{j}\bigg)^{-2}.
\end{equation*}
Here, we use condition (SL) in the last identity. Thus, we get the desired result.
\end{proof}

Using the above bound, we first bound the differences between the ideal and theoretical estimators, $|\widetilde{F}_{n,\bb}(\ob) - F_{n,\bb}(\ob)|$. To this end, for $j \in \{1, \dots, m\}$, let
$g_j(\xb) = \lambda_0^{(j)}(\xb)  - \lambda^{(j)}(\xb)$. Then, by simple algebra, for $i,j \in \{1, \dots, m\}$, we have
\begin{equation}
\begin{aligned}
\widetilde{J}_{h,n}^{(j)}(\ob) - J_{h,n}^{(j)}(\ob) &= (2\pi)^{-d/2}H_{h,2}^{-1/2} |D_n|^{-1/2} H_{hg_j,1}^{(n)}(\ob) := C_{n}^{(j)}(\ob) \\
\widetilde{I}^{(i,j)}(\ob) - I^{(i,j)}(\ob) &= J_{h,n}^{(i)}(\ob) C^{(j)}_{n}(-\ob) + C^{(i)}_n(\ob) J_{h,n}^{(j)}(-\ob) + C^{(i)}_n(\ob) C^{(j)}_{n}(-\ob).
\end{aligned} \label{eq:IJ-diff}
\end{equation}
The following theorem provides the $L_2$-convergence of $|\widetilde{F}_{n,\bb}(\ob) - F_{n,\bb}(\ob)|$ to zero, with rate of convergence.
 
\begin{theorem} \label{theorem:Fij-L2}
Suppose that Assumptions \ref{assum:A}, \ref{assum:B}(for $k=2$), \ref{assum:D}, and \ref{assum:smooth} hold. Furthermore, we asume that the side lengths $\aB(n)$ satisfy (SL). Then, for any $\ob \in \R^{d}$ with $|\ob|_0 > d/4$, 
\begin{equation} \label{eq:E3-1}
\lim_{n\rightarrow \infty} \Ex |\widetilde{F}_{n,\bb}(\ob) - F_{n,\bb}(\ob)|^2 = 0.
\end{equation}
Now, let $d \in \{1,2,3\}$. Then, for any fixed $\delta >0$, we have
\begin{equation} \label{eq:E3-2}
\sup_{\|\ob\|_{\infty} > \delta} \Ex |\widetilde{F}_{n,\bb}(\ob) - F_{n,\bb}(\ob)|^2 = O(|D_n|^{1-4/d}).
\end{equation}
\end{theorem}
\begin{proof}
By using (\ref{eq:IJ-diff}), for $i,j \in \{1, \dots, m\}$, we can write $\widetilde{F}_{n,\bb}^{(i,j)}(\ob) - F_{n,\bb}^{(i,j)}(\ob) = Q_{1} + Q_{2} + Q_{3}$, where
\begin{equation*}
\begin{aligned}
Q_{1} &= \int_{\R^{d}} K_{\bb} (\ob - \xb) J_{h,n}^{(i)}(\xb) C^{(j)}_{n}(-\xb) d\xb, \quad
Q_{2} = \int_{\R^{d}} K_{\bb} (\ob - \xb) C^{(i)}_n(\xb) J_{h,n}^{(j)}(-\xb) d\xb, \\
\text{and} \quad Q_{3} &= \int_{\R^{d}} K_{\bb} (\ob - \xb)  C^{(i)}_n(\xb) C^{(j)}_{n}(-\xb) d\xb.
\end{aligned}
\end{equation*}
The terms $Q_1$ and $Q_2$ are stochastic with mean zero, whereas $Q_3$ is deterministic. We first bound $Q_3$. Note that $K_{\bb}(\ob - \cdot)$
has support $\ob + [-b_1,b_1] \times \cdots \times [-b_d, b_d]$, and since $b_1, \dots, b_d \rightarrow 0$, for large $n \in \N$ we have $|\xb|_0 \geq |\ob|_0$ on the support of $K_{\bb}(\ob - \cdot)$. Moreover, there exists $C>0$, depending only on $\ob$, such that for large $n \in \N$, $\prod_{j \in i(\xb)} x_{j} > C$ for all $\xb \in \ob + [-b_1,b_1] \times \cdots \times [-b_d, b_d]$. Therefore, by using Lemma \ref{lemma:Cn}, we obtain
\begin{equation} \label{eq:Q3}
\begin{aligned}
|Q_{3}| &\leq \int_{\R^{d}} K_{\bb} (\ob- \xb) |C^{(i)}_n(\xb) C^{(j)}_{n}(-\xb)| d\xb \\
&\leq C |D_n|^{1-4|\ob|_0/d} \int_{\R^{d}} K_{\bb} (\ob- \xb) d\xb = O(|D_n|^{1-4|\ob|_0/d}).
\end{aligned}
\end{equation}
Next, to bound the $L_2$-norm of $Q_1$, we have
\begin{equation} \label{eq:Q1}
\begin{aligned}
\Ex |Q_{1}|^2 &\leq \int_{\R^{2d}} K_{\bb} (\ob - \xb) K_{\bb} (\ob - \yb) |C^{(i)}_n(-\xb)| |C^{(j)}_{n}(\yb)| 
|\cov (J_{h,n}^{(j)}(-\xb), J_{h,n}^{(j)}(-\yb))| d\xb d\yb \\
&\leq C |D_n|^{1-4|\ob|_0/d} \int_{\R^{2d}} K_{\bb} (\ob - \xb) K_{\bb} (\ob - \yb) d\xb d\yb
= O(|D_n|^{1-4|\ob|_0/d}).
\end{aligned}
\end{equation}
Here, the second inequality follows from $\sup_{\xb, \yb} |\cov (J_{h,n}^{(j)}(-\xb), J_{h,n}^{(j)}(-\yb))|= O(1)$. Similarly, we have $\Ex |Q_2|^2 = O(|D_n|^{1-4|\ob|_0/d})$. Therefore, combining (\ref{eq:Q3}) and (\ref{eq:Q1}) and summing over $i,j$, we obtain (\ref{eq:E3-1}).

\vspace{0.5em}

To show (\ref{eq:E3-2}), without loss of generality, assume $|\omega_1| > \delta$. Then, by a modification of Lemma \ref{lemma:Cn}, one can show that for large $n \in \N$,
\begin{equation} \label{eq:Cbounds}
|C_{n}^{(j)}(\pm \xb)| = C |D_n|^{-1/2} |H_{hg_j,1}^{(n)}(\pm \xb)| \leq C |D_n|^{1/2} (A_1 x_1)^{-2} \leq C |D_n|^{1/2-2/d} \delta^{-2},
\end{equation}
for $\xb \in \ob + [-b_1,b_1] \times \cdots \times [-b_d, b_d]$.
Therefore, the bounds in (\ref{eq:Q3}) and (\ref{eq:Q1}) are respectively replaced by
\begin{equation*}
\sup_{\|\ob\|_{\infty} > \delta} |Q_3| \leq C \delta^{-2} |D_n|^{1-4/d}, \quad
\sup_{\|\ob\|_{\infty} > \delta} \Ex |Q_j|^2 \leq C \delta^{-2} |D_n|^{1-4/d}, \quad j \in \{1,2\}.
\end{equation*}
Thus, we obtain (\ref{eq:E3-2}). Altogether, we get the desired results.
\end{proof}

Lastly, we derive the convergence of the differences between the feasible and ideal statistics,
 $|\widehat{F}_{n,\bb}(\ob) - \widetilde{F}_{n,\bb}(\ob)|$ to zero, with rates.

\begin{theorem} \label{thm:KSDE-bound-c}
Suppose that Assumptions \ref{assum:A}, \ref{assum:B} (for $k=4$), \ref{assum:beta}, \ref{assum:beta2}(i), \ref{assum:D}, and \ref{assum:smooth} hold.
 Furthermore, we asume that the side lengths $\aB(n)$ satisfy (SL). Then, for any $\ob \in \R^{d}$ with $|\ob|_0 > d/4$, 
\begin{equation} \label{eq:E4-1}
|\widehat{F}_{n,\bb}(\ob) - \widetilde{F}_{n,\bb}(\ob)|= o_p(1), \quad n\rightarrow \infty.
\end{equation} If we further assume Assumption \ref{assum:beta2}(ii) for $r>4$ holds. Then, 
\begin{equation} \label{eq:E4-2}
\lim_{n\rightarrow \infty} \Ex |\widehat{F}_{n,\bb}(\ob) - \widetilde{F}_{n,\bb}(\ob)|^2 =0.
\end{equation}
Lastly, let $d \in \{1,2,3,4\}$. Then, for any fixed $\delta>0$,
\begin{equation} \label{eq:E4-3} 
\sup_{\|\ob\|_{\infty} > \delta} \Ex |\widehat{F}_{n,\bb}(\ob) - \widetilde{F}_{n,\bb}(\ob)|^2 
= O(|D_n|^{-1}(b_1 \cdots b_d)^{-1})+ O(|D_n|^{1-r/4}).
\end{equation}
\end{theorem}
\begin{proof}
First, we show (\ref{eq:E4-1}). By using (\ref{eq:Fbound11})--(\ref{eq:KKbound}), we have
\begin{equation*}
\begin{aligned}
|\widehat{F}_{n,\bb}(\ob) - \widetilde{F}_{n,\bb}(\ob)| 
&\leq o_p(1) \left( o_p(1) + C \sum_{j=1}^{m}\int_{\R^{d}} K_{\bb}(\ob - \xb) \widetilde{I}_{h,n}^{(j,j)}(\xb) d\xb \right) \\
&\leq o_p(1) \left( O_p(1) + C \sum_{j=1}^{m} \int_{\R^{d}} K_{\bb}(\ob - \xb) \{\widetilde{I}_{h,n}^{(j,j)}(\xb) - I_{h,n}^{(j,j)}(\xb)\} d\xb \right).
\end{aligned}
\end{equation*}
Here, the second inequality follows from adding and subtracting the term $\int_{\R^{d}} K_{\bb}(\ob-\xb) I_{h,n}^{(j,j)}(\xb) d\xb$, which is $O_p(1)$. Since the second term in the parentheses is $o_p(1)$ due to Theorem \ref{theorem:Fij-L2}, we obtain (\ref{eq:E4-1}).

Next, to show (\ref{eq:E4-2}), from the proof of Theorem \ref{thm:KSDE-bound-b}, we have
\begin{equation*}
\begin{aligned}
\Ex |\widehat{F}_{n,\bb}(\ob) - \widetilde{F}_{n,\bb}(\ob)|^2
&\leq O(|D_n|^{-1}(b_1 \cdots b_d)^{-1}) + O(|D_n|^{1-r/4}) \\
&\quad + C \int_{\R^{d}} K_{\bb}(\ob - \xb) 
\Ex \big| \underline{\widehat{J}}_{h,n}(\xb) - \underline{\widetilde{J}}_{h,n}(\xb) \big|^2
\big|\underline{\widetilde{J}}_{h,n}(\xb) - J_{h,n}(\xb) \big|^2 d\xb,
\end{aligned}
\end{equation*}
where the bound $O(|D_n|^{-1}(b_1 \cdots b_d)^{-1}) + O(|D_n|^{1-r/4})$ is uniform over $\ob \in \R^d$. Now, let $\|\ob\|_{\infty} > \delta$. By using (\ref{eq:Cbounds}), for large enough $n \in \N$,
\begin{equation*}
\begin{aligned}
& \int_{\R^{d}} K_{\bb}(\ob - \xb) 
\Ex \big| \underline{\widehat{J}}_{h,n}(\xb) - \underline{\widetilde{J}}_{h,n}(\xb) \big|^2
\big|\underline{\widetilde{J}}_{h,n}(\xb) - J_{h,n}(\xb) \big|^2 d\xb \\
&\quad \leq C |D_n|^{1-4/d} \int_{\R^{d}} K_{\bb}(\ob - \xb) 
\Ex \big| \underline{\widehat{J}}_{h,n}(\xb) - \underline{\widetilde{J}}_{h,n}(\xb) \big|^2 \\
&\quad = O(|D_n|^{1-4/d} \times |D_n|^{-1} (b_1 \cdots b_d)^{-1}),
\end{aligned}
\end{equation*}
where the bound is uniform over $\ob$ with $\|\ob\|_{\infty}>\delta$, and the second identity is due to (\ref{eq:KHbound}).

Substituting this into the bound for $\Ex |\widehat{F}_{n,\bb}(\ob) - \widetilde{F}_{n,\bb}(\ob)|^2$, we obtain
\begin{equation*}
\sup_{\|\ob\|_{\infty} > \delta} \Ex |\widehat{F}_{n,\bb}(\ob) - \widetilde{F}_{n,\bb}(\ob)|^2
= O(|D_n|^{-1}(b_1 \cdots b_d)^{-1}) + O(|D_n|^{1-r/4}).
\end{equation*}
This proves (\ref{eq:E4-3}). Altogether, we get the desired results.
\end{proof}

%% file: sim-appendix.tex
\section{Supplement on the simulations} \label{sec:sim-add}

In this section, we provide greater details on the simulation settings in Section \ref{sec:sim} and conduct additional simulations.

\subsection{Details on the simulation settings} \label{sec:sim-detail}

Recall the bivariate Cox process $\underline{X} = (X_1, X_2)$ on the domain $D = [-A/2,A/2]^2$ in Section \ref{sec:dgp}. The corresponding latent intensity field $(\Lambda_1(\xb), \Lambda_2(\xb))^\top$ of $\underline{X}$ is given by
\begin{equation*}
\Lambda_i(\xb) = \lambda^{(i)}(\xb/A) S_i(\xb) Y_i(\xb), \quad \xb \in D.
\end{equation*}
Here, $\lambda^{(i)}(\xb/A) = \Ex[\Lambda_i(\xb)]$ denotes the first-order intensity of $X_i$, in alignment with the asymptotic framework in Definition \ref{def:infill}(ii).

Let $\Phi_1, \Phi_2, \Phi_3$ be independent homogeneous Poisson point processes on $\R^2$ with intensities $\kappa_1, \kappa_2, \kappa_3>0$. The shot-noise field is defined as
\begin{equation*}
S_i(\cdot) = \kappa_i^{-1}\sum_{\vbb \in \Phi_i} \phi_i(\|\cdot - \vbb\|), \quad i \in \{1,2\},
\end{equation*}
where $\phi_i(u)$ denotes the isotropic Gaussian kernel with variance parameter $\sigma_i^2 \in (0,\infty)$ given by
\begin{equation} \label{eq:Gauss-ker}
\phi_i(r) = (2\pi \sigma_i^2)^{-1/2} \exp \left\{-r^2/ (2\sigma_i^2) \right\},  \quad r \in (0, \infty).
\end{equation}
Here, $S_i$ captures the intra-specific clustering induced by $\Phi_i$.

Next, the compound field is defined as
\begin{equation*}
Y_i(\cdot) = \exp{\left(-\sum_{j=1}^{3} \frac{\kappa_{j} \xi_{j \rightarrow i}}{\phi_{j}(0)}\right)}
\prod_{j=1}^{3} \prod_{\vbb \in \Phi_j} \left\{ 1 + \xi_{j \rightarrow i} \frac{\phi_j(\|\cdot - \vbb\|)}{\phi_j(0)}\right\}, \quad i \in \{1,2\}.
\end{equation*}
Here, $\phi_{3}(\cdot)$ is a Gaussian kernel as in (\ref{eq:Gauss-ker}) with variance parameter $\sigma_3^2 \in (0,\infty)$ differs from $\sigma_1^2$ and $\sigma_2^2$. The interaction parameter $\xi_{j \rightarrow i} \in (-1,\infty)$ indicates whether the offspring process $X_i$ is repelled by ($\xi_{j \rightarrow i}<0$) or clustered around ($\xi_{j \rightarrow i}>0$) the parent process $\Phi_{j}$. We set $\xi_{i \rightarrow i} = 0$ since $\Phi_i$ affects $X_i$ through the shot-noise field. Therefore, $Y_i(\cdot)$ combines the effects of all other parent processes on the $i$th offspring.

For simulations, we consider the following three parameter combinations, namely M1--M3. 
\begin{description}
\item
[M1: Homogeneous isotropic process exhibiting inter-species clustering.] 
We set the constant first-order intensities $(\lambda^{(1)}, \lambda^{(2)}) = (0.5, 1.5)$, thus the model is SOS. For the shot-noise fields, we set $(\kappa_1, \kappa_2, \kappa_3) = (0.25, 0.75, 0.2)$ and $(\sigma_1, \sigma_2, \sigma_3) = (0.6, 0.3, 1)$. For the compound fields, we set $(\xi_{1 \rightarrow 2}, \xi_{2 \rightarrow 1}, \xi_{3 \rightarrow 1}, \xi_{3 \rightarrow 2}) = (0.7, 0.9, 0.3, 0.1)$. Since both $\xi_{1 \rightarrow 2}$ and $\xi_{2 \rightarrow 1}$ are positive, the model exhibits clustering between $X_1$ and $X_2$.
	
\item[M2: Inhomogeneous isotropic process exhibiting inter-species clustering.] 
We use the same parameters for the shot-noise and compound fields as in M1, but the first-order intensities of $X_1$ and $X_2$ are given by
\begin{equation*}
\lambda^{(1)}(\xb) = 3\exp \left\{ -2(x_1^2 + x_2^2) \right\}, \quad
\lambda^{(2)}(\xb) = 2\exp\left\{-2(x_1^2 - x_2^2)\right\},
\end{equation*}
$\xb = (x_1, x_2)^\top \in [-1/2, 1/2]^2$. Therefore, the model is inhomogeneous and exhibits clustering between the two processes.
	
\item[M3: Inhomogeneous isotropic process exhibiting inter-species repulsion.]  
The first-order intensities and shot-noise fields are the same as in M2, but the interaction parameters of the compound fields are set to $(\xi_{1 \rightarrow 2}, \xi_{2 \rightarrow 1}, \xi_{3 \rightarrow 1}, \xi_{3 \rightarrow 2}) = (-0.7, -0.9, 0.3, 0.1)$. Since $\xi_{1 \rightarrow 2}, \xi_{2 \rightarrow 1} < 0$, this model exhibits repulsive behavior between $X_1$ and $X_2$.
\end{description}
 A single realization of M1--M3 is shown in Figure \ref{fig:simdata} below.

\begin{figure}[ht!]
\centering
\includegraphics[width=\textwidth]{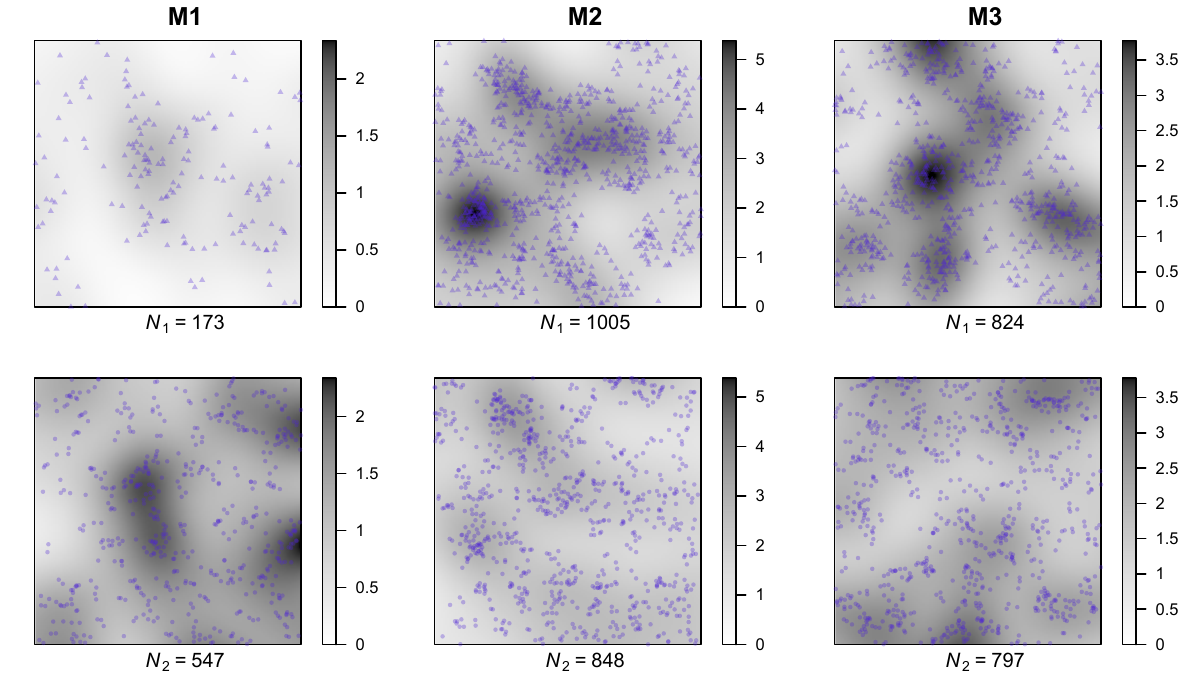}
\caption{
A single realization of M1 (left)–M3 (right) on the domain $[-10, 10]^2$. The triangles (top) indicate the first process, and the crosses (bottom) indicate the second process. $N_1$ and $N_2$ denote the number of points in the first and second processes, respectively, and the grayscale surface shows the kernel estimate of the intensity function for each process.
}
\label{fig:simdata}
\end{figure}

The pseudo-spectrum of each model above is calculated numerically using the closed-form expressions of the marginal and cross PCF (which is equal to $L_2^{(i,j)}(\xb) + 1$) given by \cite{p:jal-15}, Equations (9) and (10): for $i \neq j$,
\begin{equation*}
\begin{aligned}
L_2^{(i,i)}(\xb) &= \bigg\{ 1 + \frac{\phi_i(0)^2}{\kappa_i} c_{\ell}(\xb) \bigg\} \exp \bigg\{ \sum_{\ell} \kappa_{\ell} \xi_{\ell \rightarrow i}^2 c_{\ell}(\xb) \bigg\} - 1 \\
L_2^{(i,j)}(\xb) &=  
\bigg\{ 1 + \xi_{i \rightarrow j} \phi_i(0) c_{i}(\xb) \bigg\} 
\bigg\{ 1 + \xi_{j \rightarrow i} \phi_j(0) c_{j}(\xb) \bigg\}
\exp \bigg\{ \sum_{\ell} \kappa_{\ell} \xi_{\ell \rightarrow i} \xi_{\ell \rightarrow j} c_{\ell}(\xb) \bigg\} - 1,
\end{aligned}
\end{equation*}
where $c_i(\xb) = \phi_i(0)^{-2} \int_{\R^{2}} \phi_i(\xb + \ubb) \phi_i(\ubb) d\ubb$.

Next, we provide greater details on the practical issues for evaluating the periodogram and kernel spectral density as discussed in Section \ref{sec:sim-est}. When computing the DFT $\underline{\mathcal{J}}_{h,n}(\ob)$, we use the separable data taper $h_a(\xb) = h_a(x_1) h_a(x_2)$ with $a = 0.025$, where for $a \in (0, 1/2)$,
\begin{equation} \label{eq:ha-25} 
h_a(x) = 
\begin{cases}
(x+\frac{1}{2})/a - \frac{1}{2\pi} \sin\left(2\pi (x+\frac{1}{2}) / a \right), & -\frac{1}{2} \leq x < -\frac{1}{2}+a, \\
1, & -\frac{1}{2}+a \leq x < \frac{1}{2}-a, \\
h_a(-x), & \frac{1}{2}-a \leq x \leq \frac{1}{2}.
\end{cases}
\end{equation} 
The above data taper satisfies certain smoothness conditions, which also have been used in the simulation studies of YG24.
When calculating the feasible periodogram, the first-order intensities $\lambda^{(i)}$ ($i \in \{1,2\}$) of M2 and M3 are estimated by fitting the log-linear regression model $\lambda(\xb; \bbeta) = \exp(\beta_0 + \beta_1 x_1^2 + \beta_2 x_2^2)$, which corresponds to the correctly specified model. Here, the parameter $\bbeta$ is estimated using the method in \cite{p:waa-07}. For M1, since the model is homogeneous, we fit an intercept-only model by setting $\beta_1 = \beta_2 = 0$.

\subsection{Simulations under misspecified bias} \label{sec:sim-mis}

According to Corollary \ref{coro:KSDE2}, even though the first-order intensity is incorrectly estimated when demeaning the DFT, our kernel spectrum estimator $\widehat{F}_{n,\bb}(\ob)$ still consistently estimates the pseudo-spectrum for frequencies away from the origin. In this section, we demonstrate this result through a set of simulations.

For our data-generating process, we consider models M2 and M3 in Section \ref{sec:sim}, where both models exhibit inhomogeneous first-order intensities ((see, Appendix \ref{sec:sim-detail} for the detailed parameter settings). In contrast, we use the constant first-order intensity model $\underline{\lambda}(\xb;\bbeta) = \bbeta = (\beta_1, \beta_2)^\top \in [0,\infty)^2$ to estimate the bias of the DFT. Therefore, the resulting feasible periodogram is calculated based on the incorrect intensity model. Specifically, we consider two variants of the periodogram estimator: one based on $\widehat{\bbeta}_n = |D_n|^{-1} \underline{N}(D_n)$ (denoted $\widehat{I}_{\text{SOS}}$) and another based on $\widehat{\bbeta}_n = (0,0)^\top$ (denoted $\widehat{I}_{0}$). The former fits the SOS model (although the true model is inhomogeneous), while the latter ignores the debiasing step. Let $\widehat{F}_{\text{SOS}}$ and $\widehat{F}_{0}$ be the kernel smoothed spectral density estimators corresponding to $\widehat{I}_{\text{SOS}}$ and $\widehat{I}_{0}$, respectively. Lastly, we use the same set of 500 replications of point pattern data used in Section \ref{sec:sim}.

\subsection{Results} \label{sec:sim-mis-res}

\noindent \textit{Bandwidth selection results.} \hspace{0.1em}
First, we study how the misspecified model affects the selected bandwidth of its kernel estimator. Here, we use the cross-validation method described in Section \ref{sec:CV} to select the bandwidths of $\widehat{F}_{\text{SOS}}$ and $\widehat{F}_{0}$. However, according to Remark \ref{rmk:near-origin}, both $\widehat{F}_{\text{SOS}}$ and $\widehat{F}_{0}$ may show a large peak at frequencies near the origin. Therefore, to avoid artifacts caused by model misspecification, we evaluate the cross-validated spectral divergence $L(b)$ on $W_{o} = \{\ob: 0.1\pi \leq \|\ob\|_{\infty} \leq 1.5\pi \}$, which omits some low frequency values.

Table \ref{table:simbandapp} below displays the summary statistics of the selected bandwidths from 500 replications. For comparison, we also provide the bandwidth results for the kernel spectral density estimator based on the correctly specified intensity model (denoted $\widehat{F}_{\text{correct}}$) used in Section \ref{sec:sim}. For convenience, we denote the bandwidth for $\widehat{F}_{\text{correct}}$, $\widehat{F}_{\text{SOS}}$, and $\widehat{F}_{0}$ by $b_{\text{correct}}$, $b_{\text{SOS}}$, and $b_{0}$, respectively.

\begin{table}[h]	
	\centering
	\begin{tabular}[t]{ccccccccccc}
		\multicolumn{2}{c}{ } & \multicolumn{3}{c}{$D = [-5,5]^2$} & \multicolumn{3}{c}{$D = [-10,10]^2$} & \multicolumn{3}{c}{$D = [-20,20]^2$} \\
		\cmidrule(l{3pt}r{3pt}){3-5} \cmidrule(l{3pt}r{3pt}){6-8} \cmidrule(l{3pt}r{3pt}){9-11}
		\multirow{-2}{*}{Model} & \multirow{-2}{*}{Statistics} & $b_\text{correct}$ & $b_\text{SOS}$ & $b_0$ & $b_\text{correct}$ & $b_\text{SOS}$ & $b_0$ & $b_\text{correct}$ & $b_\text{SOS}$ & $b_0$\\
		\hline \hline
		& Q1 & 1.06 & 1.06 & 1.06 & 0.54 & 0.55 & 0.55 & 0.28 & 0.28 & 0.28\\
		
		& Q2 & 1.18 & 1.19 & 1.16 & 0.60 & 0.60 & 0.58 & 0.30 & 0.30 & 0.30\\
		
		& Q3 & 1.55 & 1.55 & 1.51 & 0.77 & 0.78 & 0.73 & 0.32 & 0.38 & 0.30\\
		
		\multirow{-4}{*}{M2} & Mean & 1.26 & 1.30 & 1.22 & 0.66 & 0.66 & 0.63 & 0.31 & 0.32 & 0.31\\
		\cmidrule{1-11}
		& Q1 & 1.08 & 1.08 & 1.06 & 0.55 & 0.55 & 0.55 & 0.28 & 0.28 & 0.28\\
		
		& Q2 & 1.18 & 1.21 & 1.13 & 0.60 & 0.60 & 0.58 & 0.30 & 0.30 & 0.30\\
		
		& Q3 & 1.53 & 1.55 & 1.28 & 0.75 & 0.78 & 0.60 & 0.30 & 0.30 & 0.30\\
		
		\multirow{-4}{*}{M3} & Mean & 1.29 & 1.31 & 1.16 & 0.65 & 0.65 & 0.59 & 0.30 & 0.31 & 0.30\\
		\hline
	\end{tabular}
	\caption{The three quartiles and average of cross-validated bandwith for the various kernel smooth estimators.}
	\label{table:simbandapp}
\end{table}
Notice that as $|D|$ increases, both $b_{\text{SOS}}$ and $b_{0}$ tend to decrease and resemble the distribution of $b_{\text{correct}}$. Since the simulation results in Section \ref{sec:sim-results} suggest that $b_{\text{correct}}$ may attain the optimal convergence rate of $|D_n|^{-1/(d+4)}$, the same holds for $b_{\text{SOS}}$ and $b_{0}$. This provides empirical evidence for the theoretical findings in Corollary \ref{coro:MSEb-rate} (see also Remark \ref{rmk:sub-opt}). Moreover, it is intriguing that for all models and observation windows, the value of $b_{\text{SOS}}$ is larger than that of $b_{0}$, while $b_{\text{correct}}$ lies between $b_{\text{SOS}}$ and $b_{0}$. However, theoretical justification of these phenomena remains an open question.

\vspace{0.5em}

\noindent \textit{Estimation accuracy.} \hspace{0.1em}
Moving on, we evaluate the global performance of $\widehat{F}_{\text{SOS}}$ and $\widehat{F}_{0}$ using two metrics, namely IBIAS$^2$ and IMSE, defined as in Section \ref{sec:sim-results}. Table \ref{table:simsummaryapp} summarizes the estimation results. As a benchmark, we also include the accuracy results for $\widehat{F}_{\text{correct}}$. 

{\footnotesize
	\begin{table}[h]
		\centering
		\begin{tabular}[t]{ccccccccc}
			\multicolumn{3}{c}{} & \multicolumn{2}{c}{$\widehat{F}_\text{correct}$} & \multicolumn{2}{c}{$\widehat{F}_\text{SOS}$} & \multicolumn{2}{c}{$\widehat{F}_0$} \\
			\cmidrule(l{3pt}r{3pt}){4-5} \cmidrule(l{3pt}r{3pt}){6-7} \cmidrule(l{3pt}r{3pt}){8-9}
			\multirow{-2}{*}{Model} & \multirow{-2}{*}{\shortstack{Pseudo-\\spectrum}} & \multirow{-2}{*}{Window} & IBIAS$^2$ & IMSE & IBIAS$^2$ & IMSE & IBIAS$^2$ & IMSE\\
			\hline \hline
			&  & $[-5,5]^2$ & 0.01 & 0.38 & 0.01 & 0.37 & 0.10 & 0.55\\
			
			&  & $[-10,10]^2$ & 0.00 & 0.16 & 0.00 & 0.17 & 0.07 & 0.25\\
			
			& \multirow{-3}{*}{$F^{(1,1)}_h$} & $[-20,20]^2$ & 0.00 & 0.13 & 0.00 & 0.13 & 0.01 & 0.14\\
			\cmidrule{2-9}
			&  & $[-5,5]^2$ & 0.00 & 0.24 & 0.00 & 0.23 & 0.29 & 0.64\\
			
			&  & $[-10,10]^2$ & 0.00 & 0.12 & 0.00 & 0.13 & 0.28 & 0.45\\
			
			& \multirow{-3}{*}{$F^{(2,2)}_h$} & $[-20,20]^2$ & 0.00 & 0.11 & 0.00 & 0.11 & 0.08 & 0.20\\
			\cmidrule{2-9}
			&  & $[-5,5]^2$ & 0.02 & 5.48 & 0.02 & 5.19 & 0.48 & 6.16\\
			
			&  & $[-10,10]^2$ & 0.01 & 2.55 & 0.01 & 2.67 & 0.36 & 3.23\\
			
			\multirow{-9}{*}{M2} & \multirow{-3}{*}{$F^{(1,2)}_h$} & $[-20,20]^2$ & 0.00 & 2.14 & 0.01 & 2.13 & 0.07 & 2.27\\
			\cmidrule{1-9}
			&  & $[-5,5]^2$ & 0.01 & 0.31 & 0.01 & 0.32 & 0.11 & 0.51\\
			
			&  & $[-10,10]^2$ & 0.00 & 0.15 & 0.00 & 0.15 & 0.07 & 0.25\\
			
			& \multirow{-3}{*}{$F^{(1,1)}_h$} & $[-20,20]^2$ & 0.00 & 0.13 & 0.00 & 0.13 & 0.01 & 0.14\\
			\cmidrule{2-9}
			&  & $[-5,5]^2$ & 0.01 & 0.19 & 0.00 & 0.20 & 0.36 & 0.66\\
			
			&  & $[-10,10]^2$ & 0.00 & 0.12 & 0.00 & 0.12 & 0.28 & 0.43\\
			
			& \multirow{-3}{*}{$F^{(2,2)}_h$} & $[-20,20]^2$ & 0.00 & 0.11 & 0.00 & 0.11 & 0.07 & 0.19\\
			\cmidrule{2-9}
			&  & $[-5,5]^2$ & 0.10 & 34.83 & 0.10 & 33.57 & 1.25 & 41.72\\
			
			&  & $[-10,10]^2$ & 0.05 & 20.46 & 0.05 & 20.25 & 0.78 & 23.58\\
			
			\multirow{-9}{*}{M3} & \multirow{-3}{*}{$F^{(1,2)}_h$} & $[-20,20]^2$ & 0.03 & 18.50 & 0.03 & 18.38 & 0.13 & 19.40\\
			\bottomrule
		\end{tabular}
		\caption{IBIAS$^2$ and IMSE for the various kernel smooth estimators.}
		\label{table:simsummaryapp}
	\end{table}
}

Analogous to the results in Table \ref{table:simsummary}, the IBIAS$^2$ and IMSE for both $\widehat{F}_{\text{SOS}}$ and $\widehat{F}_{0}$ tend to converge to zero as $|D|$ increases. This verifies the robustness results of our estimator as in Corollary \ref{coro:KSDE2}. Moreover, the IBIAS$^2$ and IMSE for $\widehat{F}_{\text{SOS}}$ are uniformly smaller than the corresponding IBIAS$^2$ and IMSE of $\widehat{F}_{0}$, with the former being a close contender to the benchmark estimator $\widehat{F}_{\text{correct}}$. By definition, $\widehat{I}_{\text{SOS}}$ estimates the bias of the DFT using the constant intensity $\widehat{\lambda}(\xb) = |D_n|^{-1} (N_1(D_n), N_2(D_n))^\top$, which is a $|D_n|^{1/2}$-consistent estimator of $(\int_{[-1/2,1/2]^2} \lambda^{(1)}(\xb) d\xb , \int_{[-1/2,1/2]^2} \lambda^{(2)}(\xb) d\xb )^{\top}$. Therefore, under the SOS scheme, the mean effect of the true spatially varying intensity function $\underline{\lambda}(\cdot)$ is removed when calculating the centered DFT. In contrast, there is no bias correction when calculating $\widehat{I}_{0}$. Consequently, the simulation results in Table \ref{table:simsummaryapp} suggest that even a simple model for $\underline{\lambda}(\cdot)$ yields an improvement in the pseudo-spectrum estimator. The theoretical justification of this observations will be a good avenue for future research.

\section{Supplement on the real data applications} \label{sec:bci-add}

\subsection{Details on the real data analysis} \label{sec:DA}

In Section \ref{sec:data}, we consider the point patterns of five tree species in the tropical forest of Barro Colorado Island (BCI), namely \textit{Capparis frondosa} ($X_1$, $n=3{,}112$); \textit{Hirtella triandra} ($X_2$, $n=4{,}552$); \textit{Protium panamense} ($X_3$, $n=3{,}119$); \textit{Protium tenuifolium} ($X_4$, $n=3{,}091$); and \textit{Tetragastris panamensis} ($X_5$, $n=4{,}961$). The point patterns of these species are displayed in the top panels of Figure \ref{fig:bci_lambda}.

\begin{figure}[ht!]
	\centering
	\includegraphics[width=1.02\textwidth]{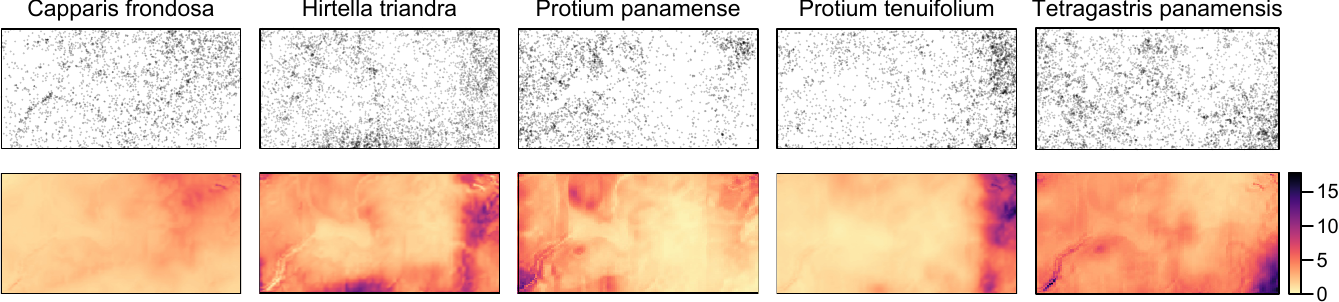}
	\caption{Top: Point patterns of five tree species. Bottom: Heatmap of the estimated first-order intensity functions using the log-linear regression model.}
	\label{fig:bci_lambda}
\end{figure}

When calculating the periodogram, we fit the log-linear first-order intensity model using ten covariates: three geolocational covariates ($x$ coordinates, $y$ coordinates, and their interaction term) and seven environmental covariates (including elevation, gradient, and soil mineral contents) provided in \cite{p:jal-15}, Supporting Information. The estimated first-order intensity functions are displayed in the bottom panels of Figure \ref{fig:bci_lambda}. A comparison between the actual point patterns (top panels) and the fitted intensities (bottom panels) indicates that the use of the log-linear model with covariate information effectively captures the first-order behavior of the point patterns. 

To consolidate this argument, we conduct a (rather informal) frequency domain goodness of fit check of the first-order intensity model. As discussed in Remark \ref{rmk:near-origin}, the behavior of the feasible periodogram provides information on whether the fitted intensity reflects the true intensity or not. Specifically, if the fitted intensity $\underline{\lambda}(\cdot; \widehat{\bbeta}_n)$ consistently estimates the true intensity $\underline{\lambda}(\cdot)$, then $\widehat{I}_{h,n}(\ob) = O_p(1)$ for any $\ob \in \R^d$. In contrast, if $\underline{\lambda}(\cdot; \widehat{\bbeta}_n)$ differs substantially from the true intensity, then $\widehat{I}_{h,n}(\ob)$ increases steeply as $\ob$ approaches to $\mathbf{0}$. 

Figure \ref{fig:per-BCI} below plots the radial averages of the feasible periodograms calculated from two different first-order intensity models: the log-linear model with ten covariates (top panels; denotes $\widehat{I}_{\text{LL}}$) and the constant intensity model (bottom panels; denotes $\widehat{I}_{\text{SOS}}$) as detailed in Appendix \ref{sec:sim-mis}.
\begin{figure}[ht!]
	\centering
	\includegraphics[width=\textwidth]{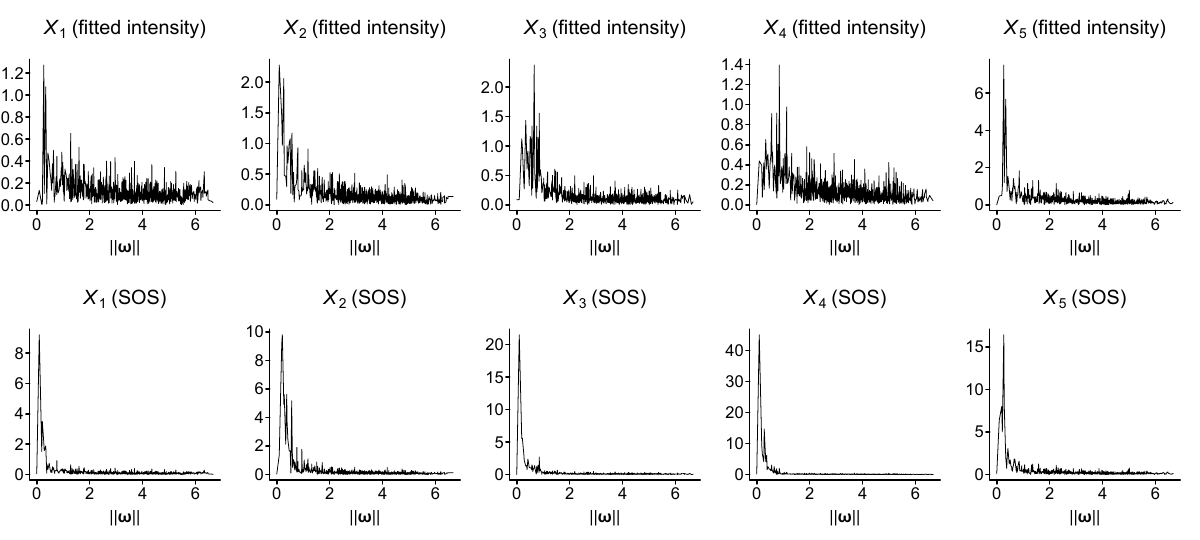}
	\caption{The radial average of the two peridograms of five spices in the BCI dataset. Top: periodogram calculated using the log-linear first-order intensity model with ten covariates. Bottom: periodogram calculated the constant intensity model.}
	\label{fig:per-BCI}
\end{figure}
For species $X_1$ through $X_4$, $\widehat{I}_{\text{LL}}$ appear to be bounded above, whereas $\widehat{I}_{\text{SOS}}$ exhibit pronounced peaks near the origin. These observations imply that the point patterns $X_1$ through $X_4$ are inhomogeneous and can be well explained by the log-linear regression model. In contrast, for $X_5$, both periodograms show a large peak near the origin, but the magnitude of the peak for $\widehat{I}_{\text{SOS}}$ is much larger than that of $\widehat{I}_{\text{LL}}$. This indicates that although the ten covariates may not be sufficient to fully explain the underlying first-order intensity of $X_5$, the inhomogeneous log-linear model still provides a more adequate description than the homogeneous model.


\subsection{Coherence and partial coherence analysis for SOIRS process} \label{sec:coh}
One of the important applications of the spectrum of a multivariate stationary stochastic process is coherence and partial coherence analysis. Suppose that $\underline{X} = (X_1, \dots, X_m)$ ($m \geq 3$) is an SOS spatial point process on $\R^d$. For $a \neq b$, we say that $X_a$ and $X_b$ are uncorrelated if $\cov(N_a(A), N_b(B)) = 0$ for any $A, B \in \mathcal{B}(\R^d)$. To define the partial uncorrelatedness, for $a \neq b \in \{1, \dots, m\}$, let  
\begin{equation*}
N_{a|-(a,b)}(A) = \int_{A} \left( \mu_{a} + \sum_{i \neq a,b} \int_{\R^{d}} \psi_{a}^{(i)}(\yb-\xb) \, N_{i}(d\xb) \right) d\yb
\end{equation*}
be the optimal linear predictor of $N_{a}(A)$ based on $N_{-(a,b)}(\cdot) = \{N_{i}(\cdot) : i \neq a,b\}$. Here, the coefficient $\mu_{a}$ and the smooth functions $\psi_{a}^{(i)}(\cdot)$ are determined by minimizing the mean squared prediction error $\Ex \left|N_{a}(A) - N_{a|-(a,b)}(A) \right|^2$ over all possible $\mu_{a}$ and $\psi_{a}^{(i)}(\cdot)$ (cf. \cite{b:dal-03}, Section 8.5). 
Using this notation, we say $X_{a}$ and $X_b$ are partially uncorrelated if  $\cov(N_{a}(A) - N_{a|-(a,b)}(A),  N_b(B) - N_{b|-(a,b)}(B)) = 0$
for any $A, B \in \mathcal{B}(\R^d)$. See \cite{p:dah-00b, p:eic-12} for the time series case and \cite{p:eck-16} for an extension to spatial point processes.  

\cite{p:dah-00b} showed that the spectrum and its inverse contain information about the uncorrelatedness and partial uncorrelatedness relationships of multivariate time series. Specifically, let  
\begin{equation} \label{eq:RD}
R^{(a,b)}(\ob) = \frac{|F^{(a,b)}(\ob)|}{\sqrt{F^{(a,a)}(\ob) F^{(b,b)}(\ob)}} 
\quad \text{and} \quad
D^{(a,b)}(\ob) = \frac{|F^{-(a,b)}(\ob)|}{\sqrt{F^{-(a,a)}(\ob) F^{-(b,b)}(\ob)}}
\end{equation} 
be the (magnitude) coherence and squared partial coherence of the two point processes $X_a$ and $X_b$, respectively, where $F^{-(a,b)}$ is the $(a,b)$-th element of the inverse spectrum $F(\ob)^{-1}$. Then, by extending the results of \cite{p:dah-00b} to the SOS spatial point process framework, it follows that both $R^{(a,b)}(\ob)$ and $D^{(a,b)}(\ob)$ take values in $[0,1]$, and that $X_a$ and $X_b$ are uncorrelated (resp., partially uncorrelated) if and only if $R^{(a,b)}(\ob) = 0$ (resp., $D^{(a,b)}(\ob) = 0$) for all $\ob \in \R^d$.

Now, we extend the coherence and partial coherence analysis to inhomogeneous spatial point processes. However, for the SOIRS process, it is unwieldy to represent the uncorrelatedness and partial uncorrelatedness relationships in the frequency domain. Therefore, we define new types of correlation structures for multivariate inhomogeneous processes. Recall the intensity reweighted process $\underline{\widetilde{X}} = (\widetilde{X}_1, \dots, \widetilde{X}_m)$ of $\underline{X}$ as in Definition \ref{def:k-IRS}.  

\begin{definition} \label{defin:uncor}
Let $\underline{X} = (X_1, \dots, X_m)$ be an $m$-variate SOIRS process. Then, for $a \neq b$, we say that $X_a$ and $X_b$ are uncorrelated after reweighting (u.a.r.) if the two SOS processes $\widetilde{X}_a$ and $\widetilde{X}_b$ are uncorrelated. Moreover, we say that $X_a$ and $X_b$ are partially uncorrelated after reweighting (partial u.a.r.) if $\widetilde{X}_a$ and $\widetilde{X}_b$ are partially uncorrelated.
\end{definition}

To represent the u.a.r. and partial u.a.r. relationships in the frequency domain, recall from Appendix \ref{sec:K-function} that the spectrum of $\underline{\widetilde{X}}$ is given by  
\begin{equation} \label{eq:Ftilde}
\widetilde{F}(\ob) = (2\pi)^{-d} I_{m} + \mathcal{F}^{-1}(L_2)(\ob).
\end{equation}
Let $\widetilde{R}^{(a,b)}(\ob)$ and $\widetilde{D}^{(a,b)}(\ob)$ be defined analogously to (\ref{eq:RD}), but with $F$ replaced by $\widetilde{F}$. Then, it follows immediately that  
\begin{equation}\label{eq:UAR} 
\begin{aligned}
& X_a \text{ and } X_b \text{ are u.a.r. if and only if } \widetilde{R}^{(a,b)}(\ob) \equiv 0, \\ 
& X_a \text{ and } X_b \text{ are partial u.a.r. if and only if } \widetilde{D}^{(a,b)}(\ob) \equiv 0.
\end{aligned}
\end{equation}

Next, we estimate $\widetilde{R}^{(a,b)}(\ob)$ and $\widetilde{D}^{(a,b)}(\ob)$ based on the observed point pattern on $D_n$, which follows the asymptotic framework in Definition \ref{def:infill}. To this end, we first estimate the inverse Fourier transform of $L_2$. Recall the pseudo-spectrum $F_h$ in (\ref{eq:f-IRS-mat}). By simple algebra, we have
\begin{equation*}
\mathcal{F}^{-1}(L_2)(\ob) = H_{h,2} 
\left( F_{h}(\ob) - (2\pi)^{-d} H_{h,2}^{-1} \diag (H_{h^2 \underline{\lambda}}) \right)
\oslash H_{h^2 \underline{\lambda}\cdot\underline{\lambda}^\top}.
\end{equation*}
Therefore, a natural estimator of $\mathcal{F}^{-1}(L_2)(\ob)$ is  
\begin{equation*}
\widehat{\mathcal{F}^{-1}(L_2)}(\ob) = H_{h,2} 
\left( \widehat{F}_{n,\bb}(\ob) - (2\pi)^{-d} H_{h,2}^{-1} \diag (H_{h^2 \underline{\widehat{\lambda}}}) \right)
\oslash H_{h^2 \underline{\widehat{\lambda}}\cdot\underline{\widehat{\lambda}}^\top},
\end{equation*}
where $\widehat{F}_{n,\bb}(\ob)$ is the kernel spectral density estimator and $\widehat{\underline{\lambda}}$ is the parametric estimator of the true intensity $\underline{\lambda}$. By substituting $\widehat{\mathcal{F}^{-1}(L_2)}(\ob)$ into (\ref{eq:Ftilde}) and applying the relationship in (\ref{eq:RD}), we obtain the estimators of $\widetilde{R}^{(a,b)}(\ob)$ and $\widetilde{D}^{(a,b)}(\ob)$, denoted by $\widehat{R}^{(a,b)}(\ob)$ and $\widehat{D}^{(a,b)}(\ob)$, respectively.

For the BCI dataset considered in Section \ref{sec:data}, we evaluate $\widehat{R}^{(a,b)}(\ob)$ and $\widehat{D}^{(a,b)}(\ob)$ on a uniform grid in $[-1.5\pi, 1.5\pi]^2$, where the grid size is greater than twice the bandwidth, $b_{\text{CV}} = 0.62$. Consequently, $\widehat{R}^{(a,b)}(\ob)$ and $\widehat{D}^{(a,b)}(\ob)$ are evaluated at 49 different frequencies. For each pair $(a,b)$, define  
\begin{equation} \label{eq:RD-max}
\widehat{R}^{(a,b)} = \max \{\widehat{R}^{(a,b)}(\ob)\}^2 \quad \text{and} \quad
\widehat{D}^{(a,b)} = \max \{\widehat{D}^{(a,b)}(\ob)\}^2,
\end{equation}
where the maximum is taken over the 49 frequencies described above. 

In Figure \ref{fig:bci_coherence} below, we plot the squared norm of the coherence estimator (upper diagonal) and partial coherence estimator (lower diagonal) against $\|\ob\|$, where the numbers in plots represent the corresponding maximum values (either $\widehat{R}^{(a,b)}$ or $\widehat{D}^{(a,b)}$).


\begin{figure}[ht!]
	\centering
	\includegraphics[width=\textwidth]{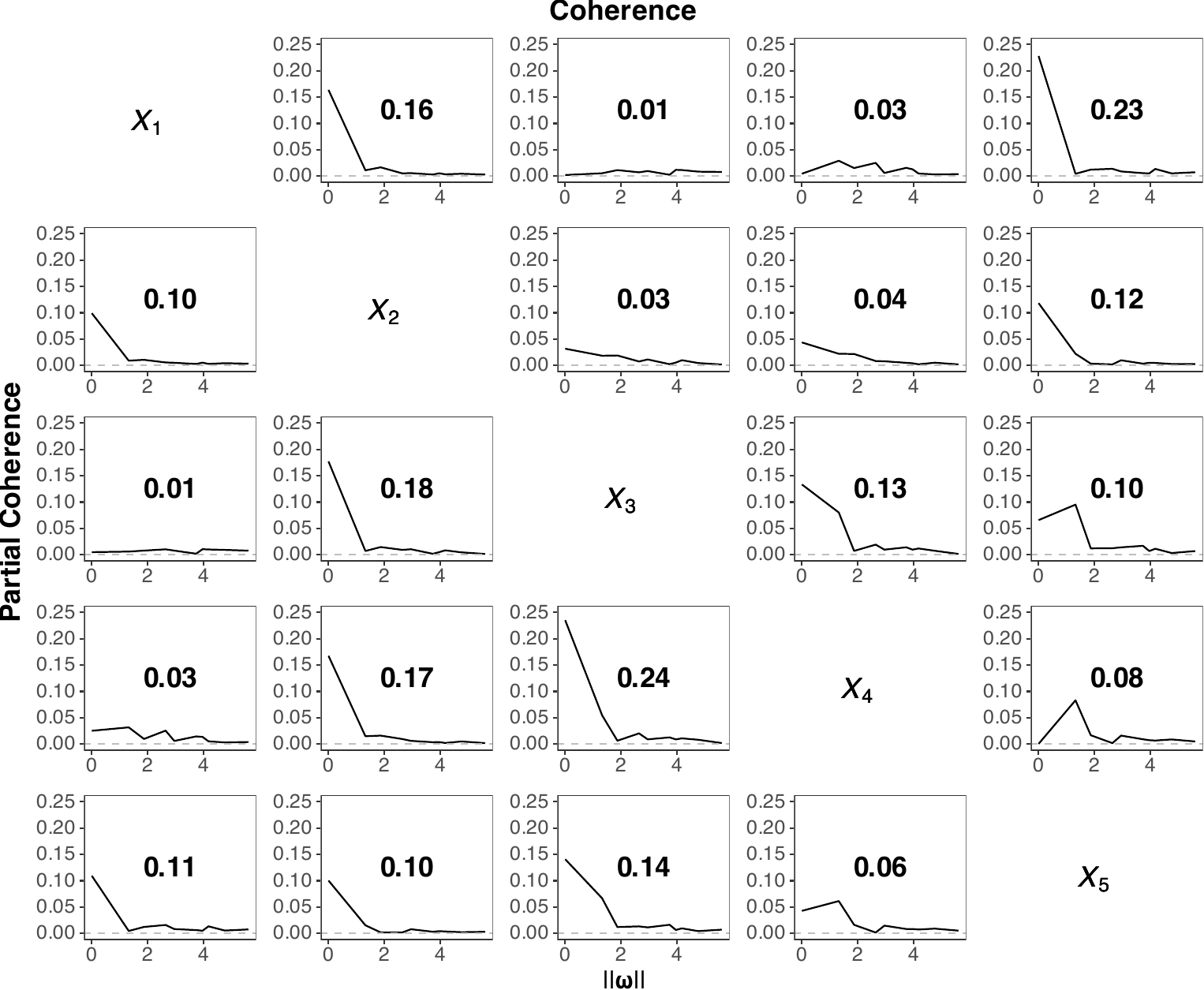}
	\caption{The squared coherence (upper diagonal) and partial coherence (lower diagonal) estimator of the BCI dataset plotted against $\|\ob\|$. The numer in plots represent the corresponding maximum values.}
	\label{fig:bci_coherence}
\end{figure}

From (\ref{eq:UAR}), the closer $\widehat{R}^{(a,b)}$ (resp., $\widehat{D}^{(a,b)}$) is to zero, the stronger the evidence that $X_a$ and $X_b$ are u.a.r. (resp., partial u.a.r.). However, the asymptotic behavior of $\widehat{R}^{(a,b)}(\ob)$ and $\widehat{D}^{(a,b)}(\ob)$ under the SOIRS framework is not yet available. Therefore, computing the $\alpha$-level confidence regions for $\widehat{R}^{(a,b)}$ and $\widehat{D}^{(a,b)}$ in (\ref{eq:RD-max}) under the null hypotheses of u.a.r. and partial u.a.r. is beyond the scope of this article (also be a promising direction for future research).

%% file: spectrum-SPP-v7.bbl
\begin{thebibliography}{32}
\providecommand{\natexlab}[1]{#1}
\providecommand{\url}[1]{\texttt{#1}}
\expandafter\ifx\csname urlstyle\endcsname\relax
  \providecommand{\doi}[1]{doi: #1}\else
  \providecommand{\doi}{doi: \begingroup \urlstyle{rm}\Url}\fi

\bibitem[Baddeley et~al.(2000)Baddeley, M{\o}ller, and Waagepetersen]{p:bad-00}
A.~J. Baddeley, J.~M{\o}ller, and R.~Waagepetersen.
\newblock Non-and semi-parametric estimation of interaction in inhomogeneous
  point patterns.
\newblock \emph{Stat. Neerl.}, 54\penalty0 (3):\penalty0 329--350, 2000.

\bibitem[Bartlett(1964)]{p:bar-64}
M.~S. Bartlett.
\newblock The spectral analysis of two-dimensional point processes.
\newblock \emph{Biometrika}, 51\penalty0 (3/4):\penalty0 299--311, 1964.

\bibitem[Beltr\~{a}o and Bloomfield(1987)]{p:bel-87}
K.~I. Beltr\~{a}o and P.~Bloomfield.
\newblock Determining the bandwidth of a kernel spectrum estimate.
\newblock \emph{J. Time Series Anal.}, 8\penalty0 (1):\penalty0 21--38, 1987.

\bibitem[Brillinger(1981)]{b:bri-81}
D.~R. Brillinger.
\newblock \emph{Time series: {D}ata {A}nalysis and {T}heory}.
\newblock Holden-Day, INC., San Francisco, CA, 1981.
\newblock Expanded edition.

\bibitem[Brockwell and Davis(2006)]{b:bro-dav-06}
P.~J. Brockwell and R.~A. Davis.
\newblock \emph{Time series: theory and methods}.
\newblock Springer Series in Statistics. Springer, New York, 2006.
\newblock Reprint of the second (1991) edition.

\bibitem[Choiruddin et~al.(2021)Choiruddin, Coeurjolly, and
  Waagepetersen]{p:cho-21}
A.~Choiruddin, J.-F. Coeurjolly, and R.~Waagepetersen.
\newblock Information criteria for inhomogeneous spatial point processes.
\newblock \emph{Aust. N. Z. J. Stat.}, 63\penalty0 (1):\penalty0 119--143,
  2021.

\bibitem[Condit et~al.(2019)Condit, P{\'e}rez, Aguilar, Lao, Foster, and
  Hubbell]{m:cond-19}
R.~Condit, R.~P{\'e}rez, S.~Aguilar, S.~Lao, R.~Foster, and S.~Hubbell.
\newblock Complete data from the barro colorado 50-ha plot: 423617 trees, 35
  years, 2019.

\bibitem[Cressie(2015)]{b:cre-15}
Noel Cressie.
\newblock \emph{Statistics for spatial data}.
\newblock John Wiley \& Sons, Hoboken, NJ, 2015.

\bibitem[Dahlhaus(1997)]{p:dah-97}
R.~Dahlhaus.
\newblock Fitting time series models to nonstationary processes.
\newblock \emph{Ann. Statist.}, 25\penalty0 (1):\penalty0 1--37, 1997.

\bibitem[Dahlhaus(2000)]{p:dah-00b}
R.~Dahlhaus.
\newblock Graphical interaction models for multivariate time series.
\newblock \emph{Metrika}, 51:\penalty0 157--172, 2000.

\bibitem[Dahlhaus and K\"{u}nsch(1987)]{p:dah-87}
R.~Dahlhaus and H.~K\"{u}nsch.
\newblock Edge effects and efficient parameter estimation for stationary random
  fields.
\newblock \emph{Biometrika}, 74\penalty0 (4):\penalty0 877--882, 1987.

\bibitem[Daley and Vere-Jones(2003)]{b:dal-03}
D.~J. Daley and D.~Vere-Jones.
\newblock \emph{An {I}ntroduction to the {T}heory of {P}oint {P}rocesses:
  {V}olume {I}: {E}lementary {T}heory and {M}ethods}.
\newblock Springer, New York City, NY., 2003.
\newblock second edition.

\bibitem[Eckardt(2016)]{p:eck-16}
M.~Eckardt.
\newblock Graphical modelling of multivariate spatial point processes.
\newblock \emph{arXiv preprint arXiv:1607.07083}, 2016.

\bibitem[Eichler(2012)]{p:eic-12}
M.~Eichler.
\newblock Graphical modelling of multivariate time series.
\newblock \emph{Probab. Theory Related Fields}, 153:\penalty0 233--268, 2012.

\bibitem[Folland(1999)]{b:fol-99}
G.~B. Folland.
\newblock \emph{Real analysis: modern techniques and their applications},
  volume~40.
\newblock John Wiley \& Sons, New York, NY, 1999.
\newblock 2nd edition.

\bibitem[Grainger et~al.(2023)Grainger, Rajala, Murrell, and Olhede]{p:gra-23}
J.~P. Grainger, T.~A. Rajala, D.~J. Murrell, and S.~C. Olhede.
\newblock Spectral estimation for spatial point processes and random fields.
\newblock \emph{arXiv preprint arXiv:2312.10176}, 2023.

\bibitem[Guan and Loh(2007)]{p:gau-07}
Y.~Guan and J.~M. Loh.
\newblock A thinned block bootstrap variance estimation procedure for
  inhomogeneous spatial point patterns.
\newblock \emph{J. Amer. Statist. Assoc.}, 102\penalty0 (480):\penalty0
  1377--1386, 2007.

\bibitem[Hall and Patil(1994)]{p:hal-94}
P.~Hall and P.~Patil.
\newblock Properties of nonparametric estimators of autocovariance for
  stationary random fields.
\newblock \emph{Probab. Theory Related Fields}, 99:\penalty0 399--424, 1994.

\bibitem[Jalilian et~al.(2015)Jalilian, Guan, Mateu, and
  Waagepetersen]{p:jal-15}
A.~Jalilian, Y.~Guan, J.~Mateu, and R.~Waagepetersen.
\newblock Multivariate product‐shot‐noise {C}ox point process models.
\newblock \emph{Biometrics}, 71:\penalty0 1022--1033, 2015.

\bibitem[Kurisu(2022)]{p:kur-22}
D.~Kurisu.
\newblock Nonparametric regression for locally stationary random fields under
  stochastic sampling design.
\newblock \emph{Bernoulli}, 28\penalty0 (2):\penalty0 1250--1275, 2022.

\bibitem[Lahiri et~al.(2025)Lahiri, McElroy, and Weinberg]{p:lah-24}
S.~Lahiri, T.~McElroy, and D.~Weinberg.
\newblock Locally stationary spatial processes.
\newblock \emph{Sankhya A}, 2025.

\bibitem[Lahiri(2003)]{p:lah-03}
S.~N. Lahiri.
\newblock Central limit theorems for weighted sums of a spatial process under a
  class of stochastic and fixed designs.
\newblock \emph{Sankhya A}, 65\penalty0 (2):\penalty0 356--388, 2003.

\bibitem[Mammen and Mueller(2023)]{p:mam-23}
E.~Mammen and M.~Mueller.
\newblock Nonparametric estimation of locally stationary hawkes processes.
\newblock \emph{Bernoulli}, 29\penalty0 (3):\penalty0 2062--2083, 2023.

\bibitem[Matsuda and Yajima(2009)]{p:mat-09}
Y.~Matsuda and Y.~Yajima.
\newblock Fourier analysis of irregularly spaced data on {$R^d$}.
\newblock \emph{J. R. Stat. Soc. Ser. B. Stat. Methodol.}, 71\penalty0
  (1):\penalty0 191--217, 2009.

\bibitem[Rajala et~al.(2023)Rajala, Olhede, Grainger, and Murrell]{p:raj-23}
T.~A. Rajala, S.~C. Olhede, J.~P. Grainger, and D.~J. Murrell.
\newblock What is the {F}ourier transform of a spatial point process?
\newblock \emph{IEEE Trans. Inform. Theory}, 2023.

\bibitem[Roueff et~al.(2016)Roueff, Von~Sachs, and Sansonnet]{p:rou-16}
F.~Roueff, R.~Von~Sachs, and L.~Sansonnet.
\newblock Locally stationary {H}awkes processes.
\newblock \emph{Stochastic Process. Appl.}, 126\penalty0 (6):\penalty0
  1710--1743, 2016.

\bibitem[Waagepetersen(2008)]{p:waa-08}
R.~Waagepetersen.
\newblock Estimating functions for inhomogeneous spatial point processes with
  incomplete covariate data.
\newblock \emph{Biometrika}, 95\penalty0 (2):\penalty0 351--363, 2008.

\bibitem[Waagepetersen and Guan(2009)]{p:waa-09}
R.~Waagepetersen and Y.~Guan.
\newblock Two-step estimation for inhomogeneous spatial point processes.
\newblock \emph{J. R. Stat. Soc. Ser. B. Stat. Methodol.}, 71\penalty0
  (3):\penalty0 685--702, 2009.

\bibitem[Waagepetersen et~al.(2016)Waagepetersen, Guan, and Mateu]{p:waa-16}
R.~Waagepetersen, A.~Guan, Y.and~Jalilian, and J.~Mateu.
\newblock Analysis of multispecies point patterns by using multivariate
  log-{G}aussian {C}ox processes.
\newblock \emph{J. R. Stat. Soc. Ser. C. Appl. Stat.}, 65\penalty0
  (1):\penalty0 77--96, 2016.

\bibitem[Waagepetersen(2007)]{p:waa-07}
R.~P. Waagepetersen.
\newblock An estimating function approach to inference for inhomogeneous
  {N}eyman--{S}cott processes.
\newblock \emph{Biometrics}, 63\penalty0 (1):\penalty0 252--258, 2007.

\bibitem[Yang and Guan(2024)]{p:yan-24}
J.~Yang and Y.~Guan.
\newblock Fourier analysis of spatial point processes.
\newblock \emph{To appear at Bernoulli}, 2024.

\bibitem[Zhu et~al.(2025)Zhu, Yang, Jun, and Cook]{p:zhu-25}
L.~Zhu, J.~Yang, M.~Jun, and S.~Cook.
\newblock On minimum contrast method for multivariate spatial point processes.
\newblock \emph{Electron. J. Stat.}, 19\penalty0 (1):\penalty0 1889--1941,
  2025.

\end{thebibliography}
